\documentclass[utf-8]{article}
\usepackage{amsfonts,amsmath,amsthm,amssymb,authblk,latexsym,bm,mathrsfs}
\usepackage{hyperref} \hypersetup{colorlinks=true} 

\theoremstyle{definition}
\newtheorem{exer}{Exercise}
\def\bex{\small\begin{exer}}
\def\eex{\end{exer}\normalsize}

\def\hint#1{~\big({\sf Hint}: {\footnotesize #1}\big)}

\theoremstyle{plain}
\newtheorem{thm}{Theorem}[section]
\newtheorem{theorem}[thm]{Theorem}
\newtheorem{corollary}[thm]{Corollary}

\newtheorem{lemma}[thm]{Lemma}
\newtheorem{definition}[thm]{Definition}
\newtheorem{example}[thm]{Example}

\newtheorem{remark}[thm]{Remark}

\numberwithin{equation}{section} \numberwithin{figure}{section}
\def\F{{\mathbb F}}
\def\d{{\rm d}}\def\w{{\rm w}} \def\R{{\rm R}} \def\E{{\rm E}}

\begin{document}
\title{Asymptotic Properties of Quasi-Group Codes}
\author{
Yun Fan, ~~ Liren Lin
}
\date{}
\maketitle

\tableofcontents

\section{Introduction}

Let $\F$ be a finite field with cardinality $q$, and $n$ be a positive integer. 
In coding theory, any sequence $a=(a_1,\cdots,a_n)\in\F^n$
is said to be a {\em word} of length~$n$ over~$\F$; 
and the number of the indexes $i$ that $a_i\ne 0$, denoted by $\w(a)$, 
is called the {\em weight} of $a$.  
For $a,b\in\F^n$, the {\em distance}
between $a,b$ is defined as $\d(a,b):=\w(a-b)$. 

Any subspace $C$ of $\F^n$, denoted by $C\le \F^n$, is called 
a {\em linear code} of {\em length}~$n$ over the finite field $\F$,
and, any $c\in C$ is called a {\em code word}.  
Let $C\le \F^n$. The dimension $\dim C$ is called the
{\em information length} of $C$, while the fraction
$\R(C)=\dim C/n$ is called the {\em rate} of $C$.
And, $\d(C):=\min\{\d(c,c')\,|\,c\ne c'\in C\}$
is said to be the {\em minimal distance} of $C$, while
the fraction $\Delta(C):=\d(C)/n$ is called
 the {\em relative minimal distance} of $C$.

The good codes possess large rate and large relative minimum distance.
A class of codes is said to be {\em asymptotically good} if there exists a positive
real number $\delta$ such that, for any positive integer $N$ we can find 
a code in the class with code length greater than $N$, and 
with both the rate and the relative minimum distance greater than $\delta$. 

In 1950's the linear codes over a finite field 
are proved to be asymptotically good (\cite{G, V}): 
for any $N$ there exist linear codes of code length~$n>N$ 
with rate and relative minimal distance attaining the (asymptotic) GV-bound 
(cf. Definition \ref{d entropy} and Figure \ref{p GV} below). 
Later, the GV-bound is proved to be a phase transition point for the linear codes,  
i.e., in an asymptotic sense, the parameters of most (not only ``there exist'') 
linear codes attain the GV-bound (\cite{P67}). 
However, the algorithmic performance of linear codes 
is unsatisfactory, because the syndrome tables (for decoding the linear codes) 
have size about $\exp(n)$ (cf. \cite[\S1.11]{HP}).

In general, rich mathematical structures attached to codes
enhance the properties of the codes, 
including the improvements of the algorithmic properties.  
Let $G$ be a finite group of order $n$, and $\F G$ be the group algebra
(i.e., the $\F$-vector space with basis $G$ and equipped 
with a multiplication induced by the group multiplication of $G$).
Any element of $\F G$ is identified with a word of $\F^n$ in a natural way.
Any left ideal $C$ of $\F G$, denoted by $C\le \F G$, 
is called a {\em group code} ($\F G$-code). 
If $G$ is a cyclic group (abelian group, dihedral group, respectively), then
$C\le\F G$ is called a {\em cyclic code} ({\em abelian code, dihedral code,} respectively).  
For any integer $t\ge 1$, the outer direct sum
$\F G\times\cdots\times\F G$ (with~$t$ copies) is an $\F G$-module;
any $\F G$-submodule of it is called a {\em quasi-$\F G$ code of index $t$}.
Similarly, if $G$ is cyclic (abelian, respectively), then
quasi-$\F G$ codes of index $t$ are said to be {\em quasi-cyclic} 
({\em quasi-abelian}, respectively,) codes of index~$t$.
Please see Section \ref{group codes} for more details.

Cyclic codes are extensively studied since 1950's. 
The cyclic structure improves very much 
the algorithms for cyclic codes (cf. \cite[\S4.6]{HP}), 
and upgrade the mathematical properties of cyclic codes as well.
The well-known BCH codes, 
a subclass of cyclic codes discovered around 1960 (\cite{BC, H59}), 
have satisfactory rate and relative minimal distance if the code length is moderate; 
and have better algorithms (\cite{P61, B68, M69}), e.g.,  
 the complexity of the Berlekamp-Massey algorithm for BCH codes
(\cite{B68, M69}) is only $O(n^3)$ ($n$ is the code length).
However, if the code lengths go to infinity but keep the rates of 
BCH codes greater than a positive real number, 
then the relative minimal distances of the BCH codes go to zero;  
in other words, BCH codes are not asymptotically good (\cite{LW67}).

Nearly at the same time (around 1967), it was fund that,
if consider the cyclic codes with such code length, whose any prime divisor  
is coprime to ${\rm char}\,\F$ (the characteristic of $\F$)
and less than a given number (though the code length itself goes to infinity), 
then such cyclic codes are not asymptotically good (\cite{B67}).
Later, another subclass of the cyclic codes of length $n=mp^k$ with $k\to\infty$, 
where $p={\rm char}\,\F$ and $m$ is fixed and coprime to~$p$
(such codes are the so-called {\em repeated-root cyclic codes}), 
is proved to be asymptotically bad (\cite{CMSS91, JLLX}).
And, 
one more asymptotically bad subclass of cyclic codes is shown in \cite{MW06}.

Therefore, it is a 
 long-standing mysterious open question that:
{\it whether or not the cyclic codes over a finite field are asymptotically good?} 

However, from a long time ago the quasi-cyclic codes of index $2$
are proved to be asymptotically good (\cite{CPW, C, K}). 
What is more, 
 even if another interesting mathematical property ``{\em self-duality}''
(cf. Eq.\eqref{inner} below) is attached,  
the binary self-dual quasi-cyclic codes of index $2$ are still proved to be
asymptotically good~(\cite{MW07}).  
And, the asymptotic goodness of self-dual quasi-cyclic codes 
(but with index $t\to\!\infty$) are also obtained (\cite{D, LS}).  

The dihedral groups are non-abelian but near to cyclic groups 
(they have cyclic subgroups of index $2$).  In \cite{BM},
the binary dihedral codes are proved to be asymptotically good;
and, the binary quasi-abelian codes of fixed index $t\ge 2$
are also proved to be asymptotically good. 
The asymptotic goodness of quasi-abelian codes of fixed index $t\ge 2$ 
has been extended to any $q$-ary case in \cite{L PhD}. 
And the asymptotic goodness of dihedral codes 
has been extended to any $q$-ary case recently in \cite{FL20}. 

This chapter consists of some of our studies after 2010 
on the asymptotic properties of several classes of quasi-group codes.
We'll explain the studies in a consistent and self-contained style.  
 
We begin with the classical results on linear codes in Section \ref{linear codes}. 

Section \ref{group codes} is a brief sketch about group codes and quasi-group codes.

A general result on the volumes of subsets of the so-called balanced codes
 is described in Section \ref{balanced codes}.

In Section \ref{quasi-abelian}, for an abelian group $G$, the
quasi-$\F G$ codes of index $t\to\infty$ are considered, and the GV-bound
is proved to be a phase transition point of the parameters of the quasi-$\F G$ codes;
in particular, such codes are asymptotically good.

In Section \ref{fractional}, 
we first prove that the quasi-abelian codes of index $2$ is asymptotically good 
(the asymptotic goodness of quasi-abelian codes 
of any fixed index $t\ge 2$ mentioned above can be proved in a similar way);
then we exhibit the asymptotic goodness of the quasi-cyclic codes of
{\em fractional index} in a concise and uniform way.

The self-dual (self-orthogonal) quasi-abelian codes of index $2$
are investigated in Section \ref{self-dual}, they are still asymptotically good 
(the self-dual ones exist if and only if $q\,{\not\equiv}\,3\!\pmod 4$
(Corollary~\ref{c self-dual exist} below), while the
self-orthogonal ones exist unconditionally). 
But, in our study their parameters do not attain the GV-bound
(as a comparison, the asymptotic good quasi-abelian codes of index $2$
mentioned above have parameters attaining the GV-bound).

Finally the story on dihedral codes is described in Section \ref{dihedral codes}.
We not only extend the asymptotic goodness of
 dihedral codes to any $q$-ary case (as mentioned above),
 and investigate more precisely the algebraic properties of the
 asymptotic good dihedral codes as well.

We treat the several topics in a consistent random style.  
To make the chapter self-contained, 
we put a sketch about the probabilistic method used in this chapter 
in Appendix (Section \ref{appendix}). 

At the end of each section 
(except for the sections \ref{group codes} and \ref{appendix}),
few comments on the references related to the story of the section
are remarked.

\section{Asymptotic property of linear codes}\label{linear codes}

We sketch preliminaries about linear codes and asymptotic properties,
then turn to a kind of random linear codes.

\subsection{About linear codes and asymptotic properties}

In this chapter $\F$ is always a finite field with cardinality $|\F|=q$, 
which is a prime power.
By $|S|$ we denote the cardinality of any set $S$.

Let $n$ be a positive integer.
Let $\F^n=\{a=(a_1,\cdots,a_n)\,|\,a_i\in\F\}$,
which is an $\F$-vector space of dimension $n$.  
Any sequence $a=(a_1,\cdots,a_n)\in\F^n$
is said to be a {\em word} of length $n$ over $\F$. 
For $a\in\F^n$,  $\w(a):=|\{i\,|\,1\le i\le n, a_i\ne 0\}|$ 
is called the {\em Hamming weight} of $a$. 
For $a,b\in\F^n$, the {\em Hamming distance} 
is defined as $\d(a,b):=\w(a-b)$. 

Any non-empty subset $C\subseteq \F^n$ is called a 
{\em code over $\F$ of length $n$}.
And $\d(C):=\min\{{\rm d}(c,c')\,|\,c\ne c'\in C\}$ is called 
the {\em minimum distance} of $C$.
If a code $C$ is given, any word in $C$ is called a {\em code word}.

Any subspace $C$ of $\F^n$, denoted by $C\le\F^n$,  
is called a {\em linear code over~$\F$ of length $n$}; 
more precisely, $C$ is said to be an $[n,k,d]$ code over $\F$, 
where $k=\dim C$ and $d=\d(C)$. 
Note that $\F^n$ is equipped with the so-called {\em euclidean inner product}:
\begin{equation}\label{inner}
\langle a,b\rangle=a_1b_1+\cdots+a_n b_n, \qquad
a,b\in\F^n.
\end{equation}
Set $C^\bot\! :=\{a\in\F^n\,|\,\langle c,a\rangle=0,~\forall~ c\in C\}$. 
If $C= C^\bot$ ($C\subseteq C^\bot$, respectively), 
then $C$ is said to be {\em self-dual} ({\em self-orthogonal}, respectively).
On the other hand,  $C$ is called a {\em linear complementary dual} 
(LCD in short) code if $C\cap C^\bot=0$.

\bex
If $C\le \F^n$, then ${\rm d}(C)={\rm w}(C)$, 
where ${\rm w}(C):=\min\{{\rm w}(c)\,|\,0\ne c\in C\}$
is called the {\em minimum weight} of the linear code $C$.
\eex

When we consider asymptotic properties of codes (with code length $n\to\infty$), 
the following relative versions are convenient:
\begin{itemize}
\item \vskip-4pt
 Denote $\R(C)=\!\log_q|C|/n$ ($=\dim C/n $ if $C$ is linear), 
 called the {\em rate} of $C$.
\item \vskip-5pt
 For $a\in\F^n$, $\w(a)/n$ is called the {\rm relative weight} of $a$.
\item\vskip-4pt
 Denote $\Delta(C)={\rm d}(C)/n$,
called the {\em relative minimum distance} of $C$.
If $C$ is a linear code, then $\Delta(C)={\rm d}(C)/n=\w(C)/n$, 
which is also called the {\em relative minimum weight} of $C$.  
\end{itemize}

\vskip-5pt\noindent
The bigger both the rate and the relative minimum distance are, 
the better the codes are.

\begin{definition}\label{asym good}\rm
A code sequence $C_1,C_2,\cdots$ is said to be {\em asymptotically good} 
if the lengths $n_i$ of $C_i$ goes to infinity and
there is a positive real number $\delta$ such that 
$\R(C_i)>\delta$ and $\Delta(C_i)>\delta$ for $i=1,2,\cdots$.  

A class of codes is said to be {\em asymptotically good}
if there exists an asymptotically good code sequence $C_1,C_2,\cdots$
within the class.
\end{definition}

\begin{lemma}
For any $\varepsilon>0$ there is an integer $N$ such that, 
if a code $C$ has length $n>N$ and $\Delta(C)> 1-q^{-1}$, 
then $\R(C)<\varepsilon$.
\end{lemma}
\begin{proof}
By Plotkin Bound (Exercise \ref{Plotkin} and its notation), 
$\frac{1}{M}\ge\frac{\Delta-(1-q^{-1})}{\Delta}$. 
Since $\Delta(C)> 1-q^{-1}$, we have that 
$M\le\frac{\Delta}{\Delta-(1-q^{-1})}=\frac{dq}{dq-(q-1)n}\le dq\le nq$.
So
$\textstyle
\R(C)=\frac{\log_q M}{n}\le \frac{1+\log_q n}{n}
 \;\mathop{\longrightarrow}\limits_{n\to\infty}\; 0. 
$
\end{proof}

\bex\label{Plotkin}
Let $C\subseteq\F^n$, $M=|C|$, $d=\d(C)$, $\Delta=\frac{d}{n}$. 

(1)~ Let $T=\sum_{c\ne c'\in C}\d(c,c')$ be the total sum of distances between
code words of~$C$. Then $2T\ge M(M-1)d$.

(2)
Let 
${\cal M}=\begin{pmatrix} c_{11} &\cdots& c_{1j}&\cdots & c_{1n}\\
 \cdots& \cdots &\cdots& \cdots &\cdots\\ 
 c_{M1} &\cdots & c_{1j}&\cdots & c_{Mn} \end{pmatrix}$
 be the matrix whose rows are all code words of $C$. 
For $\alpha\in\F$, let $t_{\alpha,j}$ denote the number of entries 
in the $j$-th column of ${\cal M}$ that equals~$\alpha$. 
Then $2T=nM^2-\sum_{j=1}^n\sum_{\alpha\in\F}t^2_{\alpha,j}$.
\hint{in the $j$-th column, the pair $(\alpha,\alpha')$ 
for $\alpha\ne\alpha'\in\F$ appears $t_{\alpha,j}t_{\alpha',j}$ times;
and $\sum_{\alpha\ne\alpha'\in\F}t_{\alpha,j}t_{\alpha',j}
=\frac{1}{2}\big((\sum_{\alpha\in\F}t_{\alpha,j})^2
  -(\sum_{\alpha\in \F}t_{\alpha,j}^2)\big)$.}

(3)~ $2T\le (1-q^{-1})nM^2$.  
\hint
{$\sum_{\alpha\in \F}t_{\alpha,j}^2\big/|\F|\ge 
  (\sum_{\alpha\in \F}t_{\alpha,j}\big/|\F|)^2 $.}

(4) ({\sf Plotkin Bound})~ $\frac{1}{M}\ge\frac{\Delta-(1-q^{-1})}{\Delta}$.
\hint{combine (1) and (3).}
\eex

Therefore, to consider asymptotic properties, 
in Definition \ref{asym good} only the real number 
$\delta\in[0,1-q^{-1}]$ makes sense.

\begin{definition}\label{d entropy}\rm
The following is called the {\em $q$-entropy} function: 
\begin{equation}\label{e entropy}
h_q(\delta)=\delta\log_q(q-1)-\delta\log_q\delta-(1-\delta)\log_q(1-\delta),
\quad \delta\in [0,1-q^{-1}],
\end{equation}
which is an increasing and concave function on the interval $[0,1-q^{-1}]$
with $h_q(0)=0$ and $h_q(1-q^{-1})=1$. 
And $g_q(\delta)=1-h_q(\delta)$ is called the {\em asymptotic GV-bound}, 
or {\em GV-bound} for short.
\end{definition}

\begin{figure}[h]\label{p GV}
\begin{center}\setlength{\unitlength}{0.8pt}
\begin{picture}(125,125)
\put(10,110){\line(1,0){100}} \put(110,10){\line(0,1){100}}
\bezier{50}(92,10)(92,60)(92,110)
\thicklines
\put(10,10){\vector(0,1){115}}\put(10,10){\vector(1,0){115}}
\put(-18,118){\scriptsize$h_q(\delta)$} \put(122,1){\scriptsize$\delta$}
\qbezier(10,10)(18,65)(48,94) \qbezier(48,94)(68,110)(92,109.8)
\put(10,110){\line(1,0){2}} \put(4,105){\scriptsize$1$}
\put(92,10){\line(0,1){2}} \put(85,1){\scriptsize$1$-$\frac{1}{q}$}
\end{picture}
\hskip17mm
\begin{picture}(125,125)
\put(10,110){\line(1,0){100}} \put(110,10){\line(0,1){100}}
\thicklines
\put(10,10){\vector(0,1){115}}\put(10,10){\vector(1,0){115}}
\put(-18,118){\scriptsize$g_q(\delta)$} \put(122,1){\scriptsize$\delta$}
\qbezier(10,110)(18,55)(48,26) \qbezier(48,26)(68,10)(92,10.3)
\put(10,110){\line(1,0){2}} \put(4,105){\scriptsize$1$}
\put(92,10){\line(0,1){2}} \put(85,1){\scriptsize$1$-$\frac{1}{q}$}
\end{picture}
\caption{Graphs of $h_q(\delta)$ and $g_q(\delta)$}\label{graph h_q}
\end{center}
\end{figure}
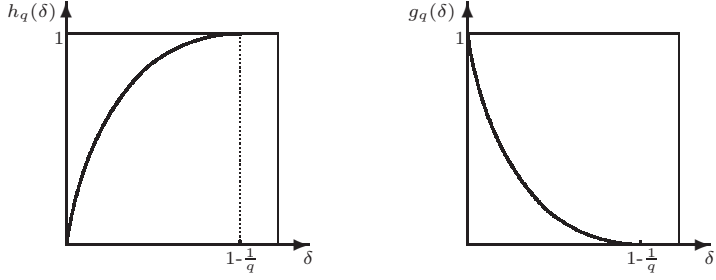

\bex\label{e g_q}
$h_q(\delta)$ is an increasing concave function on the interval $[0,1-\frac{1}{q}]$ 
with $h_q(0)=0$ and $h_q(1-q^{-1})=1$; 
while $g_q(\delta)$ is a decreasing convex function 
on the interval $[0,1-\frac{1}{q}]$ 
with $g_q(0)=1$ and $g_q(1-q^{-1})=0$. 
(Figure \ref{graph h_q}.)
\eex

Let $\delta\in (1,1-q^{-1})$. We denote 
\begin{equation}\label{Hamming ball}
(\F^n)^{\le\delta}=\{a\,|\,a\in\F^n, \w(a)\le\delta n\}, 
\end{equation}
which is known as the {\em Hamming ball} 
with radius $\lfloor\delta n\rfloor$ centred at $0$.

\bex
The cardinality 
$|(\F^n)^{\le\delta}|=\sum_{i=0}^t\binom{n}{i}(q-1)^i$, 
where $\binom{n}{i}$'s denote the binomial coefficients, 
and $t=\lfloor\delta n\rfloor$ denotes the largest integer less or equal $\delta n$
(so $t\approx \delta n$ and in the following we'll write $\sum_{i=0}^{\delta n}$).  
\eex

\begin{lemma}\label{<Hamming ball<}
Let $0<\delta<1-\frac{1}{q}$. Then~
$q^{n\big(h_q(\delta)-\frac{\log_q(n+1)}{n}\big)}\le
   \big|(\F^n)^{\le\delta}\big|\le q^{nh_q(\delta)}. $
   
\noindent{\rm({\bf Remark}: 
That is, $|(\F^n)^{\le\delta}|=q^{n(h_q(\delta)-o(1))}$,
 where $o(1)$ is a non-negative infinitesimal; more precisely, there are
 $\varepsilon_1, \varepsilon_2, \cdots$ such that $0\le\varepsilon_n\le\frac{\log_q(n+1)}{n}$, 
 and $|(\F^n)^{\le\delta}|=q^{n(h_q(\delta)-\varepsilon_n)}$.)}
\end{lemma}
\begin{proof} 
Since $q^{nh_q(\delta)}=\big(q^{h_q(\delta)}\big)^n
 =(q-1)^{\delta n}\delta^{\delta n}(1-\delta)^{(1-\delta)n}$, we have
$$\textstyle
  \frac{|(\F^n)^{\le\delta}|}{q^{nh_q(\delta)}}=
    \frac{\sum_{i=0}^{\delta n}\binom{n}{i}(q-1)^i}
    {(q-1)^{\delta n}\delta^{\delta n}(1-\delta)^{(1-\delta)n}}
  =\sum_{i=0}^{\delta n}\binom{n}{i}(q-1)^i(1-\delta)^{n}
   \Big(\frac{\delta}{(q-1)(1-\delta)}\Big)^{\delta n}.
$$
Because $\delta\le 1-\frac{1}{q}$, we have $\frac{\delta}{q-1}\le 1-\delta$,
hence $\frac{\delta}{(q-1)(1-\delta)}\le 1$. Thus
\begin{align*}\textstyle
 \frac{|(\F^n)^{\le\delta}|}{q^{nh_q(\delta)}}
 &\textstyle\le \sum_{i=0}^{\delta n}\binom{n}{i}(q-1)^i(1-\delta)^{n}
   \Big(\frac{\delta}{(q-1)(1-\delta)}\Big)^{i}
   =\sum_{i=0}^{\delta n}\binom{n}{i}\delta^{i}(1-\delta)^{n-i}\\
 &\textstyle \le\sum_{i=0}^{n}\binom{n}{i}\delta^{i}(1-\delta)^{n-i}
 =(\delta+1-\delta)^n=1. 
\end{align*}
On the other hand, by the following Exercise \ref{e binomial}(3), 
\begin{align*}
 |(\F^n)^{\le\delta}|
 &\textstyle =\sum_{i=0}^{\delta n}\binom{n}{i}(q-1)^i
      \ge\binom{n}{\delta n}(q-1)^{\delta n}
      =\binom{n}{\delta n}q^{\delta n\log_q(q-1)}\\
 &\ge q^{n\big(-\delta\log_q\delta-(1-\delta)\log_q(1-\delta)
   -\frac{\log_q(n+1)}{n}\big)}\cdot q^{\delta n\log_q(q-1)}\\
 &= q^{n\big(\delta\log_q(q-1)-\delta\log_q\delta-(1-\delta)\log_q(1-\delta)
   -\frac{\log_q(n+1)}{n}\big)}.
\end{align*}
That is, $|(\F^n)^{\le\delta}|\ge q^{n\big(h_q(\delta)-\frac{\log_q(n+1)}{n}\big)}$.
\end{proof}

\bex\label{e binomial}
Let $0<\delta<1$. Consider 
$
 1=\big(\delta+(1-\delta)\big)^n=\sum_{i=0}^n\binom{n}{i}\delta^{i}(1-\delta)^{n-i}
$.

(1) For $0\le i\le n$, 
$\binom{n}{i}\delta^{i}(1-\delta)^{n-i}\le\binom{n}{\delta n}\delta^{\delta n}(1-\delta)^{(1-\delta)n}$.
\hint{consider the ratio 
$\binom{n}{i}\delta^{i}(1-\delta)^{n-i}
  \big/\binom{n}{i+1}\delta^{i+1}(1-\delta)^{n-i-1}$.}

(2)
$(n+1)\cdot\binom{n}{\delta n}\delta^{\delta n}(1-\delta)^{(1-\delta)n}
 \ge \sum_{i=0}^n\binom{n}{i}\delta^i(1-\delta)^{n-i} =1.$

(3)
$\binom{n}{\delta n}\ge
q^{n\big(-\delta\log_q\delta-(1-\delta)\log_q(1-\delta)-\frac{\log_q(n+1)}{n}\big)}.
$ 
\eex

\subsection{Random linear codes}

\smallskip
Now we study the asymptotic properties of a kind of random linear codes.
Some fundamentals about probabilistic methods are sketched in 
Appendix (Section \ref{appendix}).  

\begin{definition}\label{C_M}\rm
Assume that $0<r<1$, $k=\lfloor rn\rfloor$,  and $0<\delta<1-\frac{1}{q}$.
Let $\F^{k\times n}=\big\{M=(m_{ij})_{k\times n}\;\big|\;m_{ij}\in\F\big\}$, 
where $(m_{ij})_{k\times n}$ denotes the $k\times n$ matrix over $\F$ 
with $(i,j)$-entries $m_{ij}$.
We consider $\F^{k\times n}$ as a probability space with equal probability for each
sample $M\in\F^{k\times n}$; and take the following  notation.  
\begin{itemize}
\item[(1)] \vskip-2pt
$C_M=\big\{bM\;\big|\; b\in\F^k\big\}\le\F^n$
is a {\em random linear code} over the probability space $\F^{k\times n}$, 
where $b=(b_1,\cdots,b_k)\in\F^k$. 
In a notation of Shannon's information theory,  
$C_M$ for $M\in\F^{k\times t}$
is called a {\em random linear code ensemble};
cf. \cite[ch.6]{MM}.
\item[(2)] \vskip-4pt
 Both $\Delta(C_M)$ and $\R(C_M)$ are random variables; \\
then $\R(C_M)\le \frac{k}{n}\le r$, and 
 $\R(C_M)=\frac{k}{n}$ $\iff$ ${\rm rank}(M)=k$. 
\item[(3)] \vskip-4pt
For each $b\in\F^k$, ~
$X_b=\begin{cases}1, & 0<\frac{\w(bM)}{n}\le\delta;\\ 0, &\mbox{otherwise}; \end{cases}$ is a $0$-$1$ random variable. 
\item[(4)] \vskip-4pt
$X=\sum_{b\in \F^k}X_b$ is a non-negative integer random variable,
which stands for the number of the non-zero code words in $C_M$ 
with relative weight $\le\delta$.\\
Thus $\Delta(C_M)\le\delta$ if and only if $X\ge 1$; i.e., 
for the probabilities we have
\begin{equation}\label{X>1}
\Pr(\Delta(C_M)\le\delta)=\Pr(X\ge 1).
\end{equation}
\end{itemize}
\end{definition}

\begin{lemma}\label{l E(X)}
$\lim\limits_{n\to \infty}\E(X)
=\begin{cases}0, & r<g_q(\delta);\\ \infty, & r>g_q(\delta);  \end{cases}$
where $\E(X)$ denotes the expectation (cf. Definition \ref{d random var}) 
of the random variable $X$ defined in Definition \ref{C_M}(4).
\end{lemma}

\begin{proof} 
By Definition \ref{C_M}(3), $X_0=0$ hence $\E(X_0)=0$.
Let $0\ne b\in\F^k$. We compute $\E(X_b)$. Note that 
$ \F^{k\times n}\to \F^n,~ M\mapsto bM, $
is a surjective linear map, and the size of the kernel equals 
$|\F^{k\times n}|\big/|\F^n|=q^{kn-n}$. Thus
for any $a\in\F^n$ there are exactly $q^{kn-n}$ elements $M\in\F^{k\times n}$
such that $bM=a$; and $X_b=1$ if and only if $0<\frac{\w(a)}{n}\le \delta$.
So (cf. Eq.\eqref{E of 0-1})
$$\textstyle
\E(X_b)=\Pr(X_b=1)
 =\frac{(|(\F^n)^{\le\delta}|-1)\cdot q^{kn-n}}{|\F^{k\times n}|}
 =\frac{|(\F^n)^{\le\delta}|-1}{q^n}; 
$$
and by Lemma \ref{<Hamming ball<},
\begin{equation}\label{E(X_b)}\textstyle
\E(X_b)=\frac{q^{n(h_q(\delta)-o(1))}-1}{q^n}=q^{-n(1-h_q(\delta)+o(1))}
=q^{-n(g_q(\delta)+o(1))},
\end{equation}
where $o(1)$ is a non-negative infinitesimal. 
By the linearity of the expectation,
$$\textstyle
\E(X)=\sum_{b\in\F^k-\{0\}}\E(X_b)=(q^k-1)q^{-n(g_q(\delta)+o(1))}
=q^{-n\big(g_q(\delta)-\frac{k}{n}+o(1)\big)}.
$$
Note that $\lim\limits_{n\to\infty}\frac{k}{n}=r$. So
$
\lim\limits_{n\to\infty}\E(X)
=\lim\limits_{n\to\infty}q^{-n\big(g_q(\delta)-r+o(1)\big)},
$ 
 from which the lemma follows at once.
\end{proof}

\begin{theorem}\label{Delta(C_M) >delta}
$\lim\limits_{n\to \infty}\Pr\big(\Delta(C_M)>\delta\big)
=\begin{cases}1, & r<g_q(\delta);\\ 0, & r>g_q(\delta).  \end{cases}$
\end{theorem}
\begin{proof}
If $r\!<\! g_q(\delta)$, by Eq.\eqref{X>1}, 
Markov Inequality(Lemma \ref{Markov}) and Lemma~\ref{l E(X)},
$$\lim\limits_{n\to \infty}\Pr\big(\Delta(C_M)\le\delta\big)
=\lim\limits_{n\to \infty}\Pr(X\ge 1)
\le\lim\limits_{n\to \infty}\E(X)=0.$$
That is, $\lim\limits_{n\to \infty}\Pr\big(\Delta(C_M)>\delta\big)=1$.

In the following, we assume that $r>g_q(\delta)$, 
and apply Lemma \ref{2'nd moment} to show that 
$\lim\limits_{n\to \infty}\Pr\big(X\!\ge\! 1\big)=1$; by~Eq.\eqref{X>1},
it implies that $\lim\limits_{n\to \infty}\Pr\big(\Delta(C_M)\!>\!\delta\big)=0$.

Given any $0\ne b\in\F^k$. By the linearity of the expectation, we have
\begin{equation}\label{X|X_b}\textstyle
 \E(X|X_b=1)=\sum_{b'\in\F^k}\E(X_{b'}|X_b=1) 
=\sum_{0\ne b'\in\F^k}\E(X_{b'}|X_b=1). 
\end{equation}
Let $0\ne b'\in\F^k$. 
The expectation $\E(X_b)$ is obtained in Eq.\eqref{E(X_b)}; 
in particular, $\E(X_{b'})=\E(X_b)$. 
Further, there are two cases: 
\begin{itemize}
\item\vskip-5pt
$b'$ and $b$ are linearly dependent. The 
probability of this case is $\frac{q-1}{q^k-1}$.
In this case $b'M$ and $bM$
are also linearly dependent, hence $\w(b'M)=\w(bM)$. Thus,
$$\E(X_{b'}|X_b=1)=\Pr(X_{b'}=1|X_b=1)=1.$$
\item\vskip-3pt
$b'$ and $b$ are linearly independent. 
The probability of this case is $1-\frac{q-1}{q^k-1}=\frac{q^k-q}{q^k-1}$.
In this case, for any $a', a\in\F^n$ there is a matrix $M\in\F^{k\times n}$
such that $(b'M,bM)=(a',a)$. When $M$ is running on $\F^{k\times n}$,
$(b'M,bM)$ is running on $\F^n\times\F^n$ uniformly.   
Thus the events ``$X_{b'}=1$'' and ``$X_b=1$'' are independent, 
hence (cf. the note after Lemma \ref{conditional})
$$\E(X_{b'}|X_b=1)=\Pr(X_{b'}=1|X_b=1)=\E(X_{b'})=\E(X_b).$$  
\end{itemize}
\vskip-3pt
By Total Probability Formula (Lemma \ref{l total}),
\begin{align*}
\textstyle
  \E (X_{b'}|X_b=1)=1\cdot\frac{q-1}{q^k-1}+
  \E (X_{b})\cdot\frac{q^k-q}{q^k-1}
  =\frac{(q-1)+(q^k-q) \E (X_{b})}{q^k-1}\,.
\end{align*}
Combining it with Eq.\eqref{X|X_b}, for any $0\ne b\in\F^k$ we have: 
\begin{align*}\textstyle
\E(X|X_b=1)=\sum_{0\ne b'\in\F^k}
 \frac{(q-1)+(q^k-q) \E (X_{b})}{q^k-1}
 =(q-1)+(q^{k}-q) \E (X_{b})\,.
\end{align*}
By Lemma~\ref{2'nd moment},
$$\textstyle
 \Pr(X\ge 1)\ge\sum_{{0}\ne{b}\in \F^k}
 \frac{\E(X_b)}{\E(X|X_b=1)}
 =\sum_{{0}\ne{b}\in \F^k}
  \frac{\E (X_{b})}{(q-1)+(q^{k}-q) \E (X_{b})}\,.
$$
Given $0\ne c\in\F^k$. Then $\E (X_{b})= \E (X_{c})$ 
for any $0\ne b\in\F^k$. We get 
$$
  \Pr(X\ge 1)\ge\frac{(q^k-1) \E (X_c)}
    {(q-1)+(q^k-q) \E (X_{{c}})}
  =\frac{1}
    {\frac{q-1}{(q^k-1) \E (X_{{b}})}+\frac{q^k-q}{q^k-1}}\,.
$$
Clearly, $\lim\limits_{n\to\infty}\frac{q^k-q}{q^k-1}
   =\lim\limits_{n\to\infty}\frac{q^{rn}-q}{q^{rn}-1}=1$.
By the linearity of the expectation again, 
$$\textstyle
 \E (X)=\sum_{0\ne b\in \F^k}\E(X_{b})
  =(q^k-1) \E (X_{c}).
$$ 
Because $ r > g_q(\delta)$,  
by Lemma \ref{l E(X)} we get
$\lim\limits_{n\to\infty}(q^k-1) \E (X_{c})=
 \lim\limits_{n\to\infty} \E (X)=\infty$. Thus
$$
 \lim\limits_{n\to\infty}\Pr(X\ge 1)\ge
 \lim\limits_{n\to\infty} \frac{1}
    {\frac{q-1}{(q^k-1) \E (X_{b})}+\frac{q^k-q}{q^k-1}}
 =1.
$$
That is, $\lim\limits_{n\to\infty}\Pr(X\ge 1)=1$.
\end{proof}

As explained in Definition \ref{C_M}(2), 
the rate $\R(C_M)\le\frac{k}{n}\le r$, and 
$\R(C_M)=\frac{k}{n}$ if and only if ${\rm rank}(M)=k$.
We take $N\in\F^{k\times n}$ such that ${\rm rank}(N)=k$. 
Then $C_N$ is a random linear code over $\F$ of length $n$ and dimension $k$.

\begin{theorem}\label{Delta(C_N) >delta}
$\lim\limits_{n\to\infty}\Pr\big(\Delta(C_N)>\delta\big)=\begin{cases}
   1, & r < g_q(\delta);\\
   0, & r > g_q(\delta).\end{cases}
$
\end{theorem}

\begin{proof} 
The probability $\Pr\!\big(\Delta(C_N)\!>\!\delta\big)$
is written as a conditional probability: 
$$
 \Pr\big(\Delta(C_N)>\delta\big)
   =\Pr\big(\Delta(C_M)>\delta\mid{\rm rank}(M)=k\big).
$$
By Total Probability Formula,
\begin{align*}
\Pr\big(\Delta(C_M)>\delta\big)
 &=\Pr\big(\Delta(C_M)>\delta\mid{\rm rank}(M)=k\big)
        \Pr\big({\rm rank}(M)=k\big)\\
 &~~ +\Pr\big(\Delta(C_M)>\delta\mid{\rm rank}(M)<k\big)
    \Pr\big({\rm rank}(M)<k\big).
\end{align*}
By Exercise \ref{e rank=} below, 
$$
 \lim_{n\to\infty}\Pr\big({\rm rank}(M)<k\big)=0\,,\qquad
 \lim_{n\to\infty}\Pr\big({\rm rank}(M)=k\big)=1\,.
$$
Thus
$$
\lim_{n\to\infty}\Pr\big(\Delta(C_M)>\delta\mid{\rm rank}(M)=k\big)
 =\lim_{n\to\infty}\Pr\big(\Delta(C_M)>\delta\big)\,.
$$
So, the theorem follows from Theorem \ref{Delta(C_M) >delta} at once.
\end{proof}

\bex\label{e rank=}
Let $M$ be a random $k\times n$ matrix over $\F$, 
where $k=\lfloor rn\rfloor$ with $r\in(0,1)$. 

(1)~ ${\rm rank}\,M<k$ if and only if all the column vectors of 
$M$ belong to one and the same $(k-1)$-dimensional subspace of $\F^k$.

(2)~ Let $W\le\F^k$ be a $(k-1)$-dimensional subspace.
Then the probability that all the column vectors of $M$ belong to $W$
is equal to $\frac{1}{q^n}$. 
\hint{the probability that a random vector of $\F^k$ 
 belongs to $W$ equals $q^{k-1}/q^k=1/q$.}

(3)~ The number of the $(k-1)$-dimensional subspaces of $\F^k$
is equal to $\frac{q^k-1}{q-1}$. \\
\hint{$(k-1)$-dimensional subspaces are 1-1 corresponding to 
the $1$-dimensional subspaces;
 the number of $1$-dimensional subspaces is equal to $(q^k-1)/(q-1)$.}

(4)~ $\Pr\big({\rm rank}(M)<k\big)\le q^k/q^n$, hence
$\lim\limits_{n\to\infty}\Pr\big({\rm rank}(M)<k\big)=0$.\\
\hint{\kern-2pt$\Pr\!\big({\rm rank}(M)\!<\!k\big)\le
 \sum_{W}\Pr\big(\mbox{all columns of $M$ belong to $W$}\big)$,
 where the subscript $W$ runs over $(k\!-\!1)$-dimensional subspaces of $\F^k$.}
\eex

As a consequence of the first equality (the case ``$r< g_q(\delta)$'') 
of Theorem \ref{Delta(C_N) >delta}, 
we have:

\begin{corollary}\label{c C_N}
If $\delta\in(0,1-q^{-1})$ and $0<r<g_q(\delta)$, 
then there exist linear codes $C_1, C_2,\cdots$ 
with the length $n_i$ of $C_i$ going to infinity such that
$\lim\limits_{i\to\infty}\R(C_i)=r$ and
$\Delta(C_i)>\delta$ for all $i=1,2,\cdots$. 
\end{corollary}

In other words, there are asymptotically good linear code sequences
whose parameters attain the GV-bound.  

\begin{remark}\label{r linear}\rm
We conclude this section by some remarks. 

(1) The first inequality of Theorem \ref{Delta(C_N) >delta} is just the 
well-known GV-bound in \cite{G}, \cite{V}; see Corollary~\ref{c C_N}.  
While the second inequality of Theorem~\ref{Delta(C_N) >delta}
was proved in \cite{P67}. The argument we stated here is based on \cite{FLLSX}.

(2) As we have seen, Theorem~\ref{Delta(C_N) >delta} follows from
Theorem~\ref{Delta(C_M) >delta} easily.
The theorems imply that, if $\delta\in(0,1-q^{-1})$ is given, 
then $g_q(\delta)$ is a threshold point
(phase transition point) for the event ``$\Delta(C_M)>\delta$''.
Another explanation: if $r$ is fixed, then the random variable 
$\Delta(C_M)$ is asymptotically distributed at the point $g_q^{-1}(r)$, 
where $g_q^{-1}(\cdot)$ denotes the inverse function of~$g_q(\cdot)$.   

(3) The argument for the first inequality of 
Theorem \ref{Delta(C_M) >delta} is usually 
called the {\em first moment method};  
while the so--called {\em second moment method}
is used to prove the second inequality, which is more complicated than the
first moment method. See Section \ref{appendix} Appendix please. 
\end{remark}

\section{Fundamentals on group codes}\label{group codes} 

From now on, 
 $\F$ is always a finite field with $|\F|=q$, and
 $G$ is always a finite group of order $n$. As usual, $1_G$ (or $1$ for short)
denotes the identity of $G$.

Let $\F G=\big\{\sum_{x\in G}a_x x\,|\,a_x\in\F\big\}$
be the $\F$-vector space with basis $G$. 
The multiplication of the group $G$ is linearly extended to
a multiplication of~$\F G$:
\begin{equation}\label{multiplication}
\textstyle
\big(\sum_{x\in G}a_x x\big)\cdot\big(\sum_{y\in G}b_y y\big)
=\sum_{z\in G}\big(\sum_{xy=z}a_xb_y\big)z.
\end{equation}
Then $\F G$ is a (finite) ring, more precisely, 
an $n$-dimensional  $\F$-algebra, 
called the {\em group algebra} of $G$ over $\F$.
Any element $\sum_{x\in G}a_x x\in\F G$ can be 
viewed as a word $(a_x)_{x\in G}\in\F^n$ of length $n$, i.e., 
the entries of the word labelled by $G$ 
and the $x$-th entry of  the word 
equals $a_x$. Then $\F G$ is $\F$-linearly isomorphic to $\F^n$. 
For $a=\sum_{x\in G}a_x x, b=\sum_{x\in G}b_x x\in\F G$,
 the euclidean inner product of $a,b$ is defined as (see Eq.\eqref{inner}): 
$\langle  a,b\rangle=\sum_{x\in G}a_xb_x$.

\begin{definition}\rm
Any left ideal $C$ of $\F G$, denoted by $C\le \F G$, is called a {\em group code} 
of the group $G$ over the field $\F$, 
or an {\em $\F G$-code} for short. 
With the euclidean inner product as above, 
 self-orthogonal $\F G$-codes, self-dual $\F G$-codes and
 LCD  $\F G$-codes etc. are defined as usual.

If $G$ is abelian, then $\F G$ is commutative,
 and ``left ideal'', ``right ideal'' and
``ideal (two-sided ideal)'' are the same, 
and the $\F G$-codes are called {\em abelian codes}. 
In particular,  the $\F G$-codes are called {\em cyclic codes}
if $G$ is a cyclic group.

Similarly, the $\F G$-codes are called {\em dihedral codes}
if $G$ is a dihedral group.
\end{definition}

\begin{example}\label{e cyclic}\rm
Let $G=\{1,x,\cdots,x^{n-1}\}$, with $x^n=1$, 
be a cyclic group of order~$n$. The $\F G$-codes are just well-known
{\em cyclic codes} of length $n$ over~$\F$, which are extensively studied
and applied for a very long time, e.g., see \cite{HP}.  
An element $\sum_{j=0}^{n-1}a_jx^j\in \F G$ is viewed 
as a word $(a_0,a_1,\cdots,a_{n-1})\in\F^n$. 
And, $x(a_0+a_1x+\cdots+a_{n-1}x^{n-1})=a_{n-1}+a_0x+\cdots+a_{n-2}x^{n-1}$. 
So, multiplying by~$x$ transforms 
$(a_0,a_1,\cdots,a_{n-1})$ to $(a_{n-1},a_0,\cdots,a_{n-2})$; i.e.,
multiplying by~$x$ is a cyclic permutation $\rho\!=\!(0,1,\cdots,n\!-\!1)$ 
on the entries of the words of length~$n$. 
Then $C\subseteq\F G$ is an ideal if and only if 
$C$ is a linear subspace of $\F^n$ and
\begin{equation}\label{eq cyclic}
c=(c_0,c_1,\cdots,c_{n-1})\in C ~\implies~ 
\rho(c)=(c_{n-1},c_0,\cdots,c_{n-2})\in C.
\end{equation}
In fact, it is another common definition: ``$C$ is called a cyclic code if  
$C$ is a linear subspace of $\F^n$ and Eq.\eqref{eq cyclic} holds.''

As exhibited in Introduction, 
it is a long-standing open problem: {\em are the cyclic codes
over a finite field asymptotically good?} 
\end{example}

\begin{remark}\label{Cayley permutation}\rm
For any finite group $G$ of order $n$, we have an explanation
for $\F G$-codes similar to Eq.\eqref{eq cyclic}. 
List the elements of $G$ as 
\begin{equation*}
G=\{x_0=1,x_1,\cdots,x_{n-1}\}. 
\end{equation*}
Any $a=\sum_{j=0}^{n-1}a_jx_j\in \F G$ corresponds to a sequence 
$\vec a=(a_0,a_1,\cdots,a_{n-1})\in \F^n$. 
Any $x_i\in G$ provides a permutation $\rho_i$: 
\begin{equation*}
 x_i \longmapsto \rho_i\!=\!\begin{pmatrix}
 0 & 1 & \cdots & j & \cdots & n-1\\
 0' & 1' & \cdots & j' & \cdots & (n-1)'
\end{pmatrix},
\end{equation*}
where $x_ix_j=x_{j'}$ for $j=0,1,\cdots,n-1$.
In particular,  $0'=i$ because $x_i x_0=x_i$.
The map $x_i\mapsto\rho_i$ is called  
the {\em Cayley representation} of the group $G$, 
and $\rho_i$ is called the {\em Cayley permutation} of $x_i$. 
Clearly, $\rho_0$ is the identity permutation. Then 
$$\textstyle
 x_ia=\sum_{j=0}^{n-1}a_jx_ix_j=\sum_{j=0}^{n-1}a_jx_{j'};
$$ 
i.e.,  $x_ia$ corresponds to the sequence 
$\rho_i(\vec a)$ which is obtained 
by $\rho_i$-permuting on the sequence 
$\vec a=(a_0, a_1,\cdots,a_{n-1})$. 
Assume that the group $G$ is generated by $x_1,\cdots,x_r$, $0< r< n$.
Then $C\subseteq \F G$ is a left ideal  ($\F G$-code) if and only if 
$C$ is a subspace of $\F^n$ and 
\begin{equation*}
\vec c \in C ~\implies~ 
\rho_i(\vec c)\in C,~~\forall ~ i=1,\cdots,r.
\end{equation*}
\end{remark}

\begin{example}\rm
Assume that $G=\{x_0=1,x_1,\cdots,x_{n-1}\}$ 
generated by two elements $x_1,x_2$.
Then, $C\subseteq\F G$ is a left ideal ($\F G$-code) if and only if 
$C$ is a linear subspace of $\F^n$ and:
\begin{equation*}
\vec c\in C ~\implies~ 
\rho_1(\vec c), \rho_2(\vec c)\in C.
\end{equation*}

A very simple non-cyclic example: 
let $q=2$ and $G=\{1,a,b,ab\}$ be the Klein four group,  
i.e., $G$ is abelian and generated by two elements $a,b$ of order~$2$.
Then $x_0=1,x_1=a,x_2=b,x_3=ab$, 
$\rho_1=(01)(23)$, $\rho_2=(02)(13)$.   Consider two subspaces: 
\begin{align*}
 A=\{(0,0,0,0),(1,1,0,0),(1,0,0,0),(0,1,0,0)\}\subseteq\F^4,\\
 B=\{(0,0,0,0),(1,1,0,0),(0,0,1,1),(1,1,1,1)\}\subseteq\F^4.
\end{align*}
Then $A$ is not an $\F G$-code, because 
$\rho_2(1,1,0,0)=(0,0,1,1)\notin A$. While
$B$ is an $\F G$-code, because $\rho_1(B)=B$ and $\rho_2(B)=B$.
\end{example}

\begin{remark}\rm 
We recall some general notation about rings.
In this chapter, any ring $R$  has identity $1_R\ne 0$ (or $1\ne 0$ in short);
subrings and ring homomorphisms are always identity-preserving;
and $R^\times$ denotes the {\em multiplicative unit group} 
of all units (invertible elements) of $R$.
We say that $M$ is a {\em (left) $R$-module} 
if $M$ is an additive group (abelian group with operation ``$+$'') and 
there is a map $R\times M\to M$, $(r,m)\mapsto rm$ 
(the image is denoted by $rm$), such that
\begin{itemize}
\item \vskip-5pt
$r(m+m')=rm+rm'$,~ $\forall$ $r\in R$, $m,m'\in M$; 
\item \vskip-4pt
$(r+r')m=rm+r'm$,~ $\forall$ $r,r'\in R$, $m\in M$; 
\item \vskip-4pt
$(rr')m=r(r'm)$,~ $\forall$ $r,r'\in R$, $m\in M$; 
\item \vskip-4pt
$1_Rm=m$,~ $\forall$ $m\in M$.  
\end{itemize}
\vskip-5pt
$R$-modules are just $R$-vector spaces once $R$ is a field. 
The $R$-submodules, quotient modules, module homomorphisms
(or $R$-homomorphisms), etc. can be defined similarly 
to the  subspaces, quotient spaces, linear transformations of vector spaces, etc. 
Both the image and the kernel of an $R$-module homomorphism 
are $R$-submodules. 
The sum $M_1+\cdots+M_t$ of submodules $M_i$ of an $R$-module $M$ 
is also defined similarly. We say that $M$ is a direct sum of $M_1,\cdots,M_t$
and write $M=M_1\oplus\cdots\oplus M_t$ 
if for any $a\in M$ there are unique $a_i\in M_i$ for $i=1,\cdots,t$
such that $a=a_1+\cdots+a_t$. 
\end{remark}

\bex\label{tilde R-module}
If there is a surjective ring homomorphism $\tilde R\to R$, then
any $R$-module~$M$ is also an $\tilde R$-module.
\eex

The map $R\times R\to R$, $(r,m)\mapsto rm$
(the multiplication of $r$ and $m$), satisfies the above four conditions. 
So $R$ is an $R$-module, called the {\em (left) regular $R$-module}.
The $R$-submodules of the regular module $R$
are just left ideals of $R$. 

A left ideal (two-sided ideal) $I$ of $R$ is said to be {\em complementary} if
there is a left ideal (two-sided ideal) $I'$ of $R$ 
such that $R=I\oplus I'$ (i.e., 
any $r\in R$ is uniquely written as $r=a+a'$ with $a\in I$ and $a'\in I'$; 
or equivalently, $R=I+I'$ and $I\cap I'=0$). 
An element $e\in R$ is called an {\em idempotent} if $e^2=e$.
An idempotent $e$ is said to be central if $e\in{\rm Z}(R)$, where
${\rm Z}(R)$ denotes the center of~$R$, which is a subring of $R$. 

\bex\label{idempotent}
(1)~ A left ideal $I$ of $R$ is complementary if and only if $I$ is generated by 
an idempotent $e$ (i.e., $I=Re$). 
\hint{if $R=I\oplus I'$, then $1=e+e'$ for unique $e\in I$, $e'\in I'$, then
$e=e^2+ee'$, hence $e^2=e$ and $ee'=0$; similarly, $e'^2=e'$ and $e'e=0$.}

(2)~ 
A two-sided ideal $I$ of $R$ is complementary if and only if $I$ is generated by 
a central idempotent $e$ (i.e., $I=Re$).
\hint{if $R=I\oplus I'$, by (1), $1=e+e'$, $I=Re$ and $I'=Re'$;
for any $a\in R$, $a=ae+ae'=ea+e'a$, $ae\in I$, $ea\in I$, 
$ae'\in I'$, $e'a\in I'$; so $ae=ea$; $e\in{\rm Z}(R)$.}
\eex

In a ring $R$, $e_1,\cdots,e_m$ are said to be a 
{\em complete system of central idempotents} 
if $e_1,\cdots,e_m$ are non-zero central idempotents, $e_ie_j\!=\!0$ 
for any $1\!\le\! i\!\ne\! j\le\! m$, and $1=e_1+\cdots+e_m$.

\begin{lemma}\label{complete system idempotent} 
Let $R$ be a ring. The following two are equivalent to each other:
\begin{itemize}
\item[\rm(1)]\vskip-5pt
 $R\!=\!R_1\oplus\cdots\oplus R_m$ with $R_i\ne 0$ for $i\!=\!1,\cdots,m$ being
 two-sided ideals of~$R$.
\item[\rm(2)]\vskip-5pt
 There is a complete system of central idempotents $e_1,\cdots,e_m$ 
 such that $R_i=Re_i$, $i=1,\cdots,m$.
\end{itemize}
\vskip-5pt
If it is the case, then 
each $R_i=Re_i$, $i=1,\cdots,m$, is a ring with identity $e_i$ 
(but not a subring of $R$ in general because $e_i\ne 1_R$ in general),  
and for any left ideal $L$, $L=Le_1\oplus\cdots\oplus Le_m$
with each $Le_i=L\cap R_i$ being a left ideal of $R_i$.
\end{lemma}

\begin{proof}
(1) $\Rightarrow$ (2). By the direct sum in (1), 
we have unique $e_i\in R_i$ such that $1_R=e_1+\cdots+e_m$. 
Then $ae_1+\cdots+ae_m=a1_R=a=1_Ra=e_1a+\cdots+e_ma$
with $ae_i,e_ia\in R_i$. By the direct sum again,
we get $ae_i=e_ia$ for $i=1,\cdots,m$. 
Letting $a$ run over $R$, we see that $e_i\in{\rm Z}(R)$. 
Taking $a=e_j$, we obtain that: $e_ie_j=e_j$ if $i=j$, and $e_ie_j=0$ if $i\ne j$.
Finally, taking $a\in R_i$, we get that $a=ae_i$; hence $R_i=Re_i$.

(2) $\Rightarrow$ (1). 
As $e_i\!\in\!{\rm Z}(R)$, $R_i\!=\!Re_i$ is a two-sided ideal.
For any ${a\!\in\!R}$, $a=a1_R=ae_1+\cdots+ae_m$.
Thus $R=R_1+\cdots+R_m$. 
Assume that $a=a_1+\cdots+a_m$ with $a_i\in R_i=Re_i$; 
we show that $a_j=ae_j$ for $j=1,\cdots,m$, hence the sum is a direct sum.
First, $a_i=be_i$ for some $b\in R$, hence $a_i=be_i=be_ie_i=a_ie_i$
(in other words, $e_i$ is the identity of $R_i$);
then $ae_j=a_1e_1e_j+\cdots+a_je_je_j+\cdots+a_me_me_j=a_je_j=a_j$.

Assume that (1), (2) hold. In the above proof 
we have seen that $R_i$ is a ring with identity $e_i$.
For $a\in L\cap R_i$, $a=ae_i\in Le_i$. Thus $Le_i=L\cap R_i$.
Finally, for any $a\in L$, 
$a=a1_R=ae_1+\cdots+ae_m\in Le_1\oplus\cdots\oplus Le_m$.
So $L=Le_1\oplus\cdots\oplus Le_m$.
\end{proof}

In the above lemma, $R_i=Re_i$ is called the {\em $e_i$-component} of~$R$,
$Le_i$ is called the {\em $e_i$-component} of $L$, 
$ae_i$ is called the {\em $e_i$-component} of $a$, etc.

An idempotent $e$ of a ring $R$ is said to be {\em primitive} if $e\ne 0$, 
and there are no idempotents $e'\ne 0\ne e''$ such that $e=e'+e''$
and $e'e''=0=e''e'$.
And $e$ is said to be a {\em central primitive} 
(or {\em primitive central}) idempotent of $R$
if $e$ is a primitive idempotent of the subring ${\rm Z}(R)$.

\bex\label{e e'}
(1) If $e,e'$ are central idempotents of a ring $R$
and $e$ is primitive central, then either $ee'=e$ or $ee'=0$.
\hint{$e=e'e+(e-e'e)$, $e'e(e-e'e)=0$.}

(2) If $e,e'$ are both primitive central idempotents, 
then either $e=e'$ or $ee'=0$. 

(3) If $e_1,\cdots,e_m$ are a complete system of primitive central idempotents of~$R$,
then $e_1,\cdots,e_m$ are all the distinct primitive central idempotents of~$R$.
\hint{for any primitive central idempotent $e$, $e=e1=ee_1+\cdots+ee_m$;
cite (2) above.}

(4) If $R$ is a finite ring, then the primitive central idempotent exists. 
\hint{if $1$ is not primitive central, then $1=e+e'$, $ee'=0$, $e\ne 0\ne e'$;
so $R=Re\oplus Re'$, $|Re|<|R|$.} 
\eex

\begin{lemma}\label{finite primitive}
Let $R$ be a finite ring. Then
$e_1,\cdots,e_m$ are a complete system of primitive central idempotents of $R$
if and only if $e_1,\cdots,e_m$ are all the distinct primitive central idempotents of~$R$.
\end{lemma}

\begin{proof} 
If $e_1,\cdots,e_m$ are all distinct primitive central idempotents of~$R$, then,
by Exercise \ref{e e'}, $e_ie_j=0$ for $1\le i\ne j\le m$;  
and, setting $e:=e_1+\cdots+e_m$, we must have $e=1$; because:
otherwise by Lemma \ref{complete system idempotent}
we have $R=Re\oplus Re'$ with $e'=1-e\ne 0$ and
$Re'$ is a finite ring with identity $e'$,
hence $Re'$ has a primitive central idempotent other than any $e_i$ for $i=1,\cdots,m$. 
\end{proof}

A ring $R$ is said to be {\em simple} if $R$ has no ideals other than $0$ and $R$.
A ring $R$ is said to be {\em semisimple} 
if any left ideal of $R$ is complementary. 

\bex\label{simple}
A commutative simple ring is a field. 
\hint{$Ra\!=\!R$ for any $0\!\ne\! a\in\! R$.}
\eex

\begin{theorem}\label{c semisimple}
If $R$ is a finite commutative semisimple ring and 
 $e_1,\cdots,e_m$ are all the primitive central idempotents, then: 

{\rm(1)} 
$R=Re_1\oplus\cdots\oplus Re_m$, and $Re_i$, $i=1,\cdots,m$, are finite fields.

{\rm(2)} For any ideal $0\ne I\le R$, 
$I=\bigoplus_{e\in E_I}Re=Re_I$, 
where $E_I=\{e_i\,|\,1\le i\le m,\,Ie_i\ne 0\}$ and $e_I=\sum_{e\in E_I}e$.  
In particular,  $I$ is a ring with identity $e_I$.

{\rm(3)} Set $E'_I=\{e_1,\cdots,e_m\}-E_I$ and $I'=\bigoplus_{e'\in E'_I}Re'$, 
then $R=I\oplus I'$, hence the quotient module $R/I\cong I'$.
\end{theorem}

\begin{proof}
(1). The first conclusion follows from 
Lemma \ref{finite primitive} and Lemma \ref{complete system idempotent}.
Any non-zero ideal $J$ of $Re_i$ is an ideal of $R$, by Exercise \ref{idempotent}(1), 
$J=Re$ for an idempotent $e\ne 0$.
Since $e\in J\subseteq Re_i$, $e_ie=e\ne 0$. 
But $e_i$ is primitive, by Exercise \ref{e e'}, $e_i=e_ie$.
So, $e=e_i$ and $J=Re_i$. 
Thus, $Re_i$ is a commutative simple ring, hence is a field 
(by Exercise~\ref{simple}). 

(2). By Lemma \ref{complete system idempotent}, 
$I=(I\cap Re_1)\oplus\cdots\oplus(I\cap Re_m)$. 
Since $Re_i$ is a field, $I\cap Re_i=0$ or $Re_i$.
Let $E_I=\{e_i\,|\,1\le i\le m,\, I\cap Re_i=Re_i\}$.
Then  $I=\bigoplus_{e\in E_I}Re=Re_I$.  

(3) follows from (1) and (2) immediately, 
\end{proof}

A remark: the finiteness assumption in 
Lemma \ref{finite primitive} and Theorem \ref{c semisimple}
can be replaced by the so-called Artin's condition. 

\medskip
Turn to the group algebra $\F G$, where $\F$ is a finite field with $|\F|=q$
and $G$ is a finite group of order $n>1$. 

\bex\label{G subset FG}
(1) $\F\to\F G$, $\alpha\mapsto \alpha 1_G$, 
is an injective ring homomorphism.

(2) $G\to(\F G)^\times$, $g\mapsto \sum_{x\in G}\gamma_x x$,  
where $\gamma_x=\begin{cases}1, & x=g;\\ 0, &x\ne g; \end{cases}$
is an injective group homomorphism.
\eex

By Exercise \ref{G subset FG} we can think of that
$\F\subseteq\F G$ and $G \subseteq\F G$.

\bex [{\sf The universal property of group algebras}]
\label{e universal} 
Let $R$ be any ring.

(1)~ 
If $\rho: \F G\to R$ is a ring homomorphism, then
$\rho(\F)$ is a subring of $R$ which is isomorphic to $\F$, 
 $\rho|_G: G\to R^\times$ is a group homomorphism,
and $\rho(G)$ commutes with $\rho(\F)$ element-wise.

(2)~ If ${\rm Z}(R)$ has a subring isomorphic to $\F$ 
(the subring is denoted by $\F$ again, 
 such a ring $R$ is called an {\em $\F$-algebra}),  
and $\tau: G\to R^\times$ is a group homomorphism,
then $\rho:\F G\to R$, 
$\rho\big(\sum_{x\in G}a_xx\big)=\sum_{x\in G}a_x\tau(x)$, 
is a ring homomorphism.
\eex 

\begin{lemma}[{\bfseries Maschck Theorem}] \label{Maschck}
If $\gcd(n,q)=1$, then $\F G$ is semisimple.
\end{lemma}

\begin{proof}
Let $L$ be a left ideal of $\F G$. There is an $\F$-subspace $S$ of $\F G$ 
such that $\F G=L\oplus S$ (subspace direct sum).
Assume that $\rho: \F G\to L$ is the projection from $\F G$ onto $L$
with kernel $S$. Set $\overline{\rho}=\frac{1}{n}\sum_{x\in G}x\rho x^{-1}$, i.e., 
for $a\in \F G$, $\overline{\rho}(a)=\frac{1}{n}\sum_{x\in G}x\rho(x^{-1}a)$. 
Then $\overline{\rho}(\F G)\subseteq L$; and, 
for $b\in L$, $\overline{\rho}(b)=\frac{1}{n}\sum_{x\in G}x\rho (x^{-1}b)
  =\frac{1}{n}\sum_{x\in G}x( x^{-1}b)=b$. 
So, $\overline{\rho}$ is a projection of $\F G$ onto $L$. 
For any $y\in G$, $y\overline{\rho} y^{-1}=\overline{\rho}$, 
i.e., $y\overline{\rho}=\overline{\rho} y$. 
Thus $\overline{\rho}$ is an 
$\F G$-endomorphism of~$\F G$.  Let $L'$ be the kernel of $\overline{\rho}$.
Then $L'$ is an $\F G$-submodule and $\F G=L\oplus L'$. 
\end{proof}

\begin{theorem}\label{semisimple abelian}
Let $G$ be an abelian group of order $n>1$ with $\gcd(n,q)=1$.

{\rm(1)} $e_0=\frac{1}{n}\sum_{x\in G}x$ is a primitive idempotent and 
$\F Ge_0=\F e_0\cong \F$.

{\rm(2)} Let $E=\{e_0=\frac{1}{n}\sum_{x\in G}x,~e_1,\cdots,e_m\}$
be the set of all the primitive idempotents of $\F G$. Then the simple ideals
$\F_i:=\F Ge_i$ for $i=0,1,\cdots,m$ are field extensions over $\F$
with $\dim \F_i=d_i\ge 1$, and
\begin{equation}\label{FG=F_0+...}
 \F G = \F_0\;\oplus\;\F_1\;\oplus\;\cdots\;\oplus\;\F_m.
\end{equation}

{\rm(3)} For any ideal ($\F G$-code) $C$ of $\F G$,  
$C=\bigoplus_{e\in E_C}\F Ge=\F Ge_C$, 
where $E_C=\{e_i\,|\,0\le i\le m,\,Ce_i\ne 0\}$ 
and $e_C=\sum_{e\in E_C}e$. In particular, $C$ is a ring with identity $e_C$.

{\rm(4)} Set $E'_I=E-E_I$ and $C'=\bigoplus_{e'\in E'_I}\F Ge'$, 
then $\F G=C\oplus C'$, hence the quotient module $\F G/C\cong C'$.
\end{theorem}

\begin{proof}
(1). For $y\!\in\! G$, $ye_0\!=\!e_0$. So $e_0e_0\!=\!e_0$, and 
$\F Ge_0\!=\!\{fe_0\,|\,f\!\in\! \F\}\!=\!\F e_0$.

(2), (3) and (4) follow from Lemma \ref{Maschck} and 
Theorem \ref{c semisimple} immediately.
\end{proof}

\begin{definition}\label{q-coset}\rm
By ${\Bbb Z}_n$ we denote the residue integer ring modulo $n>1$.
Assume that $\gcd(n,q)=1$, i.e., $q\in{\Bbb Z}_n^\times$. 
Let $\langle q\rangle_n$ 
be the cyclic subgroup 
of the multiplicative unit group ${\Bbb Z}_n^\times$ generated by $q$.
For $q^i\in \langle q\rangle_n$, 
the map $a\mapsto aq^i$ for $a\in{\Bbb Z}_n$ 
is a permutation of ${\Bbb Z}_n$. So the group $\langle q\rangle_n$ acts 
 on ${\Bbb Z}_n$. The $\langle q\rangle_n$-orbits 
 of ${\Bbb Z}_n$ are called the {\em $q$-cosets} of  ${\Bbb Z}_n$. 
 For $a\in{\Bbb Z}_n$, the $q$-coset~$Q_a$ containing $a$ is
 $Q_a=\{a, qa, q^2a,\cdots\}$. In particular, $\{0\}$ is a $q$-coset, called 
 the {\em trivial $q$-coset}.  By $\mu_q(n)$ we denote the minimal size of
 non-trivial $q$-cosets of~${\Bbb Z}_n$.
\end{definition}

\bex\label{cyclic mu(n)} 
Keep the notation in Theorem \ref{semisimple abelian}. 
Further assume that $G$ is cyclic.

(1) The $\F Ge_i$'s for $e_i\in E$ 
are 1-1 corresponding to the $q$-cosets $Q_i$'s of ${\Bbb Z}_n$ 
such that $\dim_\F(\F Ge_i)=|Q_i|$.
\hint{$\F Ge_i$ corresponds to a monic irreducible factor $f_i(X)$ of
$X^n-1$ such that $\F Ge_i\cong \F[X]/\langle f_i(X)\rangle$; 
the monic irreducible factors $f_i(X)$ of $X^n-1$ are 1-1 
corresponding to the $q$-cosets $Q_i$ of ${\Bbb Z}_n$
such that $f_i(X)=\prod_{k\in Q_i}(X-\xi^k)$ where $\xi$ is a primitive 
$n$-th root of $1$; cf. \cite[\S3.7,\S4.1]{HP}.}

(2) 
If $n'\,|\,n$ and $n'>1$, then $\mu_q(n')\ge\mu_q(n)$. 
\hint{The cyclic group $G'$ of order $n'$ is 
a quotient group of the cyclic group $G$ of order $n$;
by Exercise \ref{e universal}(2), there is a surjective homomorphism
$\F G\to\F G'$; cite Theorem \ref{semisimple abelian}(4).}

(3) (\cite{BM})~ $\mu_q(n)\!=\!
 \min\big\{\mu_q(p)\,\big|\,\mbox{$p$ runs over the prime divisors of $n$}\big\}$.
\hint{let $n\!\nmid\! a$, $Q=\{a,aq,\cdots,aq^{d-1}\}$ be 
a $q$-coset with size $d=\mu_q(n)$;
there are $p$ and $r>s\ge 0$ such that $p^{r}||n$, $p^{s}|| a$; then
$\{\frac{a}{p^{s}},\frac{a}{p^{s}}q,\cdots,\frac{a}{p^{s}}q^{d-1}\}$
is a $q$-coset of ${\Bbb Z}_p$ with size $\le d$; 
so $\mu_q(n)\ge\mu_q(p)$.}

(4) For a prime $p$,  $\mu_q(p)\!=\!{\rm ord}_{p}(q)$, 
 where ${\rm ord}_{p}(q)=|\langle q\rangle_p|$ is
 the order of $q$ in the multiplicative unit group ${\Bbb Z}_p^\times$. 
\hint{the non-trivial $q$-cosets are the cosets of the subgroup 
$\langle q\rangle_p$ in ${\Bbb Z}_p^\times$.}
\eex

\begin{lemma} \label{c semisimple abelian}
Keep notation in Theorem \ref{semisimple abelian},
in particular, $d_i=\dim_\F(\F_i)$, $1\le i\le m$. Then
$\min\{d_1,\cdots,d_m\}=\mu_q(n)$. 
\end{lemma}
\begin{proof} This result was proved in \cite{BM} for binary case.
We show an alternative proof for the general case. 
Note that, by Exercise~\ref{cyclic mu(n)}(1), the conclusion holds 
for finite cyclic group $G$. 

Let $1\le i\le m$, and 
$\rho_i: \F G=\bigoplus_{j=0}^{m}\F Ge_j \to\F Ge_i=\F_i$ be the projection.
Since $\F_i$ is a finite field, $\F_i^\times$ is a
cyclic group; by Exercise \ref{e universal}(1), $G':=\rho_i(G)$ 
is a cyclic group of order $n'>1$ with $n'|n$ (if $n'=1$ then
$\rho_i(x)=1$ for all $x\in G$, hence $i=0$).
By Exercise \ref{e universal}(2) we get a surjective homomorphism: $\F G'\to\F_i$. 
By Theorem~\ref{semisimple abelian}(4), 
$\F_i$ is isomorphic to a simple ideal of $\F G'$.
Since $G'$ is cyclic, $d_i\ge\mu_q(n')$.
Then, by Exercise~\ref{cyclic mu(n)}(2), $d_i\ge\mu_q(n)$, 
for all $i=1,\cdots,m$.

By Exercise~\ref{cyclic mu(n)}(3), there is a prime $p|n$ and
${\mu_q(n)=\mu_q(p)}$. Then $G$ has a quotient group 
$G'$ of order $p$,  hence $G'$ is cyclic and $\F G'$ has  
a non-trivial simple ideal $I'$ such that $\dim_F I'=\mu_q(p)$. 
By Exercise \ref{e universal}(2), 
we have a surjective homomorphism $\F G\to\F G'$.
By Theorem \ref{semisimple abelian}(4), there is an $i$ with $1\le i\le m$
such that $I'\cong\F_i$ (as $\F G$-modules). Thus, $d_i=\mu_q(p)=\mu_q(n)$. 
\end{proof}

\begin{lemma}\label{l n_1,...}
There exist positive integers $n_1,n_2,\cdots$ coprime to $q$ such that
$\textstyle\lim\limits_{i\to\infty}\frac{\log_q\! n_i}{\mu_q\!(n_i\!)}=0.$
\end{lemma}

\begin{proof}
By Exercise~\ref{e n_1,...}, there exist primes $p_1,p_2,\cdots$ such that
$\lim\limits_{i\to\infty}\frac{\log_q p_i}{\mu_q(p_i)}\!=\!0$. 
\end{proof}

\begin{remark}\label{r n_1, ...}\rm
(1). Though $p_1,p_2,\cdots$ in the proof are primes,  the integers 
 $n_1,n_2,\cdots$ in the lemma are not necessarily primes;
for example, taking $n_i\!=\!p_i^2$, by Exercise~\ref{cyclic mu(n)}(3) 
we still have 
$\lim\limits_{i\to\infty}\frac{\log_q n_i}{\mu_q(n_i)}\!=\!0$.

(2). By Exercise~\ref{cyclic mu(n)}(3) again, the integers 
 $n_1,n_2,\cdots$ in the lemma are almost odd; because: 
 $\mu_q(2)=1$ hence $\mu_q(n)=1$ for any even $n$.
\end{remark}

\bex\label{e n_1,...}
Let $t>q$ be an integer,  and $\pi(t)$ be the number of the primes 
$\le t$. Set
\begin{align*}
{\cal G}_t =\big\{\mbox{\rm prime $p$}\,\big|\, 
  q<p\le t,~ \mu_q(p)\ge(\log_q t)^2\big\};\\
\overline{\cal G}_t =\big\{\mbox{\rm prime $p$}\,\big|\, 
 q<p\le t,~\mu_q(p)<(\log_q t)^2\big\}.
\end{align*}

(1) If $r<(\log_q t)^2$, $p_1,\cdots,p_k\in \overline{\cal G}_t$ such that
$\mu_q(p_i)=r$, $i=1,\cdots,k$, then $k<(\log_q t)^2$.
\hint{``$\mu_q(p_i)=r$'' implies ``$p_i|(q^r\!-\!1)$''; 
hence $q^r\!-\!1=p_1\cdots p_k s$.}

(2)~ The natural density
$\lim\limits_{t\to\infty}\frac{|{\cal G}_t|}{\pi(t)}=1$.
\hint{by Gauss' Lemma, $\lim\limits_{t\to\infty}\frac{\pi(t)}{t/\ln t}=1$;
by~(1), $|\overline{\cal G}_t|<(\log_q t)^4$; so
$\lim\limits_{t\to\infty}\frac{|\overline{\cal G}_t|}{\pi(t)}
=\lim\limits_{t\to\infty}\frac{|\overline{\cal G}_t|}{t/\ln t}=0$.}
\eex

We'll see that $\mu_q(n)$ plays a key role in discussion of asymptotic
properties of group codes.

\begin{definition}\label{d bar map}\rm
(1)~ Mapping $x\in G$ to $x^{-1}\in G$ is an anti-automorphism of the group $G$.
For $a=\sum_{x\in G} a_xx\in\F G$, let
$\overline a=\sum_{x\in G} a_xx^{-1}$.
Then mapping $a\mapsto \overline a$ 
is an anti-automorphism of $\F G$ of order $2$ 
(i.e., it is bijective linear transformation 
and Exercise \ref{e bar map}(1) holds).
We call the anti-automorphism $a\mapsto\bar a$ by ``bar map''.
Note that the bar map is an automorphism if $G$ is abelian.

(2)~ For $a=\sum_{x\in G}a_xx\in\F G$,  
the map $\sigma$: $\F G\to \F$,~ $a\mapsto a_{1_G}$
($1_G$ is the identity of $G$), is a linear form on $\F G$.
\end{definition}

\bex\label{e bar map}
(1)~ $\overline{\overline a}=a$,~ $\overline{ab}=\overline b\overline a$,~ 
for $a,b\in\F G$.

(2)~ $\sigma(\overline a)=\sigma(a)$,~ $\sigma(ab)=\sigma(ba)$,~
for $a,b\in\F G$.

(3)~ $\langle a,b\rangle=\sigma(a\overline b)=\sigma(\overline a\,b)$, 
  for $a,b\in\F G$.
($\langle a,b\rangle$ denotes the inner product.)

(4)~ $\langle da,b\rangle=\langle a,\overline d b\rangle$,~
for $a,b,d\in\F G$.
\eex

\begin{lemma}\label{l orthogonal and bar}
Assume that $C,D$ are $\F G$-codes.

{\rm(1)}~ $C^\bot:=\{a\in\F G\,|\,\langle c,a\rangle=0,\; \forall\,c\in C\}$ 
is still an $\F G$-code.

{\rm(2)}~ $\langle C,D\rangle=0$ if and only if 
$C\overline D=0$ ($C\overline D:=\{c\,\overline d\,|\,c\in C, d\in D\}$).

{\rm(3)}~ If both $C, D$ are generated by one element, i.e., 
 $C=\F Ge$ and $D=\F Gf$, then $\langle C,D\rangle=0$  if and only if 
 $e \overline{f}=0$.
\end{lemma}

\begin{proof}
(1) follows from Exercise \ref{e bar map}(4).

(2). By Exercise \ref{e bar map}(3),
  $C\overline D=0$ implies that $\langle C,D\rangle\!=\!0$.
Conversely, assume that $c\,\overline d\ne 0$ for $c\in C$, $d\in D$;
then $c\,\overline d =\sum_{x\in G}\lambda_x x$ with $\lambda_z\ne 0$
 for some $z\in G$; 
then $z^{-1}c\in C$ and 
$\langle z^{-1}c,d\rangle=\sigma(z^{-1}c\,\overline d)=\lambda_z\ne 0$.

(3) holds by (2) immediately. 
\end{proof}

\begin{definition}\label{quasi group code}\rm
The product 
$(\F G)^2:=\F G\times\F G=\{(a,b)\,|\,a,b\in\F G\}$ 
is an $\F G$-module.
Any submodule $C$ of $(\F G)^2$, denoted by $C\le(\F G)^2$, is called a
{\em quasi-group code} (or {\em quasi-$\F G$ code}) of {\em index} $2$. 
 If $G$ is cyclic (abelian), then $C\le(\F G)^2$ is also called a 
 {\em quasi-cyclic} ({\em quasi-abelian}) code of index $2$.

More generally, for any integer $t\ge 1$, the
quasi-group codes (quasi-$\F G$ codes) of {\em index $t$} 
can be defined similarly. 
Quasi-group codes (quasi-$\F G$-codes) of index $t$
are also called {\em $t$-quasi-group codes} ({\em $t$-quasi-$\F G$ codes}).

In particular, the quasi-group codes of index~$1$ are just group codes.
\end{definition}

\bex\label{ex inner}
Let $(a,b), (a',b')\in(\F G)^2$. 
And $\langle -,-\rangle$ denotes the inner product.

(1) 
$\big\langle(a,b), (a',b')\big\rangle=\sigma(a\overline{a'}+b\overline{b'})$, 
where $\sigma$ is defined in Definition~\ref{d bar map}(2).

(2) $\big\langle \F G(a,b), \,\F G(a',b')\big\rangle=0$ 
if and only if $a\overline{a'}+b\overline{b'}=0$,
where $\F G(a,b)=\{(da,db)\,|\,d\in\F G\}\le(\F G)^2$
is the quasi-$\F G$ code of index $2$ generated by $(a,b)$. 
\eex

\begin{lemma}\label{l 2-quasi self-dual}
Let $(a,b)\in(\F G)^2$, and $C_{a,b}=\F G(a,b)\le(\F G)^2$. 

{\rm(1)} As $\F G$-modules, $C_{a,b}\cong \F G/{\rm Ann}_{\F G}(a,b)$, 
where ${\rm Ann}_{\F G}(a,b)=\{d\in\F G\,|\,da=db=0\}$.
In particular, $\dim C_{a,b}=n-\dim {\rm Ann}_{\F G}(a,b)$. 

{\rm(2)} $C_{a,b}$ is self-orthogonal if and only if $a\overline a+b\overline b=0$.

{\rm(3)} $C_{a,b}$ is self-dual if and only if $a\overline a+b\overline b=0$
 and ${\rm Ann}_{\F G}(a,b)=0$. 
\end{lemma}

\begin{proof}
(1) $\F G\to C_{a,b}$, $d\mapsto (da,db)$, is a surjective homomorphism.

(2) follows from Exercise \ref{ex inner}(2).

(3) follows from the above (1) and (2).
\end{proof}

In many cases the finite group $G$ we considered is assumed to be 
abelian, hence $\F G$ is a commutative ring. 
We sketch a few fundamentals of linear algebra 
over a commutative ring~$R$ with identity for later quotations. 
By $R^{k\times t}$ we denote the set of all $k\times t$ matrices over $R$;  
and $R^t=R^{1\times t}$ as usual.

\bex\label{e R}
Let $R$ be a commutative ring with identity.
${\bf a}_1,\cdots, {\bf a}_h\in R^t$ are said to be 
{\em linearly independent} if the linear combination 
$\alpha_1{\bf a}_1+\cdots+\alpha_h {\bf a}_h\ne 0$ provided 
$\alpha_1,\cdots,\alpha_h\in R$ are not all zero.
A matrix $A=(a_{ij})_{k\times t}\in R^{k\times t}$, $k\le t$,  
is said to be of {\em full rank} if the rows of $A$ are linearly independent.

(1) $b\!=\!(b_1,\cdots,b_k)\in R^k$ is linearly independent
if and only if ${{\rm Ann}_R(b_1,\cdots,b_k)\!=\!0}$, where
${\rm Ann}_R(b_1,\cdots,b_k)=\{a\in R\,|\,ab_1=\cdots=ab_k=0\}$.

(2) Given $A\in R^{k\times t}$. Then
$\tau_A: R^k\to R^t$, $b\mapsto bA$, is an $R$-homomorphism 
(recall that $R^k=R^{1\times k}$); 
$\tau_A$ is injective if and only if the matrix $A$ over $R$ is of full rank.

(3) Given $b=(b_1,\cdots,b_k)\in R^k$. 
Then $\tau_b: R^{k\times t}\to R^t$, $A\mapsto bA$, 
is an $R$-homomorphism and the image 
${\rm Im}(\tau_b)=(I_b)^t$, where $I_b:=Rb_1+\cdots+Rb_k$.
Further, if $R$ is semisimple, then
$\tau_b$ is surjective (i.e., $I_b=R$) if and only if
$b\in R^k$ is linearly independent.
\hint{if $I_b=R$ then ${\rm Ann}_R(b_1,\cdots,b_k)={\rm Ann}_R(I_b)=0$; conversely, if 
$I_b\ne R$ then $R=I'\oplus I_b$ with 
$0\ne I'\subseteq {\rm Ann}_R(b_1,\cdots,b_k)$.}
\eex

Because of Theorem \ref{semisimple abelian},
the following knowledge would be useful.

\bex\label{e R+R}
Assume that a commutative ring $R=\F_1\oplus \F_2$ is a direct sum
of two finite fields $\F_1$ and $\F_2$ with $|\F_i|=q^{d_i}$, i.e., 
any $a\in R$ can be uniquely written as $a=a^{(1)}+a^{(2)}$
with $a^{(i)}\in \F_i$, and 
$a+b=(a^{(1)}+b^{(1)})+(a^{(2)}+b^{(2)})$,  
$a\cdot b=(a^{(1)}\cdot b^{(1)})+(a^{(2)}\cdot b^{(2)})$.  
Assume that $1\le k\le t$.

(1) Let $b=(b_1,\cdots,b_k)\in R^k$, and 
$b=b^{(1)}+b^{(2)}$ with 
$b^{(i)}=\big(b_1^{(i)},\cdots,b_k^{(i)}\big)\in\F_i^k$.
Then $b$ is linearly independent
 if and only if $b^{(i)}\ne 0$ for $i=1,2$. 
 
(2) Let $a, b,\cdots,d\in R^k$, and $a=a^{(1)}+a^{(2)}$, $b=b^{(1)}+b^{(2)}$, $\cdots$,
$d=d^{(1)}+d^{(2)}$ as above. 
Then $a, b,\cdots,d$ are linearly independent
if and only if
$a^{(i)}, b^{(i)},\cdots,d^{(i)}$ are linearly independent over $\F_i$ for $i=1,2$.

(3) Any $A=(a_{ij})_{k\times t}\in R^{k\times t}$ 
 can be written as $A=A^{(1)}\oplus A^{(2)}$ with
 $A^{(i)}=\big(a_{ij}^{(i)}\big)_{k\times t}\in \F_i^{k\times t}$.  
 And, $A$ has full rank if and only if 
 $A^{(i)}$ has full rank for $i=1,2$.

(4) Select $A\in R^{k\times t}$ randomly. Then
 $\Pr\!\big(\mbox{$A$ has full rank}\big)\ge 
 (1-\frac{q^{d_1k}}{q^{d_1t}})(1-\frac{q^{d_2k}}{q^{d_2t}})$.
 \hint{cf. Exercise \ref{e rank=}(4).}
 
(5) Extend (1)--(4) to the direct sum of $m$ finite fields 
 $R=\F_1\oplus\cdots\oplus \F_m$. 
\eex

\section{Balanced codes}\label{balanced codes}

In this section, 
$\F$ is a finite field with $|\F|=q$,  $n>1$ is an integer,
$\delta\in[0,1-q^{-1}]$ is a real number.

For an index set $I=\{i_1,\cdots,i_d\}$, 
$\F^I=\{(a_{i_1},\cdots,a_{i_d})\,|\, a_{i_j}\in \F\}$.
As usual, $\F^n=\F^I$ with $I=\{1,2,\cdots,n\}$.
For any subset $I'\subseteq I$, there is a natural projection 
$\rho_{I'}:\F^I\to \F^{I'}$.

\begin{definition}\label{d balanced}\rm
Let $C\subseteq \F^n\!=\!\F^I$ with $I\!=\!\{1,\cdots,n\}$.
If there are subsets $I_1,\cdots,I_s$ (with repetition allowed) of 
$I$ and integers $k$, $t$ such that 

{(1)} the cardinality $|I_j|=k$ for $j=1,\cdots,s$;

{(2)} 
for any $i\in I$, the number of such subscripts
$j$ that $i\in I_j$ equals $t$;

\par\hangindent33pt {(3)}  
for any $j\!=\!1,\cdots,s$,
the projection $\rho_j\!:\F^I\to \F^{I_j}$ 
maps $C$ bijectively onto~$\F^{I_j}$;

\par\noindent\hangindent0pt
then, following  \cite{P85}, 
we say that $C$ is a {\em balanced code} over $\F$ of length $n$ and
{\em information length}~$k$, and $I_1,\cdots,I_s$ form
a {\em balanced system of information index sets} of $C$.
\end{definition}

Note that the phrase ``balanced codes'' might be used for
other concepts in literature, e.g., in \cite{IW}. 

\begin{remark}\label{cosets is ...}\rm
If $C\subseteq \F^n$ is a balanced code, 
then, for any $a\in\F^n$, $a+C=\{a+c\mid c\in C\}$ is a balanced code too.
Because: if the projection $\rho_j\!:\F^I\to \F^{I_j}$ 
maps $C$ bijectively onto~$\F^{I_j}$, then $\rho_j$ also maps
 $a+C$ bijectively onto~$\F^{I_j}$.
\end{remark}

\begin{theorem}\label{t balanced}
Let $C$ be a balanced code over $\F$ of length $n$ and information length $k$,
$B\subseteq C$ and $\omega=\sum_{b\in B}\frac{\w(b)}{n|B|}$
(the average relative weight of $B$).
If $0\le\omega\le 1-q^{-1}$, then
\begin{equation}\label{weight balanced}
 |B|\le q^{k h_q(\omega)}.
 \end{equation}
\end{theorem}

We'll prove the theorem after Lemma \ref{k=n} below. 
We state here a corollary and 
describe how to apply the theorem and the corollary to quasi-group codes. 
For any $C\subseteq \F^n$, we denote   
\begin{equation}\label{C^<delta}
 \textstyle C^{\le\delta}
 =\big\{c\;\big|\;c\in C,~\frac{{\rm w}(c)}{n}\le\delta\big\}.
\end{equation}

\begin{corollary}\label{c balanced}
If $C$ is a balanced code over $\F$ of length $n$ and information length $k$,
then $|C^{\le\delta}|\le q^{kh_q(\delta)}$. 
\end{corollary}
\begin{proof}
The average relative weight of $C^{\le\delta}$ is at most $\delta$,
and $h_q(\cdot)$ is an increasing function in $[0, 1-q^{-1}]$.
The corollary follows from Theorem \ref{t balanced} at once.
\end{proof}

In particular, $\F^n$ is itself a balanced code and  
$(\F^n)^{\le\delta}$ is just the Hamming ball 
centered at $0$ with radius $\lfloor\delta n \rfloor$,
see Eq.\eqref{Hamming ball}.
Corollary \ref{c balanced} is an extension of (part of)
Lemma \ref{<Hamming ball<}.

\begin{lemma}[{\cite[Lemma 2.2]{BM}}]\label{g-code balanced}
Group codes are balanced codes.
\end{lemma}

\begin{proof}
Let $G$ be a finite group, $C\le\F G$ be an $\F G$-code. 
Then $G$ is an $\F$-basis of $\F G$. Let 
$I\subseteq G$ be an information index set of the linear code $C$,
i.e., the projection $\rho_I:\F G\to \F I$ maps $C$ bijectively onto $\F I$.
Then one can check that 
$xI$, $x\in G$, form a balanced system of information index sets of $C$. 
\end{proof}

For $C\subseteq \F^n$,  the product $C^{n'}$ of $n'$ copies of $C$
in $(\F^n)^{n'}$ is as follows:
\begin{equation}\label{product}
 C^{n'}=\big\{({\bf c}_1,\cdots,{\bf c}_{n'})
  \;\big|\; {\bf c}_i\in C,~i=1,\cdots,n'\big\}\subseteq (\F^n)^{n'}
  =\F^{nn'}.
\end{equation}

\begin{lemma}\label{product balanced}
Let $C$ be a balanced code of $\F^n$
with information length $k$. Then $C^{n'}$
is a balanced code of $\F^{nn'}$ with information length $kn'$;
in particular, if $0\le\delta\le 1-q^{-1}$ then
$
\big|(C^{n'})^{\le\delta}\big|\le q^{kn'h_q(\delta)}.
$
\end{lemma}

\begin{proof} 
Let $\F^n=\F^I$ where $I=\{1,\cdots,n\}$.  We write
$$ I^{n'}=\big\{1^{(1)},\cdots,n^{(1)},~\cdots,~
   1^{(n')},\cdots,n^{(n')}\big\}.$$
Then $\F^{nn'}=\F^{I^{n'}}$.
Assume that subsets $I_1,\cdots,I_s$ of $I$ 
form a balanced system of information index sets of~$C$.
For each $I_j=\{j_1,\cdots,j_k\}$,
we have a subset $I_j^{n'}$ of $I^{n'}$ 
as follows:
$$ I_j^{n'}=\big\{j_1^{(1)},\cdots,j_k^{(1)},~\cdots,~
  j_1^{(n')},\cdots,j_k^{(n')}\big\}.  $$
Then one can check that
$I_1^{n'},~\cdots,~I_s^{n'}$
form a balanced system of information index sets
of the product code $C^{n'}$. 
\end{proof}

The following lemma is a key step for proving Theorem \ref{t balanced}.

\begin{lemma}\label{k=n}
 Let  $\emptyset\ne B\subseteq\F^n$,
$\omega=\sum_{b\in B}\frac{\w(b)}{n|B|}$.
If $0\le\omega\le 1-q^{-1}$, then
\begin{equation}\label{weight balanced}
 |B|\le q^{n h_q(\omega)}.
 \end{equation}
\end{lemma}

\begin{proof}
 Set $M=|B|$. Let 
$
{\cal B}=\begin{pmatrix} b_{11} &\cdots& b_{1j}&\cdots & b_{1n}\\
 \cdots& \cdots &\cdots& \cdots &\cdots\\ 
 b_{M1} &\cdots & b_{Mj}&\cdots & b_{Mn} \end{pmatrix}
$
be the matrix whose rows are all $b_i=(b_{i1},\cdots,b_{in})\in B$. 
For $\alpha\in\F$, let $t_{\alpha,j}$ denote the number of entries 
in the $j$-th column of ${\cal B}$ that equals~$\alpha$. 
For the $j$-th column of~${\cal B}$, $1\le j\le n$, we have a random
variable $X_j$ taking values in $\F$ with distribution 
$p_j(\alpha)=\Pr(X_j=\alpha)=t_{\alpha,j}/M$. 
By Definition \ref{d inf entropy}, 
Eq.\eqref{e inf entropy} and Lemma \ref{l inf entropy}, 
the joint variable $(X_1,.\cdots,X_n)$ and the
information entropies with base $q$ satisfy that 
\begin{equation}\label{H_q(X)}
\textstyle
 H_q(\!X_1,\cdots,X_n\!)\le H_q(\!X_1\!)+\cdots+H_q(\!X_n\!)
 = \sum_{j=1}^n\!\sum_{\alpha\in \F}-p_j(\alpha)\log_q\! p_j(\alpha). 
\end{equation} 
For any $a=(a_1,\cdots,a_n)\in \F^n$,  we see that the joint distribution
$$
 p(a_1,\cdots,a_n)= \Pr\big((X_1,\cdots,X_n)=(a_1,\cdots,a_n)\big)
 =\begin{cases}\frac{1}{M}, & (a_1,\cdots,a_n)\in B;\\
 0, & (a_1,\cdots,a_n)\notin B.\end{cases}
$$
By Definition \ref{d inf entropy}, 
\begin{equation}\label{log M}\textstyle
 H_q(X_1,\cdots,X_n) =\sum_{a\in\F^n}-p(a_1,\cdots,a_n)\log_q p(a_1,\cdots,a_n)=\log_q M.
\end{equation}
On the other hand, it can be checked directly that
\begin{equation}\label{omega}\textstyle
 \omega
 =\sum_{\alpha\in \F^\times}\sum_{j=1}^n \frac{p_j(\alpha)}{n};
 \qquad
1-\omega=\sum_{j=1}^n \frac{p_j(0)}{n}.
\end{equation}
Since $-x\log_q x$ is a concave function, 
by Jensen Inequality (Exercise \ref{Jensen}),  
\begin{align*} \textstyle 
\frac{\sum_{j=1}^n\sum_{\alpha\in\F^\times}
  \!-p_j(\alpha)\log_q p_j(\alpha)}{n(q-1)}
\le
-\frac{\sum_{j=1}^n\sum_{\alpha\in\F^\times}p_j(\alpha)}{n(q-1)}
 \log_q\!\frac{\sum_{j=1}^n
 \sum_{\alpha\in\F^\times} p_j(\alpha)}{n(q-1)}.
\end{align*}
By Eq.\eqref{omega},
\begin{align}\label{over F^times}\textstyle
\sum_{j=1}^n\sum_{\alpha\in\F^\times}-p_j(\alpha)\log_q p_j(\alpha)\le
n\big(\omega\log_q(q-1)-\omega\log_q\omega\big).
\end{align}
Similarly,  we have
\begin{align}\label{over 0}\textstyle
\sum_{j=1}^n -p_j(0)\log_q p_j(0)\le -n(1-\omega)\log_q(1-\omega).
\end{align}
Combining Eq.\eqref{over F^times}, Eq.\eqref{over 0},
Eq.\eqref{log M} and Eq.\eqref{H_q(X)},
$$
\log_q M \le n\big(\omega\log_q(q-1)
-\omega\log_q\omega-(1-\omega)\log_q(1-\omega)\big)
=nh_q(\omega).
$$
which is just the inequality (\ref{weight balanced}).
\end{proof}

\smallskip\noindent{\it A Proof of Theorem \ref{t balanced}.}
Assume that $C\subseteq \F^I$ with $I=\{1,\cdots,n\}$, and 
$I_1,\cdots,I_s$ is a balanced system 
of information index sets of $C$ as in Definition~\ref{d balanced},
and $B\subseteq C$, $\omega=\sum_{b\in B}\frac{\w(b)}{Mn}$ where
$M=|B|$.
By (1) and (2) of Definition \ref{d balanced}, 
$$
 ks=|I_1|+\cdots+|I_s|=nt.
$$ 
Let 
$$
V=\F^{I_1}\times\cdots\times\F^{I_s}
=\big\{(v_1,\cdots,v_s)\;\big|\; v_j\in\F^{I_j}, j=1,\cdots,s\big\},
$$
which is an $\F$-vector space of dimension $ks=nt$; i.e., $V=\F^{ks}$.
By (3) of Definition \ref{d balanced}, we have an $\F$-isomorphism 
$$
\rho:~C^s\longrightarrow V, ~~ 
(c_1, \cdots, c_s) \longmapsto \big(\rho_1(c_1),\cdots,\rho_s(c_s)\big).
$$
where $C^s=C\times\cdots\times C$ ($s$ copies) as in Eq.\eqref{product}.
For $B^s\subseteq C^s$, 
denote $\widehat B=\rho(B^s)\subseteq V$.
Then $|\widehat B|=|B|^s=M^s$. 

Let $b=(b_{1},\cdots,b_{n})\in B$. Let $b_i\ne 0$, i.e., 
the $i$-th entry of $b$ contributes $1$ to the total weight of $B$.
By (2) of Definition \ref{d balanced}, 
there are $t$ subscripts $1\le j_1<\cdots<j_t\le s$ such that 
$i\in I_{j_h}$, $h=1,\cdots,t$. For each $j_h$,
there are $M^{s-1}$ words $(\cdots, b,\cdots)\in B^s$
with $b$ appears in $j_h$-th position. Thus,
the contribution of the $i$-th entry $b_i $ of $b$ 
to the total weight of $\widehat B$ is $M^{s-1}t$. 

The total weight of $B$ equals $Mn\omega$. Then
the total weight of $\widehat B$ is 
$$
Mn\omega\cdot M^{s-1}t =M^s nt\omega.
$$
Note that $nt=ks=\dim V$.
So the average relative weight of $\widehat B$ equals $\omega$.
Applying Lemma \ref{k=n} to $\widehat B\subset V$, we get
$$
 M^s=|\widehat B|\le q^{ks h_q(\omega)}.
$$
That is, $|B|=M\le  q^{k h_q(\omega)}$. 
The proof of Theorem \ref{t balanced} is completed. 
\qed

\begin{remark}\rm
The binary version of Theorem \ref{t balanced} (Corollary \ref{c balanced}) 
was proved in \cite{M74}, \cite{P85} and \cite{S86}. 
The general versions stated here are proved in~\cite{FL}.
Because of Lemma~\ref{g-code balanced} and Lemma~\ref{product balanced},
Theorem \ref{t balanced} (in particular, Corollary \ref{c balanced}) 
can be cited to study the weights or distances of quasi-group codes.
\end{remark}

\section{Asymptotic property of quasi-abelian codes}\label{quasi-abelian}

In this section:
\begin{itemize}
\item\vskip-5pt
 $\F$ is a finite field with $|\F|=q$; 
 \item\vskip-5pt
 $n$ is a fixed positive integer such that $\gcd(n,q)=1$; 
 \item \vskip-5pt 
$G$ is an abelian group of order $n$;
\item \vskip-5pt
$r\in(0,1)$ and $\delta\in (0, 1-q^{-1})$ are real numbers; 
\item \vskip-5pt
$t>1$ is an integer, and $k=\lfloor{rt}\rfloor$.
\end{itemize}
\vskip-5pt
We consider the {\em quasi-$\F G$ codes of index $t$}, i.e., 
any $\F G$-submodule  
$$ 
  C\le (\F G)^t=
  \overbrace{\F G\times\cdots\times\F G}^{\mbox{\scriptsize $t$ copies}}
   =\{a=(a_1,\cdots,a_t)\,|\,a_i\in\F G\}.
$$
The order $n$ of $G$ is sometimes called the {\em coindex} of the 
quasi-abelian code~$C$ in literature. 
Each element ${a}=(a_1,\cdots,a_t)\in (\F G)^t$ is identified
with a concatenated word
$\big((a_{1x})_{x\in G},\cdots,(a_{t x})_{x\in G}\big)$
of length $nt$ over $\F$, the Hamming weight
${\w}({a})={\w}(a_1,\cdots,a_t)={\rm w}(a_1)+\cdots+{\rm w}(a_t)$.

We consider the set of $k\times t$ matrices over $\F G$:
\begin{equation}\label{ensemble}
(\F G)^{k\times t}
=\left\{A=\begin{pmatrix}a_{11}&\cdots&a_{1t}\\
 \cdots&\cdots&\cdots\\ a_{k1}&\cdots&a_{kt}\end{pmatrix}
 ~\Bigg|~ a_{ij}\in \F G\right\}, 
\end{equation}
as a probability space with each sample being chosen uniformly at random 
(i.e., at equal probability). 
For $A=(a_{ij})_{k\times t}\in (\F G)^{k\times t}$ 
(i.e., $A$ is a random $k\times t$ matrix over $\F G$), 
we write $A=(A_1,\cdots,A_t)$ with $A_j=(a_{1j},\cdots,a_{kj})^T$
being the $j$'th column of the matrix $A$,
where the superscript ``$T$'' stands for the transpose.
Then we have a 
{\em random quasi-abelian code $C_A$ of index $t$} as follows:
\begin{equation}\label{e C_A}
 C_A=\big\{{b}A=( bA_1,~\cdots,~ bA_t )
 ~\big|~ b=(b_1,\cdots,b_k)\in(\F G)^k\big\},
\end{equation}
where ${b}A_j=b_1a_{1j}+\cdots+b_ka_{kj}\in \F G$.
In particular, if $G=1$ is trivial then $(\F G)^{k\times t}=\F^{k\times t}$ 
and $C_A$ is just the usual random linear code defined in Definition \ref{C_M}.

After the random quasi-$\F G$ code $C_A$ is defined, 
 both the rate $\R(C_A)=\frac{\dim_F C_A}{nt}$ and 
the relative minimum distance $\Delta(C_A)=\frac{{\rm w}(C_A)}{nt}$ 
are random variables. We consider the asymptotic behaviors 
of them with $t\to\infty$ (hence $k\to\infty$). 
The following is the main result which shows, similar to
Theorem~\ref{Delta(C_M) >delta} on linear codes, 
that the GV-bound $g_q(\delta)$ (Definition \ref{d entropy})
is a phase transition point for the random quasi-$\F G$ codes.

\begin{theorem}\label{Delta(C_A) >delta}
Let $A\in(\F G)^{k\times t}$ and $C_A=\{{b}A\mid{b}\in(\F G)^k\}$
where $k=\lfloor rn\rfloor$. Then
$$\lim\limits_{t\to\infty}\Pr\big(\Delta(C_A)>\delta\big)=
\begin{cases}
    1, & r<g_q(\delta); \\  0, & r>g_q(\delta).
\end{cases}
$$
\end{theorem}

Before proving it, we show some consequences of the theorem.

By Exercise \ref{e R}(2), $\tau_A: (\F G)^k\to (\F G)^t$, $b\mapsto bA$, 
is an $\F G$-module homomorphism and $C_A={\rm Im}(\tau_A)$; 
so the rate 
$$\textstyle
 \R(C_A)\le \frac{k}{t}; \quad \mbox{ and }~~ 
  \R(C_A)=\frac{k}{t} \iff  \mbox{$A$ is of full rank}.
$$

Since $\gcd(n,q)=1$, 
Theorem \ref{semisimple abelian} and its notation are applied. 
We have $\F G=\F_0\oplus \F_1\oplus\cdots\oplus \F_m$
with $\F_i$ being field extensions over $\F$ with $\dim \F_i=d_i$
for $i=0,1,\cdots,m$, see Eq.\eqref{FG=F_0+...}.
By Exercise~\ref{e R+R}(4), 
\begin{equation}\label{e full rank}\textstyle
 \lim\limits_{t\to\infty}\Pr(\mbox{\rm $A$ has full rank})
 \ge\lim\limits_{t\to\infty}\prod_{i=0}^m(1-\frac{q^{d_ik}}{q^{d_it}})
  =1.
\end{equation}
Similarly to Theorem \ref{Delta(C_N) >delta}, we get

\begin{theorem}
$\lim\limits_{t\to\infty}
\Pr\big(\Delta(C_A)>\delta\,\big|\,\R(C_A)=\lfloor rt\rfloor/t\big)
=\begin{cases}
    1, & r<g_q(\delta); \\  0, & r>g_q(\delta).
\end{cases}$
\end{theorem}

As shown in Remark \ref{r linear}, 
the first part (the case $r<g_q(\delta)$) of the theorem implies
that the quasi-$\F G$ codes attain the GV-bound asymptotically.

\begin{corollary}
If $\delta\in(0,1-q^{-1})$ and $0<r<g_q(\delta)$, then
there exist quasi-$\F G$-codes $C_i$ of index $t_i$ such that
the code sequence $C_1,C_2,\cdots$ satisfies that $t_i\to\infty$, 
 $\lim\limits_{i\to\infty}\R(C_i)=r$ and 
 $\Delta(C_i)>\delta$ for all $i=1,2,\cdots$.
\end{corollary}

Now we turn to prove Theorem \ref{Delta(C_A) >delta}.
And, also as shown in Remark \ref{r linear}, 
only the first moment method is not enough to prove the theorem, 
we have to apply the second moment method to prove the second part 
(the case $r>g_q(\delta)$) of it. 
We begin with fixing some notation.

\begin{definition}\label{d X_b abelian}\rm
(1)~ For ${b}\in(\F G)^k$ we define a 0-1 variable 
$$X_{b}=
 \begin{cases} 1, & 0<{\rm w}({b}A)\le \delta nt;\\
 0, & {\rm otherwise}. \end{cases}$$
 
(2) Let $X=\sum_{{b}\in(\F G)^k}X_{b}$, which is a non-negative integer variable. 
 Then $X$ stands for the number 
 of the non-zero $bA\in C_A$ with relative weight $\le\delta$. So 
\begin{equation}\label{X ge 1}
 \Pr\big(\Delta(C_A)\le\delta\big)=\Pr(X\ge 1).
\end{equation}

(3)~ For ${b}=(b_1,\cdots,b_k)\in(\F G)^k$, we set
$I_{b}=\F Gb_1+\cdots+ \F Gb_k$ which is the ideal of $\F G$ generated by
$b_1,\cdots,b_k$, 
and denote $d_{b}=\dim I_{b}$.

(4)~ For any ideal $I$ of $\F G$, set $d_I=\dim I$, 
hence $\dim I^k=d_I k$; and set 
$$
 I^{k*}=\big\{{b}=(b_1,\cdots,b_k)\in I^k\;\big|\; I_{b}=I\big\}.
$$
In particular, $(\F G)^{k*}=\{b\in(\F G)^k\mid I_b=\F G\}$.
\end{definition}

Then $(\F G)^k$ can be written as a disjoint union as follows:
\begin{equation}\label{disjoint union}
\textstyle
(\F G)^k=\{0\}\bigcup\Big(\bigcup_{0\ne I\le \F G}I^{k*}\big),
\end{equation}
where the subscript ``$0\ne I\le \F G$'' means that $I$ runs over the
non-zero ideals of $\F G$. 
By Exercise \ref{e R}(3), Eq.\eqref{FG=F_0+...} and Exercise \ref{e R+R}(1), 
$\frac{|(\F G)^{k*}|}{|(\F G)^k|}
 =\prod_{i=0}^m\frac{|\F_i^{k*}|}{|\F_i^k|}$; 
hence
\begin{equation}\label{|FGk*|}\textstyle
\lim\limits_{k\to\infty}\frac{|(\F G)^{k*}|}{|(\F G)^k|}
=\lim\limits_{k\to\infty}\prod_{i=0}^m(1-q^{-d_ik})=1.
\end{equation}

\begin{lemma} \label{l EX_b}
{\rm(1)} For ${b}\in(\F G)^k$,~
${\E}(X_{b})\le q^{- d_{b}tg_q(\delta)}-q^{-d_{b}t}
< q^{- d_{b}tg_q(\delta)}$.

{\rm(2)} For ${b}\in(\F G)^{k*}$,~ 
 ${\E}(X_{b})=\big(|(\F^{nt})^{\le\delta}|-1\big)\big/q^{nt}$. 
 And the following hold:
 
\quad {\rm(2.1)} ${\E}(X_{b})={\E}(X_{b'})$,
   $\forall$ ${b},{b'}\in(\F G)^{k*}$.

\quad {\rm(2.2)} $q^{nk}{\E}(X_{b})\ge
  q^{nt\big(\frac{k}{t}-g_q(\delta)-\frac{\log_q(nt+1)}{nt}\big)}-1$, 
   $\forall$ ${b}\in(\F G)^{k*}$. 
\end{lemma}

\begin{proof} 
Since $\tau_{b}:(\F G)^{k\times t}\to(\F G)^{t}$, $A\mapsto bA$, 
is an $\F G$-homomorphism (Exercise~\ref{e R}(3))
and $|(\F G)^{k\times t}|=q^{nkt}$
and $|(I_b)^{t}|=q^{d_b t}$, 
the number of the pre-images in $(\F G)^{k\times t}$ of any
${a}\in (I_{b})^t$ is equal to $\frac{q^{nkt}}{q^{d_{b}t}}$.
By Definition~\ref{d X_b abelian}(1), we get that
\begin{equation}\label{e EX_b}
 \textstyle
 {\E}(X_{b})
 =\big(|((I_{b})^t)^{\le\delta}\big|-1)
  \cdot\frac{q^{nkt}}{q^{d_{b}t}}\Big/q^{nkt}
 =\big(|((I_{b})^t)^{\le\delta}|-1\big)\Big/q^{d_{b}t}.
\end{equation}

(1).
By Lemma \ref{g-code balanced} and Lemma \ref{product balanced},
$\big|((I_{b})^t)^{\le\delta}\big|\le q^{d_{b}th_q(\delta)}$.
(1) is obtained.

(2).
For ${b}\in(\F G)^{k*}$ we have $I_{b}=\F G$ and $d_{b}=n$.
The equality in~(2) follows from Eq.\eqref{e EX_b}. (2.1) is trivial. 
By Lemma \ref{<Hamming ball<}, 
$|(\F^{nt})^{\le\delta}|\ge 
q^{nt h_q(\delta)-\log_q(nt+1)}$, hence
$\E(X_b)\ge q^{-ntg_q(\delta)-\log_q(nt)}-q^{-nt}$.
Since $\frac{q^{nk}}{q^{nt}}\leq 1$,  (2.2) holds.
\end{proof}

\begin{theorem}\label{E(X)}
$\lim\limits_{t\to\infty}\E(X)=
  \begin{cases}
   0,& r<g_q(\delta);\\ \infty, & r>g_q(\delta). 
  \end{cases}$
\end{theorem}

\begin{proof}
Assume that $r<g_q(\delta)$.
By Lemma \ref{l E} and Eq.\eqref{disjoint union} we obtain 
$$\textstyle
{\E}(X)
={\E}\big(\sum _{b\in (\F G)^k} X_{b}\big)
=\sum _{b\in (\F G)^k}{\E}(X_{b})
=\sum_{0\ne I\le \F G}\sum_{b\in I^{k*}}
{\E}(X_{b}).
$$
For any ideal $I$ of $\F G$ and $b\in I^{k*}$,
$d_{b}= d_I$, $|I^{k*}|\le |I^k|=q^{d_I k}$, 
by Lemma~\ref{l EX_b}(1) we have
\begin{equation*}
\textstyle
\sum_{b\in I^{k*}}{E}(X_{b})
\le q^{d_Ik}q^{- d_I t g_q(\delta)}
=q^{d_I t \big(\frac{k}{t}-g_q(\delta)\big)}.
\end{equation*}
Since $k=\lfloor rt\rfloor$, we have
$\lim\limits_{t\to\infty}\big(\frac{k}{t}-g_q(\delta)\big)=r-g_q(\delta)<0$.
Note that $\F G$ has only finitely many ideals, we get
$$\textstyle
\lim\limits_{t\to\infty}{\E}(X)
=\sum_{0\ne I\le \F G}\lim\limits_{t\to\infty}\sum_{b\in I^{k*}}
{\rm E}(X_{b}) =0.
$$

Next we assume that $r>g_q(\delta)$.
Fixing a ${b_1}\in(\F G)^{k*}$,
from Eq.\eqref{disjoint union} and Lemma \ref{l EX_b}(2) we have:
$$\textstyle
{\E}(X)\ge\sum_{b\in(\F G)^{k*}}{\E}(X_{b})
=|(\F G)^{k*}|\cdot{\E}(X_{b_1})
=\frac{|(\F G)^{k*}|}{|(\F G)^k|}\cdot q^{nk}{\E}(X_{\bf b_1}).
$$
Since $k\to\infty$ as $t\to\infty$, by Eq.\eqref{|FGk*|}, 
we have $\lim\limits_{t\to\infty}\frac{|(\F G)^{k*}|}{|(\F G)^k|}=1$.
Note that 
$\lim\limits_{n\to\infty}\big(\frac{k}{n}-g_q(\delta)\big)=r-g_q(\delta)>0$,
by Lemma \ref{l EX_b}(2.2),
we obtain that
$$\textstyle
\lim\limits_{t\to\infty}{\E}(X)
\geq\lim\limits_{t\to\infty}\frac{|(\F G)^{k*}|}{|(\F G)^k|}
\big(q^{nt\big(\frac{k}{t}-g_q(\delta)-\frac{\log_q(nt+1)}{nt}\big)}-1
\big)
=\infty.
$$
The proof of the theorem is finished. 
\end{proof}

\begin{lemma}\label{Pr(X>=1)}
For any $b\in (\F G)^{k*}$,~ 
$\Pr(X\ge 1)\ge
 \frac{|\F G^{k*}|\cdot \E(X_{b})}{\E(X|X_{b}=1)}.
$
\end{lemma}

\begin{proof} Since  $(\F G)^{k*}$ is a part of $(\F G)^k$,
by Lemma \ref{2'nd moment} we get that
\begin{equation*}\textstyle
 \Pr(X\ge 1)
  \ge \sum_{b\in (\F G)^{k}}
    \frac{\E(X_b)}{\E(X|X_{b}=1)}
  \ge \sum_{b\in (\F G)^{k*}}
    \frac{\E(X_{b})}{\E(X|X_{b}=1)}.
\end{equation*}
In Lemma \ref{l EX_b}(2.1) we have seen that the value of ${\rm E}(X_{b})$
is independent of the choice of the element $b\in (\F G)^{k*}$.
We now claim that $\E(X|X_{b}=1)$
 is independent of the choice of $b\in (\F G)^{k*}$,
which completes the proof of the lemma. Set
$$
{\cal A}(a,b)=
\left\{A\in (\F G)^{k\times t}\,\big|\,
  1\le \w(aA),\w(bA)\le nt\delta\right\}, ~~a\in (\F G)^k.
$$
By the linearity of expectations (Lemma~\ref{l E}) 
and the conditional probability formula (Lemma \ref{conditional}),  
we have
\begin{align*}
 \E(X|X_{b}=1)
 &\textstyle
   = \E\Big(\sum_{a\in(\F G)^k} X_{a}\Big|X_{b}=1\Big)
   =\sum_{a\in(\F G)^k}\E(X_{a}|X_{b}=1) \\
 &\textstyle
  = \sum_{a\in(\F G)^k}
   \frac{\Pr\big(X_{a}=1\;\&\;X_{b}=1\big)}{\Pr(X_{b}=1)}
  = \sum_{a\in(\F G)^k}\frac{|{\cal A}(a,b)|}{|(\F G)^{k\times t}|}
       \cdot \frac{1}{\E(X_{b})}\\
&\textstyle
  = \frac{1}{q^{nkt}\cdot \E(X_b)}
  \sum_{a\in(\F G)^k}|{\cal A}(a,b)|\,.
\end{align*} 
For any $b'\in(\F G)^{k*}$,
by the definition of $(\F G)^{k*}$,
there exists an invertible
$k\times k$ matrix $Q$ over $\F G$ such that $b'=bQ$.
It is easy to verity that
$|{\cal A}(a,b)|=|{\cal A}(a,b'Q^{-1})|=|{\cal A}(aQ,b')|.$
By Lemma \ref{l E}(2.1), we have that
\begin{eqnarray*}
\textstyle
\E(X|X_{b'}\!=\!1)
=\frac{1}{q^{nkt}\cdot \E(X_{b'})}
   \sum_{a\in(\F G)^k}|{\cal A}(a,b')|
=\frac{1}{q^{nkt}\cdot\E(X_{b})}
   \sum_{a\in(\F G)^k}|{\cal A}(aQ,b)|.
\end{eqnarray*}
Noting that $aQ$ runs over $(\F G)^k$ when $a$ runs over $(\F G)^k$,
we obtain that
$\E(X|X_{b'}=1)=\E(X|X_{b}=1).$
\end{proof}

In the following we fix 
\begin{equation}\label{c=} 
 c=(0,\cdots,0,1)\in(\F G)^{k*}, 
\end{equation}
and estimate the value of the conditional expectation
$\E(X_{b}|X_{c}=1)$ for $b=(b_1,\cdots,b_k)\in (\F G)^k$.
Set
$$
 \bar{b}=(b_1,\cdots,b_{k-1})\in (\F G)^{k-1}, ~~~
  I_{\bar{b}}=\F Gb_1+\cdots+\F Gb_{k-1}, ~~~
  d_{\bar{b}}=\dim_F I_{\bar{b}}\,.
$$
\begin{lemma}\label{EXb|Xb1}
{\rm(1)} 
 $\E(X_{b}|X_{c}=1) \le q^{-d_{\bar{b}}t g_q(\delta)}$.
 
{\rm(2)} If $\bar{b}\in (\F G)^{(k-1)*}$, then
$\E(X_b|X_c=1)=\E(X_c)\le q^{-nt g_q(\delta)}$.
\end{lemma}

\begin{proof}
(1).  For any $A\in(\F G)^{k\times t}$,
we denote by $\underline A_i$ the $i$'th row of $A$.
Clearly,
\begin{align*}
& \E(X_{b}|X_{c}=1)
=\frac{\big|\{A\in(\F G)^{k\times t}\mid 1\le \w(bA),~
 \w(\underline A_k)\le nt\delta\}\big|}
 {q^{nkt}\cdot \E(X_{c})}\\
&~~=\frac{\big|\{A\in(\F G)^{k\times t}\mid 1\le\w(bA)\le nt\delta,~
 0\ne\underline A_k \in((\F G)^t)^{\le\delta}\}\big|}
{q^{nkt}\cdot \E(X_{c})}\,.
\end{align*}
where $((\F G)^t)^{\le\delta}
=\{a\in(\F G)^t\mid \w(a)\le nt\delta\}$,
see Eq.\eqref{C^<delta}.
Given any non-zero 
${a}_\delta=(a_{k1},\cdots,a_{kt})\in((\F G)^t)^{\le\delta}$,
we denote
$${\cal A}(c)_{{a}_\delta}=\left\{
A\in(\F G)^{k\times t}\,\big|\, \underline A_k=a_\delta\right\}.$$
Any ${b}=(b_1,\cdots,b_{k-1},b_k)\in (\F G)^k$ induces a map:
\begin{equation}
\begin{array}{cccl}
 \bar\tau_{b}: & {\cal A}(c)_{{a}_\delta}
   & \longrightarrow & (\F G)^t, \\
   &A & \longmapsto & bA
   =b_1\underline A_1+\cdots+b_{k-1} \underline A_{k-1}+b_k{a}_\delta\,.
\end{array}
\end{equation}
It is easy to see that the image of the map $\bar\tau_{b}$
is a coset of
$I_{\bar{b}}^t\subseteq(\F G)^t$ as follows
\begin{equation}\label{Ib}
I_{\bar{b}}^t+b_k{a}_\delta,\qquad{\rm with~ cardinality}~~
    |I_{\bar{b}}^t+b_k{a}_\delta|=|I_{\bar{b}}^t|
     =q^{d_{\bar{b}}t};
\end{equation}
and the number of the pre-images in
${\cal A}(c)_{a_\delta}$ of
any $a\in I_{\bar{b}}^t+b_k{a}_\delta$ is equal to
$\frac{q^{n(k-1)t}}{q^{d_{\bar b}t}}=q^{n(k-1)t-d_{\bar{b}}t}$,
 which is independent of the choices of ${a}_\delta$ and $a$.
Hence, we get that
\begin{eqnarray}\label{E(Xb|Xc=1)}
\E(X_b|X_c=1)
=\frac{\sum_{0\ne a_\delta\in((\F G)^t)^{\le\delta}}
\big(|(I_{\bar{b}}^t+b_k {a}_\delta)^{\le\delta}|-\lambda_{a_\delta}\big)
q^{n(k-1)t-d_{\bar{b}}t}}{q^{nkt}\cdot \E(X_c)},
\end{eqnarray}
where
\begin{equation}\label{lambda}
\lambda_{a_\delta}=\begin{cases}1, 
    & {\rm if}~{0}\in I_{\bar{b}}^t+b_k a_\delta;\\
 0, &\mbox{otherwise.}\end{cases}
\end{equation}
By Lemma~\ref{g-code balanced}, Lemma~\ref{product balanced}, 
Remark \ref{cosets is ...} and Corollary~\ref{c balanced},
we see that $I_{\bar{b}}^t+b_k a_\delta$ is a balanced code,
and $|(I_{\bar{b}}^t+b_k a_\delta)^{\le\delta}|
  \leq q^{d_{\bar{b}}t h_q(\delta)}$,
hence,
\begin{align*}
\E(X_b|X_c=1)
&
 \le\frac{\left(|((\F G)^t)^{\le\delta}|-1\right)\cdot
  q^{d_{\bar{b}} t h_q(\delta)}q^{n(k-1)t- d_{\bar{b}}t}}
  {q^{nkt}{\rm E}(X_{c})}\\
&
 =  \frac{q^{nt} \E(X_{c})\cdot
  q^{d_{\bar{b}} t h_q(\delta)}q^{n(k-1)t- d_{\bar{b}}t}}
  {q^{nkt}\cdot \E(X_c)}=q^{-d_{\bar{b}}n g_q(\delta)}.
\end{align*}

(2).  
Assume that $\bar{b}\in (\F G)^{(k-1)*}$.
Then $d_{\bar{b}}=n$, $I_{\bar{b}}^t=(\F G)^t$ and
$I_{\bar{b}}^t+b_k a_\delta=(\F G)^t$;
in particular, $\lambda_{a_\delta}=1$ in Eq.\eqref{lambda}.
So we obtain 
\begin{align*}
 \E(X_{b}|X_c=1)=
  \frac{\left(|((\F G)^t)^{\le\delta}|-1\right)\cdot q^{nt}
   \E(X_c)q^{n(k-1)t- d_{\bar{b}}t}}
  {q^{nkt}\cdot \E(X_c)}
  =\E(X_c).
\end{align*}
 We are done.
\end{proof}

We are ready to complete a proof of Theorem \ref{Delta(C_A) >delta}.

\medskip
\noindent{\it A proof of Theorem \ref{Delta(C_A) >delta}.}

By Eq.\eqref{X ge 1}, we need to prove that 
$\lim\limits_{t\to\infty}\Pr(X\ge 1)=
\begin{cases}0, & r< g_q(\delta);\\ 1, & r> g_q(\delta).  \end{cases}$

If $r< g_q(\delta)$, then by Markov Inequality (Lemma \ref{Markov}) 
and Theorem \ref{E(X)}, 
$$\lim\limits_{t\to\infty}\Pr(X\ge 1)
  \le \lim\limits_{t\to\infty}\E(X)=0. 
$$

Next we assume that  $r>g_q(\delta)$.
By Lemma \ref{Pr(X>=1)}, it is enough to prove that
\begin{equation}\label{aim}
\lim_{t\to\infty}\frac{\E(X|X_c=1)}{|(\F G)^{k*}|\cdot \E(X_c)}
=\lim_{t\to\infty}
\frac{\sum_{b\in(\F G)^k} \E(X_b|X_c=1)}
{|(\F G)^{k*}|\cdot {\rm E}(X_c)} = 1,
\end{equation}
where $c\!=\!(0,\cdots,0,1)$ as in Eq.\eqref{c=}.
Since $(\F G)^{k-1}=\bigcup_{I\le\F G}I^{(k-1)*}$ 
(see Eq.\eqref{disjoint union}, 
but this time we consider $(\F G)^{k-1}$), we have
\begin{eqnarray*}
\sum_{b\in(\F G)^k} \E(X_b|X_c=1)
=\sum_{I\le \F G}\,\sum_{\bar{b}\in I^{(k-1)*}}\,
 \sum_{b_k\in \F G} \E(X_b|X_c=1),
\end{eqnarray*}
where $b=(b_1,\cdots,b_{k-1},b_k)$
and $\bar{b}=(b_1,\cdots,b_{k-1})$.
Thus
\begin{equation}\label{S_I}
\frac{\sum_{b\in(\F G)^k} \E(X_b|X_c=1)}
{|(\F G)^{k*}|\cdot \E(X_c)}
=\sum_{I\le \F G} S_I,
\end{equation}
where 
$$
S_I=\frac{\sum_{\bar{b}\in I^{(k-1)*}}\,
\sum_{b_k\in \F G}\E(X_b|X_c=1)}{|(\F G)^{k*}|\cdot \E(X_c)},
\qquad I\le\F G.
$$
We compute $\lim\limits_{k\to\infty}S_I$ for ideals $I$ of $\F G$ in two cases.

{\it Case 1}:~ $I=\F G$, hence  $I^{(k-1)*}=(\F G)^{(k-1)*}$.
By Lemma \ref{EXb|Xb1} we have
\begin{eqnarray*}
S_I=\frac{|(\F G)^{(k-1)*}|\cdot q^n\cdot \E(X_c)}
  {|(\F G)^{k*}|\cdot \E(X_c)}
=\frac{|(\F G)^{(k-1)*}|}{|(\F G)^{k-1}|}\cdot
 \frac{|(\F G)^{k}|}{|(\F G)^{k*}|}\,.
\end{eqnarray*}
By Eq.\eqref{|FGk*|}, we get ($k\to\infty$ while $t\to\infty$)
\begin{equation*}
\lim_{t\to\infty}S_I =1.
\end{equation*}

{\it Case 2}:~ $I\ne \F G$.  Then $d_{\bar{b}}=d_I=\dim_F I<n$
for any ${\bar{b}}\in I^{(k-1)*}$,
and $|I^{(k-1)*}|\le |I^{(k-1)}|=q^{d_I(k-1)}$.
By Lemma \ref{l EX_b}(2.2) and Lemma \ref{EXb|Xb1}(1), we have
\begin{eqnarray*}
S_I
&\le&
\frac{q^{d_I(k-1)}q^n\cdot q^{-d_I t g_q(\delta)}}
 {\frac{|(\F G)^{k*}|}{|(\F G)^{k}|}
  \big(q^{nt\big(\frac{k}{t}-g_q(\delta)-\frac{\log_q(nt+1)}{nt}\big)}-1)}\\
&\le&q^n\cdot\frac{|(\F G)^{k}|}{|(\F G)^{k*}|}\cdot
\frac{q^{d_I t (\frac{k}{t}-g_q(\delta))}}
 {q^{nt\big(\frac{k}{t}-g_q(\delta)-\frac{\log_q(nt+1)}{nt}\big)}-1}\,.
\end{eqnarray*}
Since $k=\lfloor rn\rfloor$ and $1>r>g_q(\delta)$, then
$\lim\limits_{k\to\infty}\big(\frac{k}{t}-g_q(\delta)\big)=r-g_q(\delta)>0$.
Recalling that $d_I<n$ and
$\lim\limits_{k\to\infty}\frac{|(\F G)^{k}|}{|(\F G)^{k*}|}=1$
(see Eq.\eqref{|FGk*|}), we get that
\begin{equation*}
\lim_{t\to\infty}S_I =0.
\end{equation*}

Finally,  since there are only finitely many ideals of $\F  G$,
by Eq.\eqref{S_I} and the limits obtained in the above two cases, 
Eq.\eqref{aim} is proved.

\begin{remark}\rm
The results in this section are published in \cite{FL}. 
In \cite{FL} we didn't assume that 
the order of the abelian group $G$ is coprime to $q$; i.e.,
the group algebra $\F G$ may be not semisimple 
(in the non-semisimple case 
 the direct summands $\F Ge_i$ in Eq.\eqref{FG=F_0+...}
may be local rings). 
To make this chapter self-contained, 
here we avoid the non-semisimple case. 
However, the arguments here are in fact valid for 
the non-semisimple case.
\end{remark}

\section{Quasi-cyclic codes of fractional index}\label{fractional}

In this section we study the asymptotic properties of the so-called 
{\em quasi-cyclic codes of index $1\frac{1}{\alpha}$}, 
where $\alpha$ is a positive integer. 
We introduced such {\em fractional index} 
so as to refine the indexes between $1$ and $2$ of quasi-cyclic codes
(\cite{FL15,FL16}). 
A special case is that $\alpha=1$, 
i.e., quasi-cyclic codes of index $2$.  
Thus we study the asymptotic property of the quasi-abelian codes of index $2$ firstly, 
and then turn to the quasi-cyclic codes of index $1\frac{1}{\alpha}$.

\subsection{Quasi-abelian codes of index $2$}\label{ss 2-quasi-abelian}

Let $\F$ be a finite field with $|\F|=q$.  
Let $n>1$ be an integer with $\gcd(n,q)=1$, 
and $G$ be any abelian group of order $n$.
Let $\delta\in(0,1-q^{-1})$.

Denote $(\F G)^2=\F G\times \F G$. For $(a,a')\in(\F G)^2$ 
we have a quasi-abelian code of index $2$: 
\begin{equation}\label{e C_a,a'}
 C_{a,a'}=\{(ba,ba')\,|\, b\in\F G\}\le (\F G)^2.
\end{equation}
Then $(\F G)^2=(\F G)^{1\times 2}$ and $C_{a,a'}$
are the special case of
Eq.\eqref{ensemble} and Eq.\eqref{e C_A} for $k=1$ and $t=2$.
Note that in Section \ref{quasi-abelian} we considered 
the asymptotic properties of $\Delta(C_A)$ and $\R(C_A)$
when $n$ is fixed and $t\to\infty$; 
whereas we are now concerned with the 
asymptotic properties of $\Delta(C_{a,a'})$ and $\R(C_{a,a'})$
when $k=1$ and $t=2$ but $n\to\infty$.

For $b\in\F G$, we denote $I_b=\F Gb$ and $d_b=\dim I_b$, 
and define 
\begin{equation}\label{X_b,X}
X_b=\begin{cases}1, & 0<\w(ba,ba')\le 2n\delta;\\
  0, & \mbox{otherwise;}   \end{cases}
  ~~~~\mbox{and}~~~~ X=\sum_{b\in\F G}X_b.
\end{equation}
Then Eq.\eqref{X ge 1} still holds; and by Lemma \ref{l EX_b}(1) we have
\begin{equation}\label{e E(X_b)}
 {\E}(X_{b}) < q^{- 2d_{b}g_q(\delta)}.
\end{equation}
Let $I\le\F G$.  
By Lemma \ref{c semisimple abelian},
if $I\ne\F Ge_0=\F e_0$ then $\dim I\ge\mu_q(n)$, 
where $\mu_q(n)$ is defined in Definition \ref{q-coset}.
For any $\ell$ with $\mu_q(n)\le\ell\le n$, let
$$
\Gamma_\ell=\big\{\,I\,\big|\,I\le\F G,\,\dim I=\ell\,\big\}, ~~~~~
I^*=\big\{\,a\,|\,a\in I,\, I_a(=\F Gb)=I\,\big\}.
$$
Refining Eq.\eqref{disjoint union}, we have 
\begin{equation}\label{e union}
\textstyle\F G=\F e_0\bigcup\big(\bigcup_{\ell=\mu_q(n)}^n
 \bigcup_{I\in\Gamma_\ell} I^*\big).
\end{equation}
\begin{lemma}\label{l |Gamma_ell|}
$|\Gamma_\ell|\le n^{\ell/\mu_q(n)}$.
\end{lemma}

\begin{proof}
For $I\in\Gamma_\ell$, by Theorem \ref{semisimple abelian}(3), 
$I$ is a direct sum of some of $\F_i$'s for $i=0,1,\cdots,m$.
Except for $\F_0$, 
the summands of $I$ are of dimension $\ge\mu_q(n)$.
There are at most $\ell/\mu_q(n)$ summands. 
Thus $ |\Gamma_\ell|\le (m+1)^{\ell/\mu_q(n)}\le n^{\ell/\mu_q(n)}$.
\end{proof}

\bex\label{g(delta)>)}
Let $\delta\in(0,1-q^{-1})$ and $\gamma\in(0,1)$. 
If $g_q(\delta)>\gamma$ then $\delta<(1-\gamma)(1-q^{-1})$.
\hint{$g_q(\cdot)$ is convex, $g_q(\delta)>\gamma g_q(0)+(1-\gamma)g_q(1-q^{-1})
  \ge g_{q}\big(\gamma\!\cdot\! 0+(1-\gamma)(1-q^{-1})\big)
   =g_{q}\big((1-\gamma)(1-q^{-1})\big)$; 
  and $g_q(\cdot)$ is decreasing.}
\eex

\begin{lemma}\label{EX for C_a,a'}
If $g_q(\delta)-\frac{1}{2}-\frac{\log_q n}{\mu_q(n)}>0$, then
$\E(X)\le
q^{- 2\mu_q(n) \big(g_q(\delta)-\frac{1}{2}-\frac{\log_q n}{\mu_q(n)}\big)}$.
\end{lemma}

\begin{proof}
Since $g_q(\delta)>\frac{1}{2}$,
by Exercise \ref{g(delta)>)} we get that 
$\delta<\frac{1}{2}(1-q^{-1}) <\frac{1}{2}$.
For $b\in \F e_0$, recalling that $e_0=\frac{1}{n}\sum_{x\in G}x$,
we see that $\w(ba,ba')=0$, $n$ or $2n$, hence 
$\E(X_b)=0$, see the definition of $X_b$ in Eq.\eqref{X_b,X}. 
By the linearity of expectations and Eq.\eqref{e union}, we have
\begin{equation}\label{e E(X)}
\textstyle\E(X)
 =\sum_{b\in\F G}\E(X_b)
 =\sum_{\ell=\mu_q(n)}^n \sum_{I\in\Gamma_\ell}
 \sum_{b\in I^*}\E(X_b).
\end{equation}
Let $I\in\Gamma_\ell$; for $b\in I^*$, $d_b=\dim I_b=\dim I=\ell$; 
hence, by Eq.\eqref{e E(X_b)} we have
\begin{align*}
\textstyle
 \sum_{b\in I^*}\E(X_b)
&\textstyle
   \le \sum_{b\in I^*} q^{- 2\ell g_q(\delta)}
=|I^*|\cdot q^{- 2\ell g_q(\delta)}\\
& \le  q^\ell q^{- 2\ell g_q(\delta)}
  =   q^{- 2\ell (g_q(\delta)-\frac{1}{2})}.
\end{align*}
Then
$$\textstyle
\E(X)
 \le \sum_{\ell=\mu_q(n)}^n \sum_{I\in\Gamma_\ell}
  q^{- 2\ell (g_q(\delta)-\frac{1}{2})}
  =\sum_{\ell=\mu_q(n)}^n |\Gamma_\ell|\cdot
  q^{- 2\ell (g_q(\delta)-\frac{1}{2})}.
$$
By Lemma \ref{l |Gamma_ell|}, 
\begin{align*}
\E(X)
\textstyle
   \le\sum_{\ell=\mu_q(n)}^n n^{\frac{\ell}{\mu_q(n)}}
   q^{- 2\ell (g_q(\delta)-\frac{1}{2})}
 =\sum_{\ell=\mu_q(n)}^n 
  q^{- 2\ell \big(g_q(\delta)-\frac{1}{2}-\frac{\log_q n}{2\mu_q(n)}\big)}.
\end{align*}
Since $g_q(\delta)-\frac{1}{2}-\frac{\log_q n}{2\mu_q(n)}>0$
and $\ell\ge\mu_q(n)$, 
\begin{align*}
\E(X)
 &\textstyle \le \sum_{\ell=\mu_q(n)}^n 
  q^{- 2\mu_q(n) \big(g_q(\delta)-\frac{1}{2}-\frac{\log_q n}{2\mu_q(n)}\big)}
\le n\cdot 
 q^{- 2\mu_q(n) \big(g_q(\delta)-\frac{1}{2}-\frac{\log_q n}{2\mu_q(n)}\big)}
 \\
 & =
  q^{- 2\mu_q(n) \big(g_q(\delta)-\frac{1}{2}-\frac{\log_q n}{2\mu_q(n)}\big)
  +\log_q n}
 =q^{- 2\mu_q(n) \big(g_q(\delta)-\frac{1}{2}-\frac{\log_q n}{\mu_q(n)}\big)}.
\end{align*}
We are done.
\end{proof}

\begin{lemma}\label{l Pr R}
The rate $\R(C_{a,a'})\le\frac{1}{2}$; and
$$\textstyle
 (1-q^{-2})q^{-q^{\log_q(2n)-2\mu_q(n)}}\le
 \Pr\big(\R(C_{a,a'})=\frac{1}{2}\big)\le 1-q^{-2}.
$$
\end{lemma}

\begin{proof}
By Lemma \ref{l 2-quasi self-dual} and  Exercise \ref{e R}(3), 
$$\textstyle
 \R(C_{a,a'})\le\frac{1}{2};~~~ \mbox{and} ~~
 \R(C_{a,a'})=\frac{1}{2} \iff (a,a')\in(\F G)^{2*}.
$$
By Eq.\eqref{FG=F_0+...} (and note that $d_0=1$),
$$\textstyle
\Pr\big(\R(C_{a,a'})\!=\!\frac{1}{2}\big)=
\frac{|(\F G)^{2*}|}{|(\F G)^2|}
 =\prod_{i=0}^m\frac{|\F_i^{2*}|}{|\F_i^2|}
 =\frac{q^2-1}{q^2}\prod_{i=1}^m\frac{q^{2d_i}-1}{q^{2d_i}}.
$$
Since $d_i\ge\mu_q(n)$ for $i=1,\cdots,m$ and $m<n$, 
\begin{equation}\label{Pr R}
\textstyle
\Pr\big(\R(C_{a,a'})\!=\!\frac{1}{2}\big)
\ge (1\!-\!q^{-2})\prod_{i=1}^m(1\!-\!q^{-2\mu_q(n)})
\ge (1\!-\!q^{-2})\cdot(1\!-\!q^{-2\mu_q(n)})^n.
\end{equation}
And
$$
\textstyle
(1-q^{-2\mu_q(n)})^n
=(1-q^{-2\mu_q(n)})^{q^{2\mu_q(n)}\frac{n}{q^{2\mu_q(n)}}}.
$$
The sequence $(1-t^{-1})^t$ for $t=2,3,\cdots$ is increasing and
$(1-2^{-1})^2=2^{-2}\ge q^{-2}$. Thus
$$
\textstyle
(1-q^{-2\mu_q(n)})^n
\ge (q^{-2})^{\frac{n}{q^{2\mu_q(n)}}} = q^{-\frac{2n}{q^{2\mu_q(n)}}}
=q^{-q^{\log_q(2n)-2\mu_q(n)}}.
$$
Combining it with Eq.\eqref{Pr R}, we finished the proof of the lemma.
\end{proof}

\bex\label{e q^k-1}
If $k_1\le\cdots\le k_m$ be integers and $k_1\ge\log_q\!m$
 then 

(1)~ $(q^{k_1}-1)\cdots(q^{k_m}-1)\ge q^{k_1+\cdots+k_m-2}$;
\hint{\!
$\frac{(q^{k_1}\!-1)\cdots(q^{k_m}\!-1)}{(q^{k_1})\cdots(q^{k_m})}
=\prod_{i=1}^{m}(1\!-\!\frac{1}{q^{k_i}})$ 
$\ge (1-\frac{1}{q^{k_1}})^m\ge (1-\frac{1}{q^{k_1}})^{q^{k_1}}$; 
 $(1-t^{-1})^t$ is increasing and $(1-2^{-1})^2=2^{-2}\ge q^{-2}$.}

(2)~ $(q^{k_1}+1)\cdots(q^{k_m}+1)\le q^{k_1+\cdots+k_m+2}$.
\hint{$(1+t^{-1})^t<2^{2}\le q^{2}$.}
\eex

By Lemma \ref{l n_1,...} we have positive integers $n_1,n_2,\cdots$
coprime to $q$ such that 
$$
 \lim\limits_{i\to\infty}\textstyle\frac{\log_q n_i}{\mu_q(n_i)}=0.
$$

\begin{theorem}\label{t C_a,a'}
Assume that $\delta\in(0,1-\frac{1}{q})$ and $g_q(\delta)>\frac{1}{2}$,  
and integers $n_1,n_2,\cdots$ are as above.
For $i=1,2,\cdots$,
let $G_i$ be any abelian group of order $n_i$, and 
 $C^{(i)}_{a,a'}$ be the random 2-quasi-$\F G_i$-codes
 as in Eq.\eqref{e C_a,a'}. Then

{\rm(1)}~ 
$\lim\limits_{i\to\infty}\Pr\big(\Delta(C^{(i)}_{a,a'})>\delta\big)=1$.

{\rm(2)}~  
$\lim\limits_{i\to\infty}\Pr\big(\R(C^{(i)}_{a,a'})=\frac{1}{2}\big)=1-q^{-2}$.
\end{theorem}

\begin{proof}
(1). 
Since $g_q(\delta)>\frac{1}{2}$ and 
$\lim\limits_{i\to\infty}\textstyle\frac{\log_q n_i}{\mu_q(n_i)}=0$,
we can take a positive real number $\varepsilon$ such that
$$\textstyle
 g_q(\delta)-\frac{1}{2}-\frac{\log_q n_i}{\mu_q(n_i)}>\varepsilon,
 ~~~~ i=1,2,\cdots.
$$
Note that $\mu_q(n_i)\to\infty$. By Lemma \ref{EX for C_a,a'},
\begin{align*}
\lim\limits_{i\to\infty}\Pr\big(\Delta(C^{(i)}_{a,a'})\le\delta\big)
&=\lim\limits_{i\to\infty}\Pr(X\ge 1)
\le\lim\limits_{i\to\infty}\E(X)\\
&\le\lim\limits_{i\to\infty}
q^{- 2\mu_q(n_i) \big(g_q(\delta)-\frac{1}{2}-\frac{\log_q n_i}{\mu_q(n_i)}\big)}
=0.
\end{align*}

(2). Since $\frac{\mu_q(n_i)}{\log_q n_i}\to\infty$ (with $i\to\infty$), 
$\log_q(2n_i)-2\mu_q(n_i)\to-\infty$; hence
$$
\lim\limits_{i\to\infty}q^{-q^{\log_q(2n_i)-2\mu_q(n_i)}}=1.
$$
And (2) follows from Lemma \ref{l Pr R} at once.
\end{proof}

\begin{corollary}\label{c C_a,a'}
Keep the notation as in Theorem \ref{t C_a,a'}. Then
$$\textstyle
 \lim\limits_{i\to\infty}\Pr\big(\Delta(C_{a,a'}^{(i)})>\delta\;\big|\; 
   \R(C_{a,a'}^{(i)})\!=\!\frac{1}{2}\big)=1.$$
\end{corollary}

\begin{proof}
By Exercise \ref{e total}(2), it follows from Theorem \ref{t C_a,a'} at once.  
\end{proof}

\begin{theorem}\label{2-quasi good}
Assume that $\delta\in(0,1-\frac{1}{q})$ and $g_q(\delta)>\frac{1}{2}$; 
and positive integers $n_1,n_2,\cdots$ satisfy that  
$\lim\limits_{i\to\infty}\frac{\log_q n_i}{\mu_q(n_i)}=0$;
and $G_i$ for $i=1,2,\cdots$ are abelian group of order $n_i$.
Then for $i=1,2,\cdots$ there exist 2-quasi-$\F G_i$ codes $C^{(i)}$
such that
\begin{itemize}
\item\vskip-3pt
 $\R(C^{(i)})=\frac{1}{2}$ for all $i=1,2,\cdots$;
\item\vskip-3pt 
$\Delta(C^{(i)})>\delta$ for all $i=1,2,\cdots$.
\end{itemize}
\end{theorem}

In conclusion, there are asymptotically good quasi-abelian codes of index $2$ 
attaining the GV-bound.

\subsection{Quasi-cyclic codes of fractional index}\label{ss fractional}

In this subsection,
 $\F$ is a finite field with $|\F|=q$, $n>1$ and $\gcd(n,q)=1$,  
 and $\delta\in(0,1-q^{-1})$ as before.
Assume that $\alpha$ is a positive integer; and: 
\begin{itemize}
\item\vskip-3pt
$G=\{1,x,\cdots,x^{n-1}\}$ is a cyclic group of order $n$; 
\item\vskip-3pt
$\tilde G=\{1,\tilde x,\cdots, \tilde x^{\alpha n-1}\}$ 
 is a cyclic group of order $\alpha n$.
\end{itemize}
\vskip-3pt
There is a natural surjective group homomorphism
$\tilde G\to G$, $\tilde x^i\mapsto x^i$. 
By Exercise \ref{e universal}(2), 
we have a natural surjective algebra homomorphism
$$\textstyle
 \F\tilde G\longrightarrow\F G, ~ 
 \sum_{i=0}^{\alpha n-1}a_i\tilde x^i ~\longmapsto~
 \sum_{i=0}^{\alpha n-1}a_ix^i 
 =\sum_{i=0}^{n-1}\big(\sum_{j=0}^{\alpha-1}a_{i+ jn}\big)x^i.
$$
Thus, $\F G$ is an $\F\tilde G$-module, and any $\F G$-module is also
an $\F\tilde G$-module.

\begin{definition}\label{d fractional}\rm
Any $\F\tilde G$-submodule $C$ of the $\F\tilde G$-module
$$\textstyle
 \F\tilde G\times \F G=\big\{(\tilde a,a)\,\big|\,  
 \tilde a=\sum_{i=0}^{\alpha n-1}\tilde a_i\tilde x^i\in\F\tilde G,
  ~ a=\sum_{i=0}^{n-1}a_ix^i\in\F G\big\},
$$
denoted by $C\le \F\tilde G\times \F G$, 
is said to be a {\em quasi-cyclic code of index $1\frac{1}{\alpha}$}.
\end{definition}

In particular, if $\alpha=1$ 
then quasi-cyclic codes of index $1\frac{1}{\alpha}$
are quasi-cyclic codes of index $2$. In Subsection \ref{ss 2-quasi-abelian},
we have seen that there exist asymptotically good
$2$-quasi-abelian code sequences attaining GV-bound.
With that good asymptotic property we can exhibit 
a good asymptotic property of quasi-cyclic codes of fractional index.

We write $a=\sum_{i=0}^{n-1}a_ix^i\in\F G$ as $a(x)\in\F G$.
And the multiplication $a(x)b(x)$ is the same 
as polynomial multiplication but modulo $x^n-1$.
  
\begin{lemma}\label{l varphi}
Let 
$\varphi=1+\tilde x^{n}+\cdots+\tilde x^{n(\alpha-1)}\in\F\tilde G$.
Then
$$\tilde\varphi: ~\F G\longrightarrow\F\tilde G,~~
   a(x)\longmapsto a(\tilde x)\varphi,
$$
is an $\F\tilde G$-module injective homomorphism, and
\begin{equation}\label{alpha weight}
 \w\big(\tilde\varphi(a(x))\big)
 =\alpha\w\big(a(x)\big),~~~~~
   \forall~a(x)\in\F G.
\end{equation}
\end{lemma}

\begin{proof}
For any $a(x)=\sum_{i=0}^{n-1}a_ix^i\in\F G$, 
$$
\tilde\varphi\big(a(x)\big)=a(\tilde x)\varphi=a(\tilde x)+a(\tilde x)\tilde x^{n}
 +\cdots+a(\tilde x)\tilde x^{n(\alpha-1)},
$$ 
which is corresponding to the word
$$
(a_0,a_1,\cdots,a_{n-1},~a_0,a_1,\cdots,a_{n-1},~\cdots, ~
 a_0,a_1,\cdots,a_{n-1})\in\F^{\alpha n}.
$$
So, $\tilde\varphi$ is an injective $\F$-linear map,
and Eq.\eqref{alpha weight} holds.
For the $\tilde x\in\F\tilde G$, in $\F G$ the ``module multiplication'' 
of $a(x)$ by $\tilde x$ is  
$$\textstyle
\tilde x\circ a(x)=x(a_0+a_1x+\cdots+a_{n-1}x^{n-1})
=a_{n-1}+a_0x+\cdots+a_{n-2}x^{n-1}.
$$
Note that in $\F\tilde G$ we have 
$(\tilde x^n-1)\varphi=\tilde x^{n\alpha}-1=0$. 
We obtain that 
\begin{align*}
&\tilde\varphi\big(\tilde x\circ a(x)\big)
=(a_{n-1}+a_0\tilde x+\cdots+a_{n-2}\tilde x^{n-1})\varphi\\
&=\big(\tilde x a(\tilde x)-a_{n-1}(\tilde x^n-1)\big)\varphi
=\tilde x a(\tilde x)\varphi=\tilde x \tilde\varphi\big(a(x)\big).
\end{align*}
Thus $\tilde\varphi$ is an $\F\tilde G$-module homomorphism.
\end{proof}

By Lemma \ref{l varphi}, we have an injective $\F\tilde G$-module
homomorphism
\begin{equation}
\Phi:~ \F G\times\F G\longrightarrow \F\tilde G\times \F G,~~
 \big(a(x),a'(x)\big)\longmapsto \big(\tilde\varphi\big(a(x)\big),\,a'(x)\big).
\end{equation}

\begin{lemma}
Let $C_{a,a'}$ be the 2-quasi-$\F G$ code in Eq.\eqref{e C_a,a'}. 
Denote $\tilde C_{a,a'}=\Phi\big(C_{a,a'}\big)
\le \F\tilde G\times \F G$.
Then $\tilde C_{a,a'}$ is a $1\frac{1}{\alpha}$-quasi-cyclic code and

{\rm (1)} $\R\big(\tilde C_{a,a'}\big)=\frac{2}{\alpha+1}\R(C_{a,a'})$; 

{\rm (2)} $\Delta\big(\tilde C_{a,a'}\big)
 \ge\frac{2}{\alpha+1}\Delta\big(C_{a,a'}\big)$. 
\end{lemma}

\begin{proof}
Since $\Phi$ is an $\F\tilde G$-homomorphism, $\tilde C_{a,a'}$ is an
$\F\tilde G$-submodule of $\F\tilde G\times\F G$.

(1). Because $\Phi$ is injective, 
$\dim \tilde C_{a,a'}=\dim C_{a,a'}$. So
$$\textstyle
\R\big(\tilde C_{a,a'}\big)=\frac{\dim \tilde C_{a,a'}}{\alpha n+n}
=\frac{2}{\alpha+1}\cdot\frac{\dim C_{a,a'}}{2n}
=\frac{2}{\alpha+1}\R(C_{a,a'}).
$$

(2). Denote $\tilde d=\d(\tilde C_{a,a'})$, $d=\d(C_{a,a'})$.
By Eq.\eqref{alpha weight}, $\tilde d\ge d$. Then
$$\textstyle
\Delta\big(\tilde C_{a,a'}\big)=\frac{\tilde d}{\alpha n+n}
\ge \frac{2}{\alpha+1}\cdot\frac{d}{2n}
=\frac{2}{\alpha+1}\Delta(C_{a,a'}).
$$
We are done.
\end{proof}

From Theorem \ref{2-quasi good} we get the following at once.

\begin{theorem}\label{fractional good}
Assume that $\delta\in(0,1-\frac{1}{q})$ and $g_q(\delta)>\frac{1}{2}$; 
and positive integers $n_1,n_2,\cdots$ satisfy that  
$\lim\limits_{i\to\infty}\frac{\log_q n_i}{\mu_q(n_i)}=0$.
Then for $i=1,2,\cdots$ there exist quasi-cyclic codes 
$\tilde C^{(i)}$ of index $1\frac{1}{\alpha}$ over $\F$ such that
\begin{itemize}
\item\vskip-3pt
 the length of $\tilde C^{(i)}$ equals $(\alpha+1)n_i$ for all $i=1,2,\cdots$;
\item\vskip-3pt
 $\R(\tilde C^{(i)})=\frac{1}{\alpha+1}$ for all $i=1,2,\cdots$;
\item\vskip-3pt 
$\Delta(\tilde C^{(i)})>\frac{2\delta}{\alpha+1}$ for all $i=1,2,\cdots$.
\end{itemize}
\end{theorem}

In conclusion, quasi-cyclic codes of index $1\frac{1}{\alpha}$ 
are asymptotically good. But, if $\alpha>1$, we cannot get a
sequence of the codes whose parameters 
attain the GV-bound. 

\begin{remark}\rm
(1) 
It is well-known for a very long time that 
binary quasi-cyclic codes of index $2$ are asymptotically good
(and attain the GV-bound), see \cite{CPW, C, K}. 
In \cite{BM}, binary quasi-abelian codes of fixed index $t\ge 2$ 
are proved to be asymptotically good. 
In the dissertation \cite{L PhD},
the asymptotic goodness of quasi-abelian codes of fixed index $t\ge 2$ 
are extended to any $q$-ary case. 
In this section we present a proof of this asymptotic goodness for index $2$. 

(2) Quasi-cyclic codes of index $1\frac{1}{2}$ and $1\frac{1}{3}$ 
are introduced and proved to be asymptotically good in \cite{FL15,FL16};
and the results are extended to quasi-cyclic codes of general fractional index
in \cite{MC}. In this section, 
based on the result on the asymptotic goodness  
for quasi-abelian codes of index $2$, 
we take a more concise argument to prove 
the asymptotic goodness of quasi-cyclic codes of index~$1\frac{1}{\alpha}$
for any $q$ and any positive integer $\alpha$.
\end{remark}

\section{Self-dual quasi-abelian codes of index 2}\label{self-dual}

In this section we always assume: 
\begin{itemize}
\item \vskip-5pt
 $\F$ is a finite field with $|\F|=q$;
\item \vskip-5pt
 $\delta\in(0,1-\frac{1}{q})$ and $g_q(\delta)>\frac{3}{4}$;
\item \vskip-5pt 
$n>1$ is an odd integer, $\gcd(n,q)=1$ and $\mu_q(n)>2\log_q n$,
 where $\mu_q(n)$ is defined in Definition \ref{q-coset}; 
\item \vskip-5pt
 $G$ is an abelian group of order $n$. 
\end{itemize}
\vskip-5pt
We refine the notation about $\F G$ in Section \ref{group codes} 
for the further study.  

By Theorem \ref{semisimple abelian}, 
$\textstyle E=\{e_0=\frac{1}{n}\sum_{x\in G}x,\,e_1,\cdots,e_m\}$
is the set of all primitive idempotents of $\F G$, and 
$\F G=Fe_0\oplus \F Ge_1\oplus\cdots\oplus \F Ge_m$
as in Eq.\eqref{FG=F_0+...};  
each $\F_i:=\F Ge_i$ is a field extension of $\F$.  
The bar map in Definition \ref{d bar map}:  
$\F G\to\F G$, $a\mapsto\overline a$, is an automorphism of $\F G$ of order $2$; 
hence $\overline E=E$, i.e., bar map induces 
a permutation of order $\le 2$ of the primitive idempotents. 
Then we can write the set $E$ of all primitive idempotents of $\F G$ as 
\begin{equation}\label{e E}
E=\{e_0\}\cup \{e_1,\cdots,e_r\}\cup\{e_{r+ 1},\overline e_{r+ 1},\,\cdots,\, e_{r+ s},\overline e_{r+ s}\},
\end{equation}
where 
$\overline e_i=e_i$  for  $i=1,\cdots,r$,
$\overline e_{r+ j}\ne e_{r+j}$  for $ j=1,\cdots,s$ 
and $1+r+2s=m$. 
For $i=1,\cdots,r$, $\overline{\F Ge_i}=\F Ge_i$.
For $j=1,\cdots,s$, 
the  restriction of the ``bar'' map to $A_{r+j}$ induces a map:
\begin{equation}\label{e e_r+j}
 \F G e_{r+j}~\longrightarrow \F G\overline e_{r+j},~~~ 
   a \longmapsto\overline a,
\end{equation} 
which is an $\F$-algebra isomorphism.
And $\F Ge_{r+ j}+\F G\overline e_{r+ j}$ 
are invariant by the automorphism ``bar''. 
So in the following we set
\begin{equation}\label{d hat E}
\begin{array}{l}
\widehat e_{r+j}=e_{r+j}+\overline e_{r+j},~  j=1,\cdots,s;
\\[4pt]
 \widehat E=\{e_0,\,e_1,\cdots,e_r,\, \widehat e_{r+1},\cdots,\widehat e_{r+s}\}. 
\end{array}
\end{equation}
 
\begin{remark}\rm
For $i=0,1,\cdots,r$, $\F Ge_i=\F_i$ is a finite field. However,  
for $j=1,\cdots,s$, 
$\F G\widehat e_{r+j}=\F Ge_{r+j}+\F G\overline e_{r+j}$
which is isomorphic to a direct sum of two copies of the finite field $\F_{r+j}$, 
see Exercise \ref{F times F}. And
\begin{align}\label{e hat E}
 1=\sum_{e\in\widehat E}e;~~~~~ 
 \overline e=e, ~\forall~e\in\widehat E;~~~~~
 ee'=\begin{cases}e, \!& e=e';\\ 0, \!& e\ne e'; \end{cases}~\forall~e,e'\in\widehat E;
\end{align}
For any $e\in\widehat E$, the ``bar'' map induces an automorphism 
of the algebra $\F Ge$.

\bex\label{F times F}
Let $\F_{r+j}=\F Ge_{r+j}$ which is a field,  
and $R=\F_{r+j}\times\F_{r+j}$. 
For $(a,b)\in R$, set $\overline{(a,b)}=(b,a)$.  Then:

(1) $(a,b)\mapsto \overline{(a,b)}$ is an automorphism of the ring~$R$. 

(2) $\iota: R\to \F G\widehat e_{r+j}$, $(a,b)\mapsto a+\overline b$, 
is a ring isomorphism, and  
$\iota\big(\overline{(a,b)}\big)=\overline{\iota(a,b)}$
for any $(a,b)\in R$. \hint{cf. Eq.\eqref{e e_r+j}.}
\eex
\end{remark}

For any ideal $A\le\F G$, denote $A^\flat=\{a\,|\,a\in A,~\bar a=a\}$.
For $j=1,\cdots,s$, 
the restriction of the bar map to $\F G\widehat e_{r+j}$ is an automorphism
of $\F G\widehat e_{r+j}$ of order $2$ described in
Eq.\eqref{e e_r+j} and Exercise \ref{F times F}; in particular,
\begin{equation}\label{e dim FG e_r+j}
\textstyle
 (\F G\widehat e_{r+j})^\flat=\{a\!+\!\overline a\,|\, a\in\F G e_{r+j}\},~~
 \dim(\F G\widehat e_{r+j})^\flat=\frac{1}{2}\dim\F G\widehat e_{r+j}.
\end{equation}

\begin{lemma}\label{l dim FG e_i}
For $i=1,\cdots,r$, $\F Ge_i\to\F Ge_i$, $a\mapsto \overline a$, is
an automorphism of the field $\F_i=\F Ge_i$ of order $2$; in particular,
$\dim_\F\F_i=\dim\F Ge_i$ is even and 
$\dim_\F\F_i^\flat=\frac{1}{2}\dim_\F\F_i$.
\end{lemma}

\begin{proof} 
For any $1\ne x\in G$, $\overline x=x^{-1}\ne x$
(because $|G|=n$ is odd). 
So, for $a=\sum_{x\in G}a_x x\in\F G$, 
$\overline a=a$ if and only if
$a_x=a_{x^{-1}}$ for all $1\ne x\in G$. Thus
$\dim (\F G)^\flat=1+\frac{n-1}{2}$. 
Note that
$$\textstyle
(\F G)^\flat=(\F Ge_0)^\flat+\sum_{i=1}^r(\F Ge_i)^\flat
+\sum_{j=1}^s(\F G\widehat e_{r+j})^\flat.
$$
It is trivial that $(\F Ge_0)^\flat=\F Ge_0=\F e_0$.
For $j=1,\cdots,s$, we have seen in Eq.\eqref{e dim FG e_r+j} that 
$\dim(\F G\widehat e_{r+j})^\flat=\frac{1}{2}\dim\F G\widehat e_{r+j}$.
For $i=1,\cdots,r$, the bar map induces an automorphism of order $\le 2$
on the finite field $\F_i=\F Ge_i$. Thus
$\dim (\F Ge_i)^\flat\ge \frac{1}{2}\dim (\F Ge_i)$.
Because $\dim (\F G)^\flat=1+\frac{n-1}{2}$, for each $i$ with $1\le i\le r$, 
the equality $\dim (\F Ge_i)^\flat=\frac{1}{2}\dim (\F Ge_i)$ has to be held.
\end{proof}

Let $k_0=\dim\F Ge_0=1$. In the following we denote: 
\begin{equation}\label{e hat E^dag}
\widehat E^\dag=\widehat E-\{e_0\}
=\{e_1,\cdots,e_r, \widehat e_{r+1},\cdots, \widehat e_{r+s}\}.
\end{equation} 
For any $e\in\widehat E^\dag$, 
by Eq.\eqref{e dim FG e_r+j} and Lemma \ref{l dim FG e_i}, 
$\dim\F Ge$ is even; so, instead of $d_i$ in Theorem \ref{semisimple abelian}, 
we denote $k_e=\frac{1}{2}\dim\F Ge$, 
and set $k_i$ as: 
\begin{equation}\label{d k_i k_r+j}
\begin{array}{l}
 k_i:=k_e=\frac{1}{2}\dim \F Ge_i, ~~ \mbox{for}~ e=e_i,~ i=1,\cdots,r;
 \\[5pt] 
 k_{r+j}:=k_e=\dim \F G e_{r+j}, ~~ 
  \mbox{for}~ e=\widehat e_{r+j},~  j=1,\cdots,s;
 \\[5pt]
 \mbox{hence}~~\frac{n-1}{2}=k_1+\cdots+k_r+k_{r+1}+\cdots+k_{r+s}.
\end{array}
\end{equation}
Note that: 
$\min\{k_e\,|\,e\in\widehat E^\dag\}\ge \lceil\frac{1}{2}\mu_q(n)\rceil$, 
hence $k_e>\log_q n$ for all $e\in\widehat E^\dag$ 
(because we have assumed that $\mu_q(n)>2\log_q n$).

\bex
$\min\{k_e\,|\,e\in\widehat E^\dag\}\ge \lceil\frac{1}{2}\mu_q(n)\rceil$.
\eex

\subsection{Self-dual quasi-$\F G$ codes of index $2$}

We consider the following ${\cal D}$ as a probability space with equal
probability for every sample:
\begin{equation}
{\cal D}=
\{(1,b)\,|\, b\in(\F G)^\times, ~ b\overline{b}=-1\}.
\end{equation}
By Lemma \ref{l 2-quasi self-dual}(3), 
we have the self-dual 2-quasi-$\F G$ codes $C_{1,b}$ as follows:
\begin{equation}\label{e C_1,b}
 C_{1,b}=\{(a,ab)\,|\,a\in\F G\},~~~ (1,b)\in{\cal D}.
\end{equation}

\begin{lemma}\label{t |D|} Let ${\cal D}$ be as above. Then
$|{\cal D}|= 
s_0\prod_{i=1}^r (q^{k_i}+1)\prod_{j=1}^s(q^{k_{r+j}}-1)$, ~
where $s_0=\begin{cases}1, &q\equiv 0~({\rm mod}\,2);\\
2, &q\equiv 1~({\rm mod}\,4);\\ 0, &q\equiv 3~({\rm mod}\,4); \end{cases}$ 
~ in particular, 
$
  s_0q^{\frac{n-1}{2}-2}\le |{\cal D}| \le q^{\frac{n-1}{2}+3}.
$
\end{lemma}

\begin{proof}
Any $a\in\F G$ is uniquely written as
\begin{equation}\label{a=a_0+...}
\begin{array}{l}
a=a_0+a_1+\cdots+a_r+a_{r+1}+\cdots+a_{r+s}, ~ 
\\[3pt]
\mbox{where}~~ a_i=ae_i\in\F Ge_i=\F_i, ~~ \mbox{for} ~i=0,1,\cdots,r,
\\[3pt]
 a_{r+j}=a\widehat e_{r+j}\in\F Ge_{r+j}+\F G\overline e_{r+j}
 \cong \F_{r+j}\times\F_{r+j},~~ 
  \mbox{for}~ j=1,\cdots,s.
\end{array}
\end{equation}  
For the last line please see Exercise \ref{F times F}.
And, for $a,b\in\F G$, 
\begin{equation}\label{ab=...}
\textstyle
ab=a_0b_0+\sum_{i=1}^r a_ib_i+\sum_{j=1}^s a_{r+j}b_{r+j}.
\end{equation}
Thus, for $b\in(\F G)^\times$,
 $b\overline{b}=-1$ if and only if
\begin{align}\label{u_i,u_i'}
b_i\overline b_i=-e_i, ~~i=0,1,\cdots,r;
~~~~~~
b_{r+j}\overline b_{r+j}=-\widehat e_{r+j}, ~~j=1,\cdots,s.
\end{align}
By Exercise~\ref{e X bar X=-1} below, we have the following three conclusions.
\begin{description}
\item[~~\it Case 1:]\vskip-5pt 
$i=0$;~ then there are $s_0$ choices of $b_0\in (\F Ge_0)^\times$  
such that $b_0\overline b_0=-e_0$.
\item[~~\it Case 2:]\vskip-5pt
 $1\le i\le r$;~ then 
there are $q^{k_i}+1$ choices of  $b_i\in (\F Ge_i)^\times$ 
such that $b_i\overline b_i=-e_i$. 
\item[~~\it Case 3:]\vskip-5pt
 $1\le j\le s$;~ then there are $q^{k_{r+j}}\!-\!1$ choices 
of $b_{r+j}\in (\F G\widehat e_{r+j})^\times$  
such that $b_{r+j}\overline b_{r+j}=-\widehat e_{r+j}$.
\end{description}
\vskip-3pt
We get
$$\textstyle
 |{\cal D}|= 
s_0\prod_{i=1}^r (q^{k_i}+1)\prod_{j=1}^s(q^{k_{r+j}}-1).
$$
By Exercise \ref{e q^k-1}(2) we have (recall that $s_0\le 2\le q$):
$$\textstyle
 |{\cal D}|\le s_0 \prod_{e\in\widehat E^\dag}(q^{k_e}+1)
 \le q q^{2+\sum_{e\in\widehat E^\dag}k_e}
 =q^{\frac{n-1}{2}+3}.
$$
And by Exercise \ref{e q^k-1}(1), 
$$\textstyle
 |{\cal D}|\ge s_0 \prod_{e\in\widehat E^\dag}(q^{k_e}-1)
 \ge s_0 q^{-2+\sum_{e\in\widehat E^\dag}k_e}
 =s_0q^{\frac{n-1}{2}-2}.
$$
The lemma is proved.
\end{proof}

\bex\label{e X bar X=-1}
(1) The number of the solutions in~$\F$ of the equation $X^2=-1$
is equal to $s_0$, where $s_0$ is the same as in Lemma \ref{t |D|}.
\hint{$\F^\times$ is a cyclic group of order $q-1$; if $-1\ne 1$
then any solution of $X^2=-1$ in $\F^\times$ is an element of order $4$.}

(2) Let $\F_i$ be the finite field with $|\F_i|=q^{2k_i}$ as in Eq.\eqref{d k_i k_r+j},
$\overline u=u^{q^{k_i}}$ for $u\in \F_i$. 
Then the number of the solutions in $\F_i$ of the equation 
$X\overline X=-1$ is equal to $q^{k_i}+1$. 
\hint{$X\overline X=X^{q^{k_i}+1}$; 
$\F_i^\times$ is a cyclic group of order $q^{2k_i}-1=(q^{k_i}+1)(q^{k_i}-1)$;
if $-1=1$ then we are done; otherwise $2|(q^{k_i}-1)$ and
there is an element in $\F_i^\times$ of order $2(q^{k_i}+1)$.}

(3) Let $\F_{r+j}$ be the finite field with $|\F_{r+j}|=q^{k_{r+j}}$  
and $R=\F_{r+j}\times\F_{r+j}$ as in Exercise \ref{F times F}.
Then the number of the solutions in $R$ of the equation 
$X\overline X=-1$ is equal to $q^{k_{r+j}}-1$.
\hint{$(a,b)$ is a solution of $X\overline X=-1$ in $R$ iff 
$ab=ba=-1$ in $\F_{r+j}$; choose $a\in\F_{r+j}^\times$ freely, 
and then $b=-a^{-1}$ is uniquely determined.}
\eex

\begin{corollary}\label{c self-dual exist}
Self-dual 2-quasi-$\F G$ codes exist if and only if ${\cal D}\ne\emptyset$, 
if and only if $q\;{\not\equiv}\,3\;({\rm mod}\,4)$.
\end{corollary}

\begin{proof}
For cyclic $G$, this result was proved in \cite{LS03}.  
We show a proof appeared in \cite{LF} for general abelian $G$. 
By Lemma \ref{t |D|}, ${\cal D}\ne\emptyset$ if and only if
$q\;{\not\equiv}\,3\;({\rm mod}\,4)$. If ${\cal D}\ne\emptyset$ then
$C_{1,b}$ in Eq.\eqref{e C_1,b} is a self-dual 2-quasi-$\F G$ code.
Finally, assume that $C$ is a self-dual 2-quasi-$\F G$ code;
then, by Exercise \ref{ex C^bot}(3) below, 
$e_0C$ is self-dual in $e_0(\F G)^2=\F e_0\times \F e_0$, hence
$X^2=-1$ has solutions in $\F$;  
by Exercise \ref{e X bar X=-1}(1), 
$s_0\ge 1$, i.e., ${\cal D}\ne\emptyset$.
\end{proof}

\bex\label{ex C^bot}
Denote $1^{\!\dag}=1-e_0$. Let $C\le(\F G)^2$. 

(1) $(\F G)^2=e_0(\F G)^2\oplus 1^{\!\dag}(\F G)^2 $ 
is an orthogonal direct sum. \hint{similar to the proof of Exercise \ref{ex inner}:
$e_01^{\!\dag}=0$, 
hence $\langle e_0(\F G)^2,\,1^{\!\dag} (\F G)^2 \rangle=0$.}

(2) $C=e_0C\oplus 1^{\!\dag} C$, and 
$C^\bot=(e_0C)^{\bot_{e_0}}\oplus (1^{\!\dag} C)^{\bot_{1^{\!\dag}}}$, where
$(e_0C)^{\bot_{e_0}}$ denotes the orthogonal subspace of $e_0C$
in $e_0(\F G)^2$; and it is similar for $(1^{\!\dag} C)^{\bot_{1^{\!\dag}}}$.

(3) $C=C^\bot$ is self-dual if and only if $e_0C=(e_0C)^{\bot_{e_0}}$
and $1^{\!\dag} C=(1^{\!\dag} C)^{\bot_{1^{\!\dag}}}$.
\eex

\begin{remark}\label{r self-dual exist}\rm
Thus, from now on till to the end of this subsection, 
we assume that $q\;{\not\equiv}\,3\;({\rm mod}\,4)$, so that $s_0\ge 1$, 
${\cal D}\ne\emptyset$ and we can consider the asymptotic properties
of the self-dual 2-quasi-$\F G$ codes $C_{1,b}$ in Eq.\eqref{e C_1,b}.  
For $(a,d)\in(\F G)^2$, we define a $0$-$1$ 
random variable $X_{(a,d)}$ over ${\cal D}$ as follows: 
\begin{equation}\label{X_a}
 X_{(a,d)}=\begin{cases}1, & (a,d)\in C_{1,b} \mbox{~and~} 1\le\w(a,d)\le 2n\delta;\\
 0, & \mbox{otherwise}; \end{cases}
\end{equation}
By the definition, we see that 
if $(a,d)=0$ or $\w(a,d)>2\delta n$ then $X_{(a,d)}=0$ (zero variable);
otherwise, the expectation 
\begin{equation}\label{D_a,d}
\E(X_{(a,d)})=|{\cal D}_{(a,d)}|\big/|{\cal D}|, ~~~\mbox{where}~~
{\cal D}_{(a,d)}=\{(1,b)\in{\cal D}\,|\,(a,d)\in C_{1,b}\}.
\end{equation}
We further define an integer variable $X$ over ${\cal D}$ as follows:
\begin{equation}\label{sum X_a}
\textstyle
 X=\sum_{0\ne(a,d)\in(\F G\times \F G)^{\le\delta}}X_{(a,d)}.
\end{equation}
Then 
\begin{equation}\label{2 E(X)}
\Pr\big(\Delta(C_{1,b})\le\delta\big)=\Pr(X\ge 1)\le\E(X), 
\end{equation}
where the second inequality follows from Markov Inequality (Lemma \ref{Markov}).
\end{remark}

For $a\in\F G$, denote
$$
 E_a=\{e\in E\,|\, ea\ne 0\}, ~~~~~~
\widehat E^\dag_a=\{e\in\widehat E^\dag \,|\, ea\ne 0\}.
$$

\begin{lemma}\label{l D_a,d}
Let $(a,d)\in(\F G)^2$ and ${\cal D}_{(a,d)}$ be as in Eq.\eqref{D_a,d}. 
If ${\cal D}_{(a,d)}\ne \emptyset$ then
$E_a=E_d$ and $|{\cal D}_{(a,d)}|\le q^{\frac{n-1}{2}-\ell_a+3}$.
\end{lemma}

\begin{proof}
Given $(1,g)\in{\cal D}_{(a,d)}$; i.e. $(a,d)\in C_{1,g}$, hence $(a,d)=(a,ag)$;
so $d=ag$. Since $g\in(\F G)^\times$, for $e\in E$, $ed\ne 0$ if and only if $ea\ne 0$; 
that is, $E_a=E_d$. Then 
$(1,b)\in{\cal D}_{(a,d)}$ if and only if $(a,ab)=(a,d)=(a,ag)$, i.e., $ab=ag$.  
By the notation in Eq.\eqref{a=a_0+...} and Eq.\eqref{ab=...},
\begin{equation*}
(1,b)\in{\cal D}_{(a,d)} \iff 
\begin{cases}a_ib_i=a_ig_i, &i=0,1,\cdots,r;\\
a_{r+j}b_{r+j}=a_{r+j}g_{r+j}, &j=1,\cdots,s.\end{cases}
\end{equation*}

{\it Case 1}: $i=0$. If $a_i=0$ then, by Exercise~\ref{e X bar X=-1}(1), there are 
$s_0$ choices of $b_0$ such that $a_0b_0=a_0g_0$.
Otherwise, $a_0\ne 0$ and there is a unique $b_0$ such that $a_0b_0=a_0g_0$.
So, there are $t_0$ choices of $b_0$ such that $a_0b_0=a_0g_0$,
where $t_0=\begin{cases}s_0, & a_0=0;\\ 1, & a_0\ne 0. \end{cases}$ 

{\it Case 2}: $1\le i\le r$. Then $a_i\in\F_i=\F Ge_i$ and
 $\F_i$ is a field with $|\F_i|=q^{2k_i}$.

{\it Subcase 2.1}: $e_i\in \widehat E^\dag_a$. Then $a_i\in\F_{i}^\times$,
there is a unique $b_i\in\F_i$ such that $a_ib_i=a_ig_i$.

{\it Subcase 2.2}: $e_i\notin \widehat E^\dag_a$. Then $a_i=0$ and hence $a_ib_i=a_ig_i=0$;
there are $q^{k_i}+1$ choices of  $b_i\in\F_i$ such that $a_ib_i=a_ig_i$,
cf. Exercise~\ref{e X bar X=-1}(2). 

{\it Case 3}: $1\le j\le s$. Then $a_{r+j}\in 
 \F G\widehat e_{r+j}=\F G e_{r+j}+\F G\overline e_{r+j}$, 
  where $\F G e_{r+j}=\F_{r+j}$ and $|\F_{r+j}|=q^{k_{r+j}}$.
We write $a_{r+j}=a'+a''$ with $a'\in\F Ge_{r+j}$ and $a''\in\F G\overline e_{r+j}$.
And $b_{r+j}=b'+b''$, $g_{r+j}=g'+g''$ in the same way.
Then $(a'+a'')(b'+b'')=a_{r+j}b_{r+j}=a_{r+j}g_{r+j}=(a'+a'')(g'+g'')$, i.e.,
\begin{align}\label{a'b'=}
a'b'=a'g' ~~~  \mbox{and }~~~ a''b''=a''g''.
\end{align}

{\it Subcase 3.1}: $\widehat e_{r+j}\in \widehat E^\dag_a$. 
 Then $(a'+a'')(e_{r+j}+\overline e_{r+j})\ne 0$, hence
 one of $a'$ and $a''$ is nonzero, say $a'\ne 0$. 
 There is a unique $b'$ such that $a'b'=a'g'$.
 Note that $b'\overline b''=-1$ (see the hint for Exercise~\ref{e X bar X=-1}(3)), 
i.e., $b'$ and $b''$ determined each other (it is the same for $g'$ and $g''$).  
There is a unique $b_{r+j}=b'+b''$ such that Eq.\eqref{a'b'=} holds.

{\it Subcase 3.2}: $\widehat e_{r+j}\notin \widehat E^\dag_a$. 
Then $a'=a''=0$, hence Eq.\eqref{a'b'=} always holds.
By Exercise~\ref{e X bar X=-1}(3), 
there are $q^{k_{r+j}}-1$ choices of $b_{r+j}=b'+b''$ such that Eq.\eqref{a'b'=} holds.

Summarizing the above three cases, we get
$$\textstyle
|{\cal D}_{(a,d)}|= t_0
 \prod_{1\le i\le r,\, e_ia\ne 0}(q^{k_i}+1)
 \prod_{1\le j\le s, \,\widehat e_{r+j}a\ne 0}(q^{k_{r+j}}-1).
$$
Since $t_0\le 2\le q$ and $q^{k_{r+j}}-1<q^{k_{r+j}}+1$, 
$$
|{\cal D}_{(a,d)}|\le q
 \prod_{e\in\widehat E^\dag-\widehat E^\dag_a}(q^{k_e}+1)
 \le q q^{2+\sum_{e\in\widehat E^\dag-\widehat E^\dag_a}k_e}
 =q^{3+\frac{n-1}{2}-\ell_a}.
$$
We are done.
\end{proof}

For any $a\in\F G$, we set
\begin{equation}\label{e L_a}
\textstyle
 L_a=\bigoplus_{e\in\widehat E^\dag_a}\F Ge,~~~
 \ell_a=\frac{1}{2}\dim L_a=\sum_{e\in\widehat E^\dag_a}k_e; ~~~
 \tilde L_a=\F e_0\oplus L_a.
\end{equation}
Note that $L_a\ne I_a$ and $\ell_a\ne d_a$ in general
(where $I_a=\F Ga$ and $d_a=\dim I_a$ 
are defined in Definition \ref{d X_b abelian}(3) for $k=1$).
But $I_a\subseteq \tilde L_a$ and $\dim\tilde L_a=2\ell_a+1$.

Further, for an integer $\ell$ with $\lceil\frac{1}{2}\mu_q(n)\rceil\le\ell\le\frac{n-1}{2}$, 
we set:
\begin{equation}\label{e Omega}
\Omega_{2\ell}=\{L\,|\, 
\mbox{$L$ is a direct sum of some of $\F Ge$'s 
 for $e\in\widehat E^\dag$, $\dim L=2\ell$}\};
\end{equation}
and for $L\in\Omega_{2\ell}$ we set 
(where $L_a$ is defined in Eq.\eqref{e L_a}): 
\begin{equation}\label{e tilde L}
\tilde L=\F e_0\oplus L, ~~ \tilde L^*=\{a\in\tilde L\,|\,L_a=L~
   \mbox{(equivalently, $\ell_a=\ell$)}\}.   
\end{equation}
Similarly to Eq.\eqref{e union} and Lemma \ref{l |Gamma_ell|}, 
we have
\begin{equation}\label{|Omega|<...}
\textstyle
\F G=\F e_0\bigcup\Big(
 \bigcup_{\ell=\lceil\frac{1}{2}\mu_q(n)\rceil}^{\frac{n-1}{2}} 
 \bigcup_{L\in\Omega_{2\ell}}\tilde L^*\Big);
~~~~
|\Omega_{2\ell}|\le  n^{2\ell/\mu_q(n)}.
\end{equation}

\begin{lemma}\label{EX for C_u,u'}
If $g_q(\delta)-\frac{3}{4}-\frac{\log_q n}{\mu_q(n)}>0$, then
$\E(X)\le q^{-2\mu_q(n)
\big( g_q(\delta)-\frac{3}{4}-\frac{\log_q n}{\mu_q(n)}\big) +6}$.
\end{lemma}

\begin{proof}
For $(a,d)\in(\F G\times\F G)^{\le\delta}$, if $d\notin \tilde L_a$, 
then $\widehat E^\dag_d\ne \widehat E^\dag_a$, 
hence $E_d\ne E_a$ and, by Lemma \ref{l D_a,d},  $\E(X_{(a,d)})=0$. 
Next, if $a\in\F e_0$ then, similarly to Eq.\eqref{e E(X)}, 
we can get that $\E(X_{(a,d)})=0$. 
Thus, by Eq.\eqref{|Omega|<...}, we can rewrite Eq.\eqref{sum X_a} as: 
\begin{equation*}
\textstyle X
 =\sum_{\ell=\lceil\frac{1}{2}\mu_q(n)\rceil}^{\frac{n-1}{2}}
   \sum_{L\in\Omega_{2\ell}}
 \sum_{(a,d)\in(\tilde L^*\times\tilde L^*)^{\le\delta}}X_{(a,d)}.
\end{equation*}
Let $L\in\Omega_{2\ell}$. Then $\dim(\tilde L\times\tilde L)=2(2\ell+1)=4\ell+2$
and, by Lemma \ref{g-code balanced} and Lemma \ref{product balanced}, 
$\big|(\tilde L^*\times\tilde L^*)^{\le\delta}\big|\le
 \big|(\tilde L\times\tilde L)^{\le\delta}\big|\le q^{(4\ell+2)h_q(\delta)}$. 
For $(a,d)\in\tilde  L^*\times\tilde L^*$, by Eq.\eqref{e tilde L}, $\ell_a=\ell$.
By Eq.\eqref{D_a,d}, Lemma \ref{l D_a,d},
Lemma \ref{t |D|} and Remark \ref{r self-dual exist} (which shows that $s_0\ge 1$),
we have
\begin{align*}
&\textstyle\sum_{(a,d)\in(\tilde L^*\times\tilde L^*)^{\le\delta}}\E(X_{(a,d)})
= \sum_{(a,d)\in(\tilde L^*\times\tilde L^*)^{\le\delta}}
 |{\cal D}_{(a,d)}|/|{\cal D}|\\
&\textstyle\le\sum_{(a,d)\in(\tilde L^*\times\tilde L^*)^{\le\delta}}
 q^{\frac{n-1}{2}-\ell_a+3}\big/q^{\frac{n-1}{2}-2}
=\big|(\tilde L^*\times\tilde L^*)^{\le\delta}\big|\cdot q^{5-\ell}\\
&\le q^{(4\ell+2)h_q(\delta)}\cdot q^{-\ell+5}
 =q^{4\ell h_q(\delta)-\ell+5+2h_q(\delta)}
 \le q^{-4\ell\big(g_q(\delta)-\frac{3}{4}\big)+6},
 \end{align*}
where $2h_q(\delta)< 1/2$ because $g_q(\delta)>3/4$. Thus
\begin{align*}
\E(X)&=\textstyle\sum_{\ell=\lceil\frac{1}{2}\mu_q(n)\rceil}^{\frac{n-1}{2}}
   \sum_{L\in\Omega_{2\ell}}
 \sum_{(a,d)\in(\tilde L^*\times\tilde L^*)^{\le\delta}}\E(X_{(a,d)})\\
&\textstyle\le\sum_{\ell=\lceil\frac{1}{2}\mu_q(n)\rceil}^{\frac{n-1}{2}}
   \sum_{L\in\Omega_{2\ell}}
   q^{-4\ell\big(g_q(\delta)-\frac{3}{4}\big)+6}\\
&\textstyle =\sum_{\ell=\lceil\frac{1}{2}\mu_q(n)\rceil}^{\frac{n-1}{2}}
 |\Omega_{2\ell}|\cdot q^{-4\ell\big(g_q(\delta)-\frac{3}{4}\big)+6}.
\end{align*}
By Eq.\eqref{|Omega|<...},
\begin{align*}
\E(X)&\textstyle\le\sum_{\ell=\lceil\frac{1}{2}\mu_q(n)\rceil}^{\frac{n-1}{2}}
 n^{2\ell/\mu_q(n)}q^{-4\ell\big(g_q(\delta)-\frac{3}{4}\big)+6}\\
 &\textstyle=\sum_{\ell=\lceil\frac{1}{2}\mu_q(n)\rceil}^{\frac{n-1}{2}}
    q^{-4\ell\big(g_q(\delta)-\frac{3}{4}-\frac{\log_q n}{2\mu_q(n)}\big)+6}.
\end{align*}
Since $g_q(\delta)-\frac{3}{4}-\frac{\log_q n}{2\mu_q(n)}>0$ and
$4\ell\ge 2\mu_q(n)$,
\begin{align*}
\E(X)&\textstyle\le
\sum_{\ell=\lceil\frac{1}{2}\mu_q(n)\rceil}^{\frac{n-1}{2}}
    q^{-2\mu_q(n)\big(g_q(\delta)-\frac{3}{4}-\frac{\log_q n}{2\mu_q(n)}\big)+6}\\
 &\textstyle\le n
 q^{-2\mu_q(n)\big(g_q(\delta)-\frac{3}{4}-\frac{\log_q n}{2\mu_q(n)}\big)+6}
 =q^{-2\mu_q(n)\big(g_q(\delta)-\frac{3}{4}-\frac{\log_q n}{\mu_q(n)}\big)+6}.
\end{align*}
We are done.
\end{proof}

By Lemma \ref{l n_1,...} and Remark \ref{r n_1, ...}, we have
odd positive integers $n_1,n_2,\cdots$ coprime to $q$ such that  
\begin{equation}\label{odd n_1, ...}
 \lim\limits_{i\to\infty}\textstyle\frac{\log_q n_i}{\mu_q(n_i)}=0.
\end{equation}

\begin{theorem}\label{t C_u,u'}
Assume that $q\;{\not\equiv}\,3\;({\rm mod}\,4)$,
$\delta\in(0,1-\frac{1}{q})$ and $g_q(\delta)>\frac{3}{4}$,  
and integers $n_1,n_2,\cdots$ are as above.
For $i=1,2,\cdots$,
let $G_i$ be any abelian group of order $n_i$, and 
 $C^{(i)}_{1,b}$ be the random self-dual 2-quasi-$\F G_i$ codes
 defined in Eq.\eqref{e C_1,b}. Then
$\lim\limits_{i\to\infty}\Pr\big(\Delta(C^{(i)}_{1,b})>\delta\big)=1$.
\end{theorem}

\begin{proof}
Take a positive real number $\varepsilon$ such that
$$\textstyle
 g_q(\delta)-\frac{3}{4}-\frac{\log_q n_i}{\mu_q(n_i)}>\varepsilon,
 ~~~~ i=1,2,\cdots.
$$
Note that $\mu_q(n_i)\to\infty$. By Eq.\eqref{2 E(X)} and Lemma \ref{EX for C_u,u'},
\begin{align*}
\lim\limits_{i\to\infty}\Pr\big(\Delta(C^{(i)}_{1,b})\le\delta\big)
&=\lim\limits_{i\to\infty}\Pr(X\ge 1)
\le\lim\limits_{i\to\infty}\E(X)\\
&\le\lim\limits_{i\to\infty}
q^{- 2\mu_q(n_i) \big(g_q(\delta)-\frac{3}{4}-\frac{\log_q n_i}{\mu_q(n_i)}\big)+6}=0.
\end{align*}
We are done.
\end{proof}

\begin{theorem}\label{self-dual 2-quasi good}
Assume that $q\,{\not\equiv}\,3\;({\rm mod}\,4)$,
$\delta\in(0,1-\frac{1}{q})$ and $g_q(\delta)>\frac{3}{4}$, 
and positive odd integers $n_1,n_2,\cdots$ satisfy that  
$\lim\limits_{i\to\infty}\frac{\log_q n_i}{\mu_q(n_i)}=0$.
Let $G_i$ for $i\!=\!1,2,\cdots$ be abelian groups of order~$n_i$.
Then for $i=1,2,\cdots$ there exist self-dual 2-quasi-$\F G_i$ codes $C^{(i)}$
such that
$\Delta(C^{(i)})>\delta$ for all $i=1,2,\cdots$.
\end{theorem}

Of course, we can take all the $G_i$ to be cyclic, so that:
{\em if $q\,{\not\equiv}\,3\,({\rm mod}\,4)$ 
then the self-dual 2-quasi-cyclic codes are asymptotically good}.

\subsection{Self-orthogonal quasi-$\F G$ codes of index $2$}

In this subsection we remove the assumption 
``$q\;{\not\equiv}\,3\;({\rm mod}\,4)$'' in Remark \ref{r self-dual exist}, 
and consider a  kind of self-orthogonal $2$-quasi-$\F G$ codes
of dimension $n-1$, which always exist and will be proved asymptotically good.

\begin{remark}\label{FG^dag}\rm
By Corollary \ref{c self-dual exist}, 
the existence of self-dual {2-quasi-$\F G$} codes is conditional. 
The obstacle for this existence is the $e_0$-components.
So, following $\widehat E^\dag$ defined in Eq.\eqref{e hat E^dag}, 
in this subsection we take the following notation.
\begin{itemize}
\item\vskip-5pt
$1^{\!\dag}=1-e_0$ as in Exercise \ref{ex C^bot},
i.e., $1^{\!\dag}=\sum_{e\in\widehat E^{\!\dag}}e$.
\item\vskip-5pt
$\F G^{\dag}=1^{\!\dag}(\F G)
 =\bigoplus_{e\in\widehat E^{\dag}}\F Ge$ 
which is a ring with identity $1^{\!\dag}$;\\
hence the multiplicative group 
$(\F G^{\dag})^\times=\prod_{e\in\widehat E^{\dag}}(\F Ge)^\times$.
\item\vskip-5pt
${\cal D}^{\dag}=
\{(1^{\!\dag},b^{\dag})\,|\,
 b^{\dag}\in(\F G^{\dag})^\times,\;
 b^{\!\dag}\overline{b^{\!\dag}}=-1^\dag\}$.
\item\vskip-5pt
$C_{1^{\!\dag}\!,b^{\dag}}=
 \{(a^{\dag}, a^{\dag}b^{\dag})\,|\,a^{\dag}\in\F G^\dag\}
 \le (\F G)^2$. By Lemma \ref{l 2-quasi self-dual},  
 $C_{1^{\!\dag}\!,b^{\dag}}$ is a self-orthogonal 2-quasi-$\F G$ code
 of dimension $n\!-\!1$ (as ${\rm Ann}(1^\dag\!,b^\dag)\!=\!\F e_0$).
\item\vskip-5pt
$X_{(a^{\!\dag}, d^{\dag})} =\begin{cases}1, & 
   (a^{\!\dag}, d^{\dag})\in C_{1^{\!\dag}\!,b^{\dag}}
    ~ \mbox{and}~1\!\le\!\w(a^{\!\dag}\!,d^{\dag})\!\le\! 2n\delta;\\
   0, & \mbox{otherwise};  \end{cases}$ ~
     $(a^{\!\dag}, d^{\dag})\in(\F G^{\dag})^2$.
\item\vskip-5pt
Set ${\cal D}^{\dag}_{(a^{\!\dag}, d^{\dag})}
=\{(1^{\!\dag},b^{\dag})\in{\cal D}^{\dag}\,|\,
 (1^{\!\dag},b^{\dag})\in C_{1^{\!\dag},b^{\dag}}\}$; then
\begin{equation*}
 \E(X_{(a^{\!\dag}, d^{\dag})})=
 |{\cal D}^{\dag}_{(a^{\!\dag}, d^{\dag})}|\big/|{\cal D}^{\dag}|.
\end{equation*}
\item\vskip-5pt
$X^\dag=\sum_{0\ne (a^{\!\dag}, d^{\dag})\in(\F G^\dag\times\F G^{\dag})^{\le\delta}}
 X_{(a^{\!\dag}, d^{\dag})}$;~ then
\begin{equation*}
 \Pr\big(\Delta(C_{1^{\!\dag},b^\dag})\le\delta\big)
=\Pr(X^\dag\!\ge 1)\le\E(X^\dag).
\end{equation*}
\end{itemize}
\end{remark}

Similarly to Lemma \ref{t |D|}
(just note that Case 1 in its proof is no longer present here, 
hence $s_0$ disappeared), we have

\begin{lemma}\label{t |D^dag|}
$|{\cal D^\dag}|=
 \prod_{i=1}^r (q^{k_i}+1)\prod_{j=1}^s(q^{k_{r+j}}-1)$; 
hence $|{\cal D}^\dag|\ge q^{\frac{n-1}{2}-2}$.
\end{lemma}

Thus the self-orthogonal 2-quasi-$\F G$ codes
 of dimension $n\!-\!1$ always exist. 
 
Also, similarly to Lemma \ref{l D_a,d}
(Case 1 in its proof is no longer necessary here, 
hence a constant is modified), we have

\begin{lemma}\label{l D_a,d^dag}
Let $(a^{\!\dag}, d^{\dag})\in(\F G^{\dag})^2$. 
If ${\cal D}^{\dag}_{(a^{\!\dag}, d^{\dag})}\ne \emptyset$ then
$\widehat E^{\dag}_{a^{\!\dag}}=\widehat E^{\dag}_{d^{\dag}}$ 
 and $|{\cal D}^{\dag}_{(a^{\!\dag}, d^{\dag})}|
  \le q^{\frac{n-1}{2}-\ell_{a^{\!\dag}}+2}$.
\end{lemma}

Similar to Lemma \ref{EX for C_u,u'}, we have

\begin{lemma}\label{EX for C_1,b^dag}
If $g_q(\delta)-\frac{3}{4}-\frac{\log_q n}{\mu_q(n)}>0$, then
$\E(X^{\dag})\le q^{-2\mu_q(n)
\big( g_q(\delta)-\frac{3}{4}-\frac{\log_q n}{\mu_q(n)}\big) +4}$.
\end{lemma}

\begin{proof} It is similarly to the proof of Lemma \ref{EX for C_u,u'} 
but with a little revision. 
Recall that $\Omega_{2\ell}$ is defined in Eq.\eqref{e Omega} and 
 $|\Omega_{2\ell}|\le  n^{2\ell/\mu_q(n)}$. Then
\begin{equation*}
\textstyle
\F G^\dag=
 \bigcup_{\ell=\lceil\frac{1}{2}\mu_q(n)\rceil}^{\frac{n-1}{2}} 
 \bigcup_{L\in\Omega_{2\ell}}L^*,~~~\mbox{where}~~
 L^*=\{a^\dag\in L\,|\,L_{a^{\!\dag}}=L\}.
\end{equation*}
In other words, instead of $\tilde L^*$, here we need $L^*$.
So
\begin{equation*}
\textstyle X^\dag
 =\sum_{\ell=\lceil\frac{1}{2}\mu_q(n)\rceil}^{\frac{n-1}{2}}
   \sum_{L\in\Omega_{2\ell}}
 \sum_{(a^{\!\dag},d^{\dag})\in( L^{\!*}\times L^{\!*})^{\le\delta}} 
    X_{(a^{\!\dag},d^{\dag})}. 
\end{equation*}
For $L\in\Omega_{2\ell}$, $\dim L=2\ell$, 
hence $\big|(L^{\!*}\times L^{\!*})^{\le\delta}\big|\le
 \big|(L\times L)^{\le\delta}\big|\le q^{4\ell h_q(\delta)}$; 
and
\begin{align*}
&\textstyle\sum_{(a^{\!\dag},d^{\dag})\in (L^{\!*}\times L^{\!*})^{\le\delta}}
\E(X_{(a^{\!\dag},d^{\dag})})
 \textstyle\le 
 \big|(L^*\times L^*)^{\le\delta}\big|\cdot 
 q^{\frac{n-1}{2}-\ell+2}\big/q^{\frac{n-1}{2}-2}\\
&\le q^{4\ell h_q(\delta)}\cdot q^{-\ell+4}
 =q^{4\ell h_q(\delta)-\ell+4}
 \le q^{-4\ell\big(g_q(\delta)-\frac{3}{4}\big)+4}.
 \end{align*}
And
\begin{align*}
\E(X^{\dag})&=\textstyle\sum_{\ell=\lceil\frac{1}{2}\mu_q(n)\rceil}^{\frac{n-1}{2}}
   \sum_{L\in\Omega_{2\ell}}
   \sum_{(a^{\!\dag},d^{\dag})\in(L^{\!*}\times L^{\!*})^{\le\delta}}
    \E(X_{(a^{\!\dag},d^{\dag})})\\
&\textstyle\le\sum_{\ell=\lceil\frac{1}{2}\mu_q(n)\rceil}^{\frac{n-1}{2}}
   \sum_{L\in\Omega_{2\ell}}
   q^{-4\ell\big(g_q(\delta)-\frac{3}{4}\big)+4}.
\end{align*}
The following computation is as the same as that for Lemma~\ref{EX for C_u,u'}, 
and finally we can get that 
$\E(X^{\dag})\le q^{-2\mu_q(n)
\big( g_q(\delta)-\frac{3}{4}-\frac{\log_q n}{\mu_q(n)}\big) +4}$.
\end{proof}

Then we can get the following two theorems similarly to
Theorem \ref{t C_u,u'} and Theorem \ref{self-dual 2-quasi good},
except that the existence of the codes in the two theorems is unconditional.

\begin{theorem}\label{t C_u^dag,u'^dag}
Assume that $\delta\in(0,1-\frac{1}{q})$ and $g_q(\delta)>\frac{3}{4}$,
and odd integers $n_1,n_2,\cdots$ are as in Eq.\eqref{odd n_1, ...}.
For $i=1,2,\cdots$,
let $G_i$ be any abelian group of order $n_i$, and 
 $C^{(i)}_{1^{\!\dag}\!,b^\dag}$ be the random 
 self-orthogonal 2-quasi-$\F G_i$ codes of dimension $n_i-1$ 
 as in Remark~\ref{FG^dag}. Then
$\lim\limits_{i\to\infty}
 \Pr\big(\Delta(C^{(i)}_{1^{\!\dag}\!,b^{\!\dag}})>\delta\big)=1$.
\end{theorem}

\begin{theorem}\label{self-orthogonal 2-quasi good}
Assume that $\delta$,  odd integers $n_1,n_2,\cdots$, and groups 
$G_1,G_2,\cdots$ are as in Theorem \ref{t C_u^dag,u'^dag}.
Then for $i=1,2,\cdots$ there exist self-orthogonal 2-quasi-$\F G_i$ codes $C^{(i)}$
such that $\lim\limits_{i\to\infty}\R(C^{(i)})=\frac{1}{2}$ and
$\Delta(C^{(i)})>\delta$ for all $i=1,2,\cdots$.
\end{theorem}

\begin{remark}\rm
In \cite{MW07} the binary self-dual 2-quasi-cyclic codes are proved to be
asymptotically good. 
The asymptotic goodness of self-dual quasi-cyclic codes (with index $\to\!\infty$)
were obtained in \cite{D}, \cite{LS}.  
The condition of the existence of self-dual 2-quasi-$\F G$ codes
was obtained in \cite{LS03} for cyclic $G$, and in \cite{LF} for abelian $G$. 
In the dissertation \cite{L PhD}, the asymptotic goodness of
the self-dual 2-quasi abelian codes (provided $q\,{\not\equiv}\,3\,({\rm mod}\,4)$) 
has been obtained by a probabilistic method. In \cite{AOS}, 
based on the Artin's primitive element conjecture,   
Alahmadi, \"Ozdemir and Sol\'e proved that 
the self-dual double circulant codes (a kind of 2-quasi cyclic codes)
are asymptotically good (if they exist); 
their method is by counting the codes. 
Recently, in \cite{LF} we introduced the 2-quasi-$\F G$ codes of type I,
which are just the double circulant codes once $G$ is cyclic; 
and, using a counting argument, we showed that 
the self-dual 2-quasi-abelian codes of type~I 
are asymptotically good (provided they exist). 
By a similar counting method, in \cite{LF}
we further proved the asymptotic goodness of 
the maximal self-orthogonal 2-quasi-abelian codes.
The two results of \cite{LF} are presented in this section 
but the arguments here are in the random style 
consistent with the whole chapter.
\end{remark}

\section{Dihedral group codes}\label{dihedral codes}

In this section we assume that:
\begin{itemize}
\item\vskip-5pt
 $\F$ is a finite field with $|\F|=q$;
\item \vskip-5pt
 $\delta\in(0,1-\frac{1}{q})$ and $g_q(\delta)>\frac{3}{4}$;
\item\vskip-5pt
 $n>1$ is an odd integer, $\gcd(n,q)=1$ and $\mu_q(n)>2\log_q n$;
\item\vskip-5pt
$\tilde G$ is a dihedral group of order $2n$, i.e., 
   $\tilde G$ has a normal cyclic subgroup $G\!=\!\{1,x,\cdots,x^{n-1}\}$ of order $n$  
  and  
  a $y\in\tilde G$ such that $y^2\!=\!1$, $yxy^{-1}\!=\!x^{-1}$; 
\item\vskip-5pt
$A=\F\tilde G$ is the group algebra of $\tilde G$ over $\F$.
\end{itemize}
\vskip-5pt
Then we have the following at once.
\begin{itemize}
\item \vskip-5pt
In $\tilde G$, $yG=Gy$, and $\tilde G=G\cup yG$ is a disjoint union; 
\item \vskip-5pt
$A=\F G\oplus \F Gy
=\big\{a(x)+b(x)y\,\big|\,a(x),b(x)\in\F G\big\}$
($\F G$-module decomposition), where
elements of $\F G$ are written as $a=a(x)=\sum_{i=0}^{n-1}a_ix^i$.
\end{itemize}
\vskip-5pt  
Because $yxy^{-1}=x^{-1}$, we have:
\begin{equation}\label{ya=bar a y}
 y a y^{-1}=\overline{a}~~(\mbox{equivalently,}~~
 ya=\overline a y,~ y\overline a= ay), ~~~ \forall~ a\in\F G.
\end{equation}  
where $a\mapsto\overline{a}$ is the bar map (Definition \ref{d bar map}).

Note that $G$ is a cyclic group of order $n$. 
The notation in Eq.\eqref{e E}---Eq.\eqref{d k_i k_r+j} 
is still applied to $\F G$. 
In particular, by Eq.\eqref{ya=bar a y}, any $e\in\widehat E$
is a central idempotent of $\F\tilde G$. 
There are two kinds of idempotents in $\widehat E$:
$e_i=\overline e_i$ for $0\le i\le r$, and
$\widehat e_{r+j}=e_{r+j}+\overline e_{r+j}$ for $1\le j\le s$. 
To simplify the story on dihedral codes,  
we'll consider one of the the two kinds by assuming that $s=0$. 

\bex[{\cite{KR}}]\label{e s=0 iff ...}
Let notation be as above, and $\langle q\rangle_n$ be the subgroup 
of ${\Bbb Z}_n^\times$ generated  by $q$, cf. Definition \ref{q-coset}.
Then $s=0$ if and only if $-1\in\langle q\rangle_n$. 
\hint{$\F Ge_t\cong\F[X]/\langle f_t(X)\rangle$ with 
$f_t(X)=\prod_{k\in Q_t}(X-\xi^k)$ as in Exercise \ref{cyclic mu(n)}(1);  
considering the zero-sets of codes, one can show that $\overline e_t=e_t$
(equivalently, $\overline{\F Ge_t}=\F Ge_t$) iff $-Q_t:=\{-k\,|\,k\in Q_t\}=Q_t$;  
if $-1\notin\langle q\rangle_n$, then $Q_g:=\langle q\rangle_n$ is itself a 
$q$-coset and $-Q_g\ne Q_g$.}
\eex

\begin{remark}\label{r s=0}\rm
In the following we further assume that $-1\in\langle q\rangle_n$
(there are infinitely many such integers $n$, see Eq.\eqref{-1 odd n_1, ... } below).
Then, by Exercise \ref{e s=0 iff ...}, $s=0$, and $1=e_0+e_1+\cdots+e_r$ with
$\overline e_i=e_i$ for $i=0,1,\cdots,r$ being 
all the primitive central idempotents of $\F G$; 
$e_ie_{i'}=0$ for $0\le i\ne i'\le r$;  
and, 
$$
 \F G=\F_0\oplus \F_1\oplus\cdots\oplus\F_r,
$$
where $\F_i=\F Ge_i$ for $0\le i\le r$ are finite fields. 
For $1\le i\le r$, $|\F_i|=q^{2k_i}$ and $|\F_i^\flat|=q^{k_i}$, 
where $\F_i^\flat=(\F Ge_i)^\flat$ is described in Lemma~\ref{l dim FG e_i}. 

Turn to $A=\F\tilde G$, 
by Lemma \ref{complete system idempotent},
\begin{equation}\label{e oplus A_i}
 A=\F\tilde G=A_0\oplus A_1\oplus\cdots\oplus A_r,  
~~~ \mbox{where $A_i=\F\tilde Ge_i$.} 
\end{equation}
By Lemma \ref{l orthogonal and bar}(3), 
$\langle A_i,A_{i'}\rangle=0$ for $0\le i\ne i'\le r$. 
So Eq.\eqref{e oplus A_i} is an orthogonal direct sum, 
hence the restriction of the inner product 
to any $A_i$, $0\le i\le r$, is non-degenerate. 
\end{remark}

We show the structures of all $A_i$ in Eq.\eqref{e oplus A_i}, 
construct our dihedral codes, then
exhibit their asymptotic properties.
We begin with $A_0$ whose behaviour is different from the others.

\begin{lemma}\label{l A_0=...}
For $e_0=\frac{1}{n}(1+x+\cdots+x^{n-1})$, set $e_{00}=e_0+e_0y$. 
Then $A_0=\F\tilde Ge_0=\F Ge_0\oplus\F Ge_0y$,
$C_0:=\F\tilde G e_{00}=\big\{\lambda\sum_{z\in\tilde G}z\,\big|\,\lambda\in\F\big\}
=\F e_{00}$ (``all-one code''), and the following hold:

{\rm(1)} If $q$ is odd, then $\frac{1}{2} e_{00}$ is a central idempotent, 
$\langle e_{00},e_{00}\rangle=\frac{2}{n}\ne 0$, 
and $\F\tilde Ge_0=\F\tilde G e_{00}\oplus\F\tilde Ge_{01}$,
where $e_{01}=e_0-e_0y$ and 
$\F\tilde Ge_{01} =\F e_{01}$.

{\rm(2)} If $q$ is even, then $e_{00}^2=0$, $\langle C_0,C_0\rangle=0$, and
$\F\tilde G e_0/\F\tilde G e_{00}\cong \F\tilde G e_{00}$.
\end{lemma}
\begin{proof}
(1). $e_{00}^2=2e_{0}+2e_{0}y=2e_{00}$; hence
 $(\frac{1}{2}e_{00})^2=\frac{1}{4}(2e_{00})=\frac{1}{2}e_{00}$.
And $\frac{1}{2}e_{00}$ is central, since 
$ye_{00}y^{-1}=ye_0y^{-1}+y(e_0y)y^{-1}=e_{00}$. 
And $e_0-\frac{1}{2}e_{00}=\frac{1}{2}e_{01}$. 
Finally, $ye_{01}=ye_0-ye_0y=-e_0+e_0y=-e_{01}$; hence
$\F\tilde Ge_{01}=\F e_{01}$. 

(2). $e_{00}^2=2e_{00}=0$. The map
$\F\tilde Ge_0\to \F\tilde e_{00}$, $\tilde ae_0\mapsto \tilde ae_{00}$, 
is a surjective homomorphism with kernel $\F\tilde Ge_{00}$.
\end{proof}

Keeping the notation in Remark \ref{r s=0}, we are going on for $1\le i\le r$. 
By ${\rm M}_2(\F)$ we denote the $2\times 2$ matrix algebra over $\F$.

\begin{lemma}\label{l A_i=...} 
$A_i\cong {\rm M}_2(\F_i^\flat)$ for $1\le i\le r$.  
 \end{lemma}

\begin{proof}
 The binary case of the lemma was proved in \cite{BM}. 
We'll  proof the lemma by constructing a specific isomorphisms 
\eqref{F tilde Ge cong} for later quotation.  

Note that $e_i$ is the identity of the ring $A_i=\F_i\oplus \F_iy$.
By Lemma~\ref{l dim FG e_i}, 
$\F_i$ is an extension over $\F_i^\flat$ of degree $|\F_i:\F_i^\flat|=2$, 
and $\alpha\mapsto \overline\alpha$ is an automorphism of $\F_i$ of order $2$.
Note that $\F_i=\F Ge_i=(\sum_{j=0}^{n-1}\F x^j)e_i=\sum_{j=0}^{n-1}\F (xe_i)^j$. 
So $xe_i\in \F_i-\F_i^\flat$, and $\F_i=\F_i^\flat +\F_i^\flat(xe_i)$.
Then the minimal polynomial $\varphi_{xe_i}(X)$ of $xe_i$ 
over $\F_i^\flat$ is an irreducible polynomial of degree $2$, 
say $\varphi_{xe_i}(X)=X^2+gX+h\in\F_i^\flat[X]$.
Then $(xe_i)^2+g(xe_i)+h=0$; Since $\overline g=g$ and $\overline h=h$,  
$(\overline{xe_i})^2+g(\overline{xe_i})+h=\overline{(xe_i)^2+g(xe_i)+h}=0$.
Thus $xe_i$ and $\overline{xe_i}$ are two roots of $\varphi_{xe_i}(X)$. 
So, $h=(xe_i)(\overline{xe_i})=x\overline x e_ie_i=xx^{-1}e_i=e_i$, and 
\begin{equation}\label{e mini polynomial}
 \varphi_{xe_i}(X)=X^2+gX+1,~~
 \mbox{for a~ $g\in\F_i^\flat$ (note that $e_i=1$ in $\F_i^\flat$)}, 
\end{equation}
where $g$~and~$2$ are not both zero 
since $\varphi_{xe_i}(X)\in \F_i^\flat[X]$ is irreducible.  
Set: 
\begin{equation}\label{eta nu}
\varepsilon=\begin{pmatrix}1&0\\0&1\end{pmatrix},\quad
 \eta=\begin{pmatrix}0 & -1\\1& -g \end{pmatrix},\quad
  \nu=\begin{pmatrix}-1&0\\-g&1 \end{pmatrix}
~\in~{\rm M}_2(\F_i^\flat).
\end{equation}
Then the characteristic polynomial of $\eta$ is
$\varphi_\eta(X)=X^2+gX+1=\varphi_{xe_i}(X)$.
We get a subalgebra $N=\F_i^\flat\varepsilon+\F_i^\flat\eta$ 
of ${\rm M}_2(\F_i^\flat)$ and, 
mapping $e_i\mapsto\varepsilon$ and $xe_i\mapsto\eta$,
we obtain a field isomorphism:
\begin{equation}\label{FGe cong}
\F_i=\F_i^\flat e_i + \F_i^\flat xe_i
 \cong \F_i^\flat [X]/\langle\varphi_{\eta}(X)\rangle
\cong \F_i^\flat\varepsilon+\F_i^\flat\eta =N
\subseteq{\rm M}_2(\F_i^\flat).
\end{equation}
Since $\dim_N{\rm M}_2(\F_i^\flat)=2$, we get that
$$
{\rm M}_2(\F_i^\flat)=N+N\nu=
\F_i^\flat\varepsilon+\F_i^\flat\eta+\F_i^\flat\nu+\F_i^\flat\eta\nu.
$$
On the other hand,
$$\F\tilde Ge_i=\F Ge_i+\F Ge_iy=\F_i+\F_i (ye_i)=
 \F_i^\flat e_i +\F_i^\flat xe_i + \F_i^\flat ye_i+\F_i^\flat xye_i.$$
Since 
$\nu^2=\varepsilon$ and
$\nu\eta\nu^{-1}=
       \begin{pmatrix} -g & 1\\-1&0\end{pmatrix}=\eta^{-1}$,  
by mapping  $ye\mapsto\nu$, $xye\mapsto\eta\nu$,
we extend the isomorphism Eq.\eqref{FGe cong}
to the following algebra isomorphism (where $a,b,c,d\in\F_i^\flat$):
\begin{equation}\label{F tilde Ge cong}
\begin{array}{rcl}
(A_i=) ~~ \F\tilde Ge_i& \mathop{\longrightarrow}\limits^{\cong}
 & {\rm M}_2(\F_i^\flat),\\[3pt]
 ae_i +b xe_i+cye_i+dxye_i
 &\longmapsto &
a\varepsilon +b \eta+c\nu+d\eta\nu.
\end{array}
\end{equation}
We are done.
\end{proof}

Some properties of $2\times 2$ matrix algebras are listed below 
for later citation.

\bex\label{M algebra}
(1) $M:={\rm M}_2(\F)$ has a subalgebra $N$ which is a field with $\dim_F N=2$.
\hint{element $u$ of $\F_{q^2}^\times$ with order $q+1$ 
has minimal polynomial $\varphi_u(X)=X^2+gX+1\in\F[X]$
cf. the proof of Eq.\eqref{e mini polynomial}; take 
$\eta=\begin{pmatrix} 0 & -1\\ 1 & -g \end{pmatrix}$; 
let $N=\F\varepsilon+\F\eta$,  
cf. Eq.\eqref{eta nu}, Eq.\eqref{FGe cong}.}

(2) 
$L:=M\varepsilon_{11}$, where 
$\varepsilon_{11}=\begin{pmatrix} 1&0\\ 0&0\end{pmatrix}$,
 is a simple left ideal of~$M$, and $L=N\varepsilon_{11}$.  
\hint{$\dim_\F L=2$, and $\dim_\F N\varepsilon_{11}=\dim_\F N=2$;
so $M\varepsilon_{11}=N\varepsilon_{11}$.}

(3)
 For any $\alpha,\beta\in N^\times$, 
 $\alpha L\beta=\{\alpha c\beta\,|\,c\in L\}=L\beta$ is a simple left ideal of $M$.
 \hint{$\alpha$ is invertible in $M$, so $\alpha M=M$; 
$\alpha L\beta=M\varepsilon_{11}\beta$ and ${\rm rank}(\varepsilon_{11}\beta)=1$}

(4) For $0\ne c\in L$ and $a,b\in N^\times$, 
$ac=cb$ if and only if $a=b\in (\F\varepsilon)^\times$.
\hint{ for $c=\varepsilon_{11}$, 
 $(a_1\varepsilon+a_2\eta)\varepsilon_{11}=
 \varepsilon_{11}(b_1\varepsilon+b_2\eta)$ with $a_i,b_i\in\F$  
 implies that 
 $\begin{pmatrix} a_1\!-\!b_1 & b_2 
 \\ a_2 & 0  \end{pmatrix}\!=\!0$, hence $a_2=0=b_2$ and $a_1=b_1$;
for general $c$, by (2), $c=d\varepsilon_{11}$ for a $d\in N^\times$.}
\eex

We refine now the notation in Remark \ref{r s=0}
and construct our dihedral codes.

\begin{remark}\label{r dihedral further notation}\rm
For $1\le i\le r$, we have the isomorphisms $A_i\cong{\rm M}_2(\F_i^\flat)$
in Eq.\eqref{F tilde Ge cong}
(and cf. Exercise~\ref{M algebra}), and have the field 
$\F_i\!\subseteq\! A_i$ with $|\F_i|=q^{2k_i}$ corresponding to the field 
$N\!\subseteq\!{\rm M}_2(\F_i^\flat)$ 
(cf. Eq.\eqref{FGe cong} and Exercise~\ref{M algebra}(1));
we denote
\begin{itemize}
\item \vskip-5pt
$Z_i:={\rm Z}(A_i)=\F_i^\flat$ corresponding to the center 
${\rm Z}\big({\rm M}_2(\F_i^\flat)\big)$, so $|Z_i|=q^{k_i}$;  
\item \vskip-5pt
$C_i=A_i(e_i-ye_i)$, which is the simple left ideal of $A_i$ corresponding to
 $L={\rm M}_2(\F_i^\flat)\varepsilon_{11}$, so $|C_i|=q^{2k_i}$.  \\
({\it proof}: by Eq.\eqref{eta nu},  $e_i-ye_i$ corresponds to 
$\varepsilon-\nu=\begin{pmatrix}2&0\\ g&0\end{pmatrix}$ 
whose first column is non-zero, see Eq.\eqref{e mini polynomial}; 
hence ${\rm M}_2(\F_i^\flat)(\varepsilon-\nu)
={\rm M}_2(\F_i^\flat)\varepsilon_{11}$.)
\end{itemize}
\vskip-5pt
And for the whole group algebra $A=\F\tilde G$,  we set:
\begin{itemize}
\item \vskip-5pt
$C=C_0\oplus C_1\oplus\cdots\oplus C_{r}$
where $C_0=\F e_{00}$ (``all-one code'') as in Lemma~\ref{l A_0=...}; 
so $C$ is a dihedral code with $\dim_F C=n$.
\item \vskip-5pt
$K^*= \{e_0\}\times \F_1^\times\times \cdots\times \F_{r}^\times$,
which is a subgroup of the multiplicative unit group $A^\times$;
so $|K^*|=\prod_{i=1}^{r}(q^{2k_i}-1)$.
\end{itemize}
\end{remark}

\begin{definition}\label{d random dihedral}\rm
Consider $K^*\times K^*$ as a probability space with 
equal probability for each sample. Let $(\alpha,\beta)\in K^*\times K^*$.
By Exercise \ref{M algebra}(3), we have a class of random dihedral codes. 
\begin{itemize}
\item[(1)] \vskip-5pt
$C_{\alpha,\beta}=\alpha C\beta$ is a random dihedral code with
$\R(C_{\alpha,\beta})=\frac{1}{2}$;
\item[(2)] \vskip-5pt
$X_c=\begin{cases}1, & 0<\frac{\w(\alpha c\beta)}{2n}\le\delta;\\
  0, & \mbox{otherwise;} \end{cases}$\quad $c\in C$; 
\item[(3)] \vskip-5pt
$X=\sum_{c\in C}X_c$; ~ 
hence $\Pr(\Delta(C_{\alpha,\beta}\le\delta)=\Pr(X\ge 1)\le \E(X)$.
\end{itemize}
\end{definition}

\begin{lemma}\label{C cap C^bot}
{\rm(1)} If $q$ is even, then $\alpha C\beta=C\beta$ is a self-dual dihedral code. 

{\rm(2)} If $q$ is odd, then $\alpha C\beta=C\beta$ is 
an LCD dihedral code of rate $\frac{1}{2}$. 
\end{lemma}
\begin{proof}
Let $\beta=\sum_{i=0}^r \beta_i$ with $\beta_0=e_0$ and 
$\beta_i=e_i\beta\in\F_i^\times$ for $1\le i\le r$. Similarly, $\alpha=\sum_{i=0}^r \alpha_i$.
Then $\alpha\beta=\sum_{i=0}^r \alpha_i\beta_i$
with $\alpha_0\beta_0=e_0$ and 
$\alpha_i\beta_i\in\F_i^\times$ for $1\le i\le r$. 
So, by Exercise \ref{M algebra}(3),  
$$\textstyle
\alpha C\beta=\bigoplus_{i=0}^r\alpha_iC_i\beta_i
=\bigoplus_{i=0}^r C_i\beta_i=C\beta.
$$
Further, by the orthogonal direct sum Eq.\eqref{e oplus A_i},
\begin{equation}\label{C^bot}
\textstyle
(C\beta)^\bot=\bigoplus_{i=0}^{r}(C_i\beta_i)^{\bot_{A_i}},
\end{equation} 
where $(C_i\beta_i)^{\bot_{A_i}}$ denotes 
the orthogonal subspace of $C_i\beta_i$ in $A_i$. 
Note that $(C_i\beta_i)\cap (C_i\beta_i)^{\bot_{A_i}}$ is a left ideal
of $A_i$, hence it is either $C_i\beta_i$ or $0$
(because $C_i\beta_i$ is a simple left ideal of $A_i$). 
By Lemma~\ref{l orthogonal and bar}(2), 
\begin{equation}\label{C cap C}
 (C_i\beta_i)\cap (C_i\beta_i)^{\bot_{A_i}}
=\begin{cases}C_i\beta_i, & \mbox{if }C_i\beta_i\overline{C_i\beta_i}=0;\\
    0, & \mbox{otherwise.}\end{cases}
\end{equation}
For $i=1,\cdots,r$, $C_i\beta_i\overline{C_i\beta_i}=(A_i(e_i-ye_i)\beta_i)
(\overline\beta_i \overline{(e_i-ye_i)}A_i)$;
since $\overline e_i=e_i$ and $\overline y=y$, 
$C_i\beta_i\overline{C_i\beta_i}
=A_i(e_i-ye_i)\beta_i\overline\beta_i(e_i-ye_i)A_i$.
But, $\beta_i\overline\beta_i\in\F_i$ and 
$\overline{\beta_i\overline\beta_i}=\beta_i\overline\beta_i$;
thus $\beta_i\overline\beta_i\in \F_i^\flat=Z_i$ (cf. Lemma \ref{l dim FG e_i}). 
Hence
$$
(e_i-ye_i)\beta_i\overline\beta_i(e_i-ye_i)
=(e_i-ye_i)(e_i-ye_i)\beta_i\overline\beta_i
=2(e_i-ye_i)\beta_i\overline\beta_i.
$$
So, $C_i\beta_i\overline{C_i\beta_i}=0$ if and only if $q$ is even.
For $i=0$, by Lemma \ref{l A_0=...} we still have that 
$C_0\beta_0\overline{C_0\beta_0}=0$ if and only if $q$ is even.
Then the lemma follows from Eq.\eqref{C^bot} and Eq.\eqref{C cap C}
at once.
\end{proof}

Given a $c\in C$; then $c=c_0+c_1+\cdots+c_r$ with $c_i=e_ic\in C_i$; 
we denote:  
\begin{equation}\label{e omega}
\textstyle
\omega=\{i\,|\,1\le i\le r,\, c_i\ne 0\},~~~ 
\ell_c=\sum_{i\in\omega}k_i, ~~~ 
\omega'=\{1,\cdots,r\}-\omega.
\end{equation}
Note that: if $\omega=\emptyset$ (equivalently, $\ell_c=0$), 
then $c=c_0$, so $\w(\alpha c\beta)=0$ if $c=0$, 
and $\w(\alpha c\beta)=2n$ otherwise; 
hence $\E(X_c)=0$ (see Definition \ref{d random dihedral}(2)).

\begin{lemma}\label{l E(X_c)}
~ $\E(X_c)\le q^{-4\ell_c g_q(\delta)+\ell_c+5}$.
\end{lemma}
\begin{proof}
We can assume that $\omega\ne\emptyset$.
Denote $K^*cK^*=\{\alpha c\beta\,|\,\alpha,\beta\in K^*\}$.
Given $d\in K^*cK^*$. 
By $f$ we denote the number of such pairs 
$(\alpha,\beta)\in K^*\times K^*$ that $\alpha c\beta=d$.
We write $d=\alpha' c\beta'$ for a given $(\alpha',\beta')\in K^*\times K^*$. 
Then $\alpha c\beta=d$,  
i.e., $\alpha'^{-1}\alpha c=c\beta'\beta^{-1}$, if and only if 
\begin{equation}\label{i-components}
\alpha_i'^{-1}\alpha_i c_i=c_i\beta_i'\beta_i^{-1}, ~~~~~ ~i=0,1,\cdots,r.
\end{equation}
For $i=0$ the above equality always holds. For $i=1,\cdots,r$, there are two cases.

Case 1: $i\in\omega$, i.e., $0\ne c_i\in C_i$. By Exercise \ref{M algebra}(4), 
$\alpha_i'^{-1}\alpha_i =\beta_i'\beta_i^{-1}\in Z_i^\times$.
For any $z_i\in Z_i^\times$ we get a unique 
$\alpha_i=\alpha' z$ and $\beta_i=\beta' z^{-1}$
satisfying Eq.\eqref{i-components}.
Thus there are $|Z_i^\times|=q^{k_i}-1$ pairs
 $(\alpha_i,\beta_i)\in\F_i^\times\times\F_i^\times$
satisfying Eq.\eqref{i-components}. 

Case 2: $i\in\omega'$, i.e., $c_i=0$. Then
any $(\alpha_i,\beta_i)\in\F_i^\times\times\F_i^\times$
satisfies Eq.\eqref{i-components}. The number of such pairs equals 
 $|\F_i^\times\times\F_i^\times|=(q^{2k_i}-1)^2$.

Summarizing the two cases, we see that the number 
$$\textstyle
 f =\prod_{i\in\omega}|Z_i^\times|\cdot
 \prod_{i\in\omega'}|\F_i^\times\times\F_i^\times|,
$$
which is independent of the choice of $d\in K^* c K^*$. Thus
$$
 \E(X_c)=f\cdot|(K^* c K^*)^{\le\delta}|\big/|K^*\times K^*|, 
$$
where the notation $(K^* c K^*)^{\le\delta}$ is defined in Eq.\eqref{C^<delta}.
Note that $|K^*\times K^*|=\prod_{i=1}^r|\F_i^\times\times\F_i^\times|$.
Combining it with the the above equation on $f$, we get
$$
\textstyle
 \E(X_c)=\frac
   {\prod_{i\in\omega}|Z_i^\times|\cdot|(K^* c K^*)^{\le\delta}|}
   {\prod_{i\in\omega}|\F_i^\times\times\F_i^\times|}.
$$
Since $K^* c K^*\subseteq C_0\bigoplus\big(\bigoplus_{i\in\omega}A_i\big)$
and 
$$\textstyle
\dim_\F C_0\bigoplus\big(\bigoplus_{i\in\omega}A_i\big)
=1+\sum_{i\in\omega}4k_i=1+4\ell_c,
$$
by Lemma \ref{g-code balanced} and Corollary \ref{c balanced}, we have that
$$\textstyle
|(K^* c K^*)^{\le\delta}|\le 
\big|\big(C_0\bigoplus\big(\bigoplus_{i\in\omega}A_i\big)\big)^{\le\delta}\big|
\le q^{(1+4\ell_c) h_q(\delta)}.
$$
Next, 
$$\textstyle
\prod_{i\in\omega}|Z_i^\times|=\prod_{i\in\omega}(q^{k_i}-1)
\le \prod_{i\in\omega} q^{k_i} =q^{\sum_{i\in\omega}k_i}=q^{\ell_c}.
$$
By Exercise \ref{e q^k-1} (recall that 
$2k_i\ge\mu_q(n)>\log_q n$ by our assumption),
$$\textstyle
\prod_{i\in\omega}|\F_i^\times\times\F_i^\times|
=\prod_{i\in\omega}(q^{2k_i}-1)^2\ge q^{2(-2+\sum_{i\in\omega}2k_i)}
=q^{4\ell_c-4}.
$$
Thus (note that $h_q(\delta)<1$),
$$\textstyle
\E(X_c)\le\frac{q^{\ell_c}\cdot q^{(1+4\ell_c) h_q(\delta)}}{q^{4\ell_c-4}}
\le q^{-4\ell_c g_q(\delta)+\ell_c+5}.
$$
We are done.
\end{proof}

Similarly to Eq.\eqref{e Omega}, 
for an integer $\ell$ with $\lceil\frac{1}{2}\mu_q(n)\rceil\le\ell\le\frac{n-1}{2}$, 
we set:
\begin{equation}\label{e dihedral Omega}
\Omega_{2\ell}=\{L\,|\, 
\mbox{$L$ is a direct sum of some of $C_1,\cdots,C_r$,\, $\dim L=2\ell$}\};
\end{equation}
and for $L\in\Omega_{2\ell}$ we further set: 
\begin{equation}\label{e dihedral tilde L}
\tilde L=C_0\oplus L, ~~~~~ \tilde L^*=\{c\in\tilde L\,|\, \ell_c=\ell\}.   
\end{equation}
Similarly to Eq.\eqref{|Omega|<...}, 
we have
\begin{equation}
\textstyle
 C=C_0\bigcup\Big(
 \bigcup_{\ell=\lceil\frac{1}{2}\mu_q(n)\rceil}^{\frac{n-1}{2}} 
 \bigcup_{L\in\Omega_{2\ell}}\tilde L^*\Big);
~~~~~
|\Omega_{2\ell}|\le  n^{2\ell/\mu_q(n)}.
\end{equation}

\begin{lemma}
If $g_q(\delta)-\frac{3}{4} -\frac{\log_q n}{\mu_q(n)}>0$, then
 $\E(X)\le q^{-2\mu_q(n)\big(g_q(\delta)-\frac{3}{4}
 -\frac{\log_q n}{\mu_q(n)}\big)+6}$.
\end{lemma}

\begin{proof} 
The proof is similar to the proof of Lemma \ref{EX for C_u,u'}.
We have that $\E(X_c)=0$ once $c\in C_0$, and 
\begin{equation*}
\textstyle\E(X)
 =\sum_{c\in C}\E(X_c)
 =\sum_{\ell=\lceil\frac{1}{2}\mu_q(n)\rceil}^{\frac{n-1}{2}}
   \sum_{L\in\Omega_{2\ell}}
 \sum_{c\in \tilde L^{\!*}}\E(X_c).
\end{equation*}
Given $L\in\Omega_{2\ell}$;  
by Eq,\eqref{e dihedral tilde L}, 
$|\tilde L|=q^{2\ell+1}$ and $\ell_c=\ell$ for $c\in\tilde L^*$;
we have 
\begin{align*}
&\textstyle\sum_{c\in \tilde L^{\!*}}\E(X_c)
\le \sum_{c\in \tilde L^{\!*}} q^{-4\ell_c g_q(\delta)+\ell_c+5}
=\sum_{c\in \tilde L^{\!*}} q^{-4\ell g_q(\delta)+\ell+5}
\\
&\textstyle
= |\tilde L^*|\cdot q^{-4\ell g_q(\delta)+\ell+5}
\le q^{2\ell+1}q^{-4\ell g_q(\delta)+\ell+5}
 =q^{-4\ell g_q(\delta)+3\ell+6}.
\end{align*}
Thus
\begin{align*}
\E(X)
& 
 \le \sum_{\ell=\lceil\frac{1}{2}\mu_q(n)\rceil}^{(n-1)/2}
   \sum_{L\in\Omega_{2\ell}}q^{-4\ell g_q(\delta)+3\ell+6}
  \le \sum_{\ell=\lceil\frac{1}{2}\mu_q(n)\rceil}^{(n-1)/2}
  n^{2\ell/\mu_q(n)} q^{-4\ell g_q(\delta)+3\ell+6}
\\
&
 =\!\sum_{\ell=\lceil\frac{1}{2}\mu_q(n)\rceil}^{(n-1)/2}\!
  q^{-4\ell g_q(\delta)+3\ell+\frac{2\ell\log_q n}{\mu_q(n)} +6} 
 \!=\!\sum_{\ell=\lceil\frac{1}{2}\mu_q(n)\rceil}^{(n-1)/2}\!
  q^{-4\ell \big( g_q(\delta)-\frac{3}{4}-\frac{\log_q n}{2\mu_q(n)}\big) +6}.
\end{align*}
Since $g_q(\delta)-\frac{3}{4}-\frac{\log_q n}{2\mu_q(n)}>0$ and
$2\ell\ge\mu_q(n)$, we obtain
\begin{align*}
\E(X)
\le 
  \!\sum_{\ell=\lceil\frac{1}{2}\mu_q(n)\rceil}^{(n-1)/2}\!
 q^{-2\mu_q(n)
  \big( g_q(\delta)-\frac{3}{4}-\frac{\log_q n}{2\mu_q(n)}\big) +6}
 \le n\, q^{-2\mu_q(n)
  \big( g_q(\delta)-\frac{3}{4}-\frac{\log_q n}{2\mu_q(n)}\big) +6}.
\end{align*}
That is, $\E(X)\le q^{-2\mu_q(n)
\big( g_q(\delta)-\frac{3}{4}-\frac{\log_q n}{\mu_q(n)}\big) +6}$.
\end{proof}

\begin{remark}\rm
In Exercise \ref{e n_1,...} we have seen that the natural density
$\lim\limits_{t\to\infty}\frac{|{\cal G}_t|}{\pi(t)}=1$.
Note that: for a prime $p>q$, 
the multiplicative group ${\Bbb Z}_p^\times$ is a cyclic group of order $p-1$; 
so $\mu_q(p)=|\langle q\rangle_p|$, see Exercise \ref{cyclic mu(n)}(4);
and $-1$ is the unique element of order $2$ in ${\Bbb Z}_p^\times$;
in conclusion, $-1\in\langle q\rangle_p$ if and only if $\mu_q(p)$ is even.
It is known from \cite{Hasse} (for Dirichlet density) and \cite{O} (for natural density)
that the density of such primes $p>q$ that $\mu_q(p)$ is even is positive.
Thus, we can find primes $p_1,p_2,\cdots$ in ${\cal G}_t$ for $t\to\infty$
such that $\lim\limits_{i\to\infty}\frac{\log_q p_i}{\mu_q(p_i)}=0$ and  
$-1\in\langle q\rangle_{p_i}$ for $i=1,2,\cdots$. 
Further, similar to Remark \ref{r n_1, ...}, we have odd positive integers
$n_1,n_2,\cdots$ coprime to $q$ such that
\begin{equation}\label{-1 odd n_1, ... }
\textstyle
\lim\limits_{i\to\infty}\frac{\log_q n_i}{\mu_q(n_i)}=0; ~~~\mbox{and}~~~  
 -1\in\langle q\rangle_{n_i}, \forall~i=1,2,\cdots;
\end{equation}
and $n_1,n_2,\cdots$ are not necessarily primes (e.g., we can take $n_i=p_i^2$).
\end{remark}

\begin{theorem}\label{t C_alpha,beta}
Assume that $\delta\in(0,1-\frac{1}{q})$ and $g_q(\delta)>\frac{3}{4}$,  
and the integers $n_1,n_2,\cdots$ are as in Eq.\eqref{-1 odd n_1, ... }.
For $i=1,2,\cdots$,
let $C^{(i)}_{\alpha,\beta}$ be the random dihedral codes
 defined in Definition \ref{d random dihedral}. Then
$\lim\limits_{i\to\infty}\Pr\big(\Delta(C^{(i)}_{\alpha,\beta})>\delta\big)=1$.
\end{theorem}
\begin{proof}
 It is the same as the proof of Theorem \ref{t C_u,u'}. 
\end{proof}

Combining it with Lemma \ref{C cap C^bot}, we get the following at once.

\begin{theorem}\label{dihedral codes good}
Assume that $\delta\in(0,1-\frac{1}{q})$ and $g_q(\delta)>\frac{3}{4}$, 
and the integers $n_1,n_2,\cdots$ are as in Eq.\eqref{-1 odd n_1, ... }. 
For $i\!=\!1,2,\cdots$, there exist dihedral codes $C^{(i)}$ of rate $\frac{1}{2}$
such that $\Delta(C^{(i)})>\delta$ for all $i=1,2,\cdots$, 
and, if $q$ is even then $C^{(1)},C^{(2)},\cdots$ are all self-dual codes;
otherwise $C^{(1)},C^{(2)},\cdots$ are all LCD codes.
\end{theorem}

\begin{remark}\rm
Finite dihedral groups are non-abelian but near to cyclic groups,
 they have cyclic subgroups of index $2$. 
The asymptotic goodness of binary dihedral codes was proved in \cite{BM}.
And \cite{BW} studied the semidirect products of a cyclic group of order prime $p$
by a cyclic group and showed the asymptotic goodness of such group codes.
In \cite{FL20} we extended the asymptotic goodness of dihedral codes to 
any $q$-ary case, and investigated more precisely the algebraic properties 
of the asymptotically good dihedral codes. 
In this section we present a simplified description 
of the results on $q$-ary dihedral codes. 
\end{remark}

\section{Appendix: About probabilistic method}\label{appendix}

Here is a sketch about the probabilistic method used in this chapter. 
See, e.g., \cite{MU} for details please.

If ${\mathscr P}$ is a finite set with  
a function $p:{\mathscr P}\to{\Bbb R}$ 
(where ${\Bbb R}$ denotes the real number field) satisfying that:
$p(s)\ge 0$, $\forall$ $s\in{\mathscr P}$; and
$\sum_{s\in{\mathscr P}}p(s)=1$;
then we say that ${\mathscr P}$ is a finite {\em probability space}
with the probability function $p$.

In the following ${\mathscr P}$ is a finite probability space.
Any element of ${\mathscr P}$ is called a {\em sample}.
Any subset $E\subseteq{\mathscr P}$ is called an {\em event}, 
and $\Pr(E)=\sum_{s\in E}p(s)$ is called the probability of the event $E$.  
Note that, we usually took the probability space ${\mathscr P}$ with
equal probability for samples, i.e., $p(s)=p(s')$ for all $s,s'\in{\mathscr P}$,
thus $\Pr(E)=\frac{|E|}{|{\mathscr P}|}$ in that case.

Let $E$, $F$ be two events over ${\mathscr P}$.
By $\Pr(E\cap F)$ we denote the probability that both $E$ and $F$ occur.
By $\Pr(E\cup F)$ we denote the probability that $E$ or $F$ occurs.
By $\Pr(E|F)$ we denote the conditional probability of $E$ when $F$ occurred. 
The following hold:

\begin{lemma}[Conditional Probability Formula]\label{conditional}
  $\Pr(E\cap F)=\Pr(E|F)\cdot\Pr(F)$.
\end{lemma}

If $\Pr(E\cap F)=\Pr(E)\cdot\Pr(F)$ (equivalently, $\Pr(E|F)=\Pr(E)$), 
then $E$ and $F$ are said to be {\em randomly independent}. 

\begin{lemma}[Total Probability Formula]\label{l total}
 If $E\cap F=\emptyset$ (impossible event) and 
 $E\cup F={\mathscr P}$ (certain event), then for any event $G$, \par\vskip5pt
 \centerline{$\Pr(G)=\Pr(G|E)\cdot\Pr(E)+\Pr(G|F)\cdot\Pr(F).$}
\end{lemma}

\bex\label{e total}
Let notation be as in Lemma \ref{l total}. 

(1) If $\Pr(E)>0$, then
 $\Pr(G|E)\ge\frac{\Pr(G)-\Pr(F)}{\Pr(E)}$.

(2) If $\Pr(G)=1$ and $\Pr(E)>0$, then $\Pr(G|E)=1$.
\eex

\begin{definition}\label{d random var}\rm
If $X: {\mathscr P}\to{\Bbb R}$ is a function, then we say that
$X$ is a {\em (real) random variable} over ${\mathscr P}$; and
$\Pr(X\!=\!r)\;=\sum_{s\in{\mathscr P}, X(s)=r}p(s)$ for $r\in{\Bbb R}$
is the {\em distribution} of $X$, and
$\E (X)=\sum_{x\in{\Bbb R}}x\cdot\Pr(X\!=\!x)$ 
is called the {\em expectation} of~$X$.
\end{definition}

Let $X$, $Y$ be real random variables, $r\in {\Bbb R}$. 
As functions  from ${\mathscr P}$ to ${\Bbb R}$,
the $X\!+\!Y$, $XY$, $rX$ are defined, and are still random variables.
It is checked directly that
\begin{align*}
&  \E (X+Y)= \E (X)+ \E (Y), \qquad  \E (rX)=r \E (X);\\
&  \E (X^2)\ge 0, \qquad  \E (X^2)=0 \iff X=0.
\end{align*}
Thus we get:

\begin{lemma}\label{l E}
The set of all real random variables form a real vector space, and,
\begin{itemize}
\item\vskip-5pt
 $\E (X)$ is a linear form of the real vector space;
\item \vskip-4pt
 $\langle X,Y\rangle:= \E (XY)$ is an inner product of the real vector space.
\end{itemize} 
\vskip-5pt
In particular, from the usual linear algebra we have

\smallskip\noindent{\bf Cauchy-Schwartz Inequality}: ~
$ \E (XY)^2\le \E (X^2)\cdot \E (Y^2)$.
\end{lemma}
 
If the variable $X$ takes value either $0$ or $1$, we say that $X$ is a 
{\em $0$-$1$ variable (Bernoulli variable)}. 
For a $0$-$1$ variable $X$ we have:  
\begin{equation}\label{E of 0-1}
 \E (X)=0\cdot\Pr(X\!\!=\!0)+1\cdot\Pr(X\!\!=\!1)=\Pr(X\!\!=\!1).
\end{equation} 

For any event $F$ there is a $0$-$1$ variable ${\bf 1}_F$ 
(called the {\em indicator variable} of~$F$): for $s\in{\mathscr P}$, 
${\bf 1}_F(s)=\begin{cases}1,\!\! & s\in F;\\ 0,\!\! & s\notin F.\end{cases}$ 
In particular, $ \E ({\bf 1}_F)=\Pr(F)$.

For a random variable $X$, 
the so-called {\em first moment method} uses the following inequality
to bound the probability $\Pr\big(X\ge a\big)$ from above.

\begin{lemma}[{\bf Markov Inequality}]\label{Markov}
Let $X$ be a non-negative random variable. If $a>0$, then
$\Pr\big(X\ge a\big)\le\frac{ \E (X)}{a}$.
\end{lemma}
\begin{proof}
Since ${\bf 1}_{X\ge a}\le \frac{X}{a}$, 
$\Pr\big(X\ge a\big)= \E \big({\bf 1}_{X\ge a}\big)
  \le  \E \big(\frac{X}{a}\big)=\frac{ \E (X)}{a}.$
\end{proof}

To estimate the lower bound of $\Pr\big(X\ge 1\big)$, 
the so-called {\em second moment method} is usually applied.

\begin{lemma}[{\bf Cauchy Inequality}]
Let $X$ be a non-zero random variable. Then
$\Pr\big(X\ne 0\big)\ge\frac{ \E (X)^2}{ \E (X^2)}$.
\end{lemma}

\begin{proof}
Write $X=X\!\cdot\!{\bf 1}_{X\ne 0}$. By Cauchy-Schwartz inequality,
$$  
 \E (X)^2= \E (X\!\cdot\!{\bf 1}_{X\ne 0})^2
\le  \E (X^2)\!\cdot\! \E ({\bf 1}_{X\ne 0}^2)
= \E (X^2)\!\cdot\!\Pr(X\ne 0).
$$
So $\Pr\big(X\ne 0\big)\ge\frac{ \E (X)^2}{ \E (X^2)}$.
\end{proof}

\noindent
{\bf Cauchy Inequality}.  
{\it Let $X$ be a non-negative integer variable such that $X\ne 0$. Then
 $\Pr\big(X\ge 1\big)\ge \frac{ \E (X)^2}{ \E (X^2)}$.
}

\smallskip
The following is an useful inequality for our applications.

\begin{lemma}[{\cite[Theorem 6.10]{MU}}]\label{2'nd moment}
Let $X=\sum_{i=1}^n X_i$, where $X_i$'s are $0$-$1$ variables. Then
$$\textstyle
 \Pr\big(X\!\ge\! 1\big)\;\ge\;
 \sum_{i=1}^n \frac{ \E (X_i)}{ \E (X|X_i=1)}.
$$
\end{lemma}

\begin{proof}
Let $Y=\begin{cases}\frac{1}{X}, & X>0;\\ 0, & X=0.\end{cases}$
Then $YX={\bf 1}_{X\ge 1}$. So
$$
\Pr(X\ge 1)= \E (YX)=\sum_{i=1}^n \E (YX_i).
$$
By Total Probability Formula,
\begin{align*}
 \E (YX_i)
  &= \E (YX_i|X_i=0)\Pr(X_i=0) +  \E (YX_i|X_i=1)\Pr(X_i=1)\\
  &= \E (0|X_i=0)\Pr(X_i=0)+ \E (Y|X_i=1)\Pr(X_i=1)\\
  &= \E ({\textstyle\frac{1}{X}}|X_i=1)\Pr(X_i=1).
\end{align*}
Applying Jensen Inequality (see the following exercise) 
to the convex function $f(x)=\frac{1}{x}$, we have:
$$
 \E ({\textstyle\frac{1}{X}}|X_i=1)\ge
 \frac{1}{ \E (X|X_i=1)}.
$$
Thus
$$
   \E (YX_i) \ge \frac{1}{ \E (X|X_i=1)}\cdot\Pr(X_i=1)
  =\frac{ \E (X_i)}{ \E (X|X_i=1)}.
$$
Hence.
$\Pr(X\ge 1) \ge \sum_{i=1}^n\frac{ \E (X_i)}{ \E (X|X_i=1)}.$
\end{proof}

\bex[{\tt Jensen Inequality}] \label{Jensen}
 Let $f(x)$, $x\in [\alpha,\beta]$, be a convex function,
i.e., for any $x,x'\in [\alpha,\beta]$, $\frac{f(x)+f(x')}{2}\ge f(\frac{x+x'}{2})$.

(1)~ If $p_1,\cdots,p_n$ are non-negative real numbers such that
$p_1+\cdots+p_n=1$ and $x_1,\cdots, x_n\in [\alpha,\beta]$, then
$\textstyle
 \sum_{i=1}^n p_i f(x_i)\ge f\big(\sum_{i=1}^n p_i x_i\big).
$

(2)~ Let $X$ be a random variable taking values in $[\alpha,\beta]$. Then
$\textstyle
  \E \big(f(X)\big)\ge f\big( \E (X)\big).
$
\eex

\begin{definition}\label{d inf entropy}\rm
Let $X$ be a random variable taking values $x_1,\cdots,x_s$
(may be not real numbers)
such that $\sum_{i=1}^s p_i=1$, where $p_i=\Pr(X=x_i)$. 
Then $p_1,\cdots,p_s$ are called the {\em distribution} of $X$, 
 and 
\begin{equation*}\textstyle
H_{\gamma}(X)=\sum_{i=1}^s -p_i\log_{\gamma} p_i,
\end{equation*}
 is called the {\em information entropy} of $X$ with base $\gamma$
 ($\gamma>1$ is a real number).
\end{definition}

Let $Y$ be another random variable taking values $y_1,\cdots,y_t$
with distribution $p'_1,\cdots,p'_{t}$.
We have the {\em joint random variable} $(X,Y)$
taking values $\{(x_i,y_j)\,|\, 1\le i\le s, 1\le j\le t\}$
 with distribution $p_{i,j}=\Pr\big((X.Y)=(x_{i},y_j)\big)$, 
  $1\le i\le s, 1\le j\le t$. Then 
\begin{equation}\label{e inf entropy}\textstyle
\sum_{j=1}^t p_{i,j}=p_i, ~~ i=1,\cdots,s;\qquad
\sum_{i=1}^s p_{i,j}=p'_j, ~~ j=1,\cdots,t.
\end{equation}
$$\begin{array}{ccc|c}
p_{1,1} & \cdots & p_{1,t}& p_{1}\\
 \cdots & \cdots & \cdots & \vdots\\
p_{s,1} & \cdots & p_{s,t}& p_{s}\\ \hline
p'_{1} & \cdots & p'_{t}& 
\end{array}$$
Note that $p_{i,j}\ne p_ip'_j$ in general. 
If $p_{i,j}= p_ip'_j$, $\forall$ $1\le i\le s, 1\le j\le t$, 
then $X$ and $Y$ are said to be {\em randomly independent}.

For random variable $X_1,X_2,\cdots,X_k$,
the {\em joint random variable} $(X_1,\cdots,X_k)$ are defined similarly, 
and the following holds.

\begin{lemma}\label{l inf entropy}~
$H_{\gamma}(X_1,\cdots,X_k)\le H_{\gamma}(X_1)+\cdots+H_{\gamma}(X_k)$.
\end{lemma}

\begin{proof}
For $k=2$, see Exercise \ref{inequality of joint} below. For $k>2$,
prove it recursively.
\end{proof}

\bex\label{inequality of joint} 
Let $X,Y$ and the joint $(X,Y)$ be as in Eq.\eqref{e inf entropy}.
Define 
\begin{align*}
H_{\gamma}(Y|X)=\textstyle
 \sum_{i=1}^s\sum_{j=1}^t -p_{i,j}\log_{\gamma} \frac{p_{i,j}}{p_i};\\
 I_{\gamma}(X;Y)= \textstyle
  \sum_{i=1}^s\sum_{j=1}^t -p_{i,j}\log_{\gamma}\frac{p_ip'_j}{p_{i,j}}.
\end{align*} 
 
(1) $H_{\gamma}(Y|X)$ and $I_{\gamma}(X;Y)$ are non-negative. 
\hint{$\frac{p_{i,j}}{p_i}\le 1$ hence $-\log_{\gamma} \frac{p_{i,j}}{p_i}\ge 0$;
for $I_{\gamma}(X;Y)$, apply Exercise \ref{Jensen} 
to the convex function $-\log_{\gamma} x$.}

(2) $I_{\gamma}(X;Y)=H_{\gamma}(X)-H_{\gamma}(X|Y)
  =H_{\gamma}(Y)-H_{\gamma}(Y|X)$. 
 \hint{compute it directly.}

(3) $H_{\gamma}(X,Y)=H_{\gamma}(X)+H_{\gamma}(Y|X)
  =H_{\gamma}(Y)+H_{\gamma}(X|Y)$.

(4) $H_{\gamma}(X,Y)\le H_{\gamma}(X)+H_{\gamma}(Y)$.
\eex

\vskip3mm
Yun Fan\par
School of Mathematics and Statistics\par
Central China Normal University, Wuhan 430079, China\par
{\it Email address}: yfan@ccnu.edu.cn
\par\vskip3mm
Liren Lin\par
School of Mathematics and Statistics\par
Hubei University, Wuhan 430062, China\par
{\it Email address}: l\_r\_lin86@163.com

\end{document}